\documentclass[a4paper,12pt,reqno]{amsart}
\usepackage[utf8]{inputenc}
\usepackage[T1]{fontenc}
\usepackage[british]{babel}
\usepackage{times}

\usepackage[text={16.0cm,25.2cm}, centering, nomarginpar, includeheadfoot, ignoremp]{geometry}
\usepackage{amsmath, amssymb, amsfonts, dsfont, amsthm, mathrsfs, bm, mathtools, eqnarray, faktor}
\usepackage{scalerel, stackengine}
\usepackage{hyperref}
\usepackage{xparse}
\usepackage{listings,enumitem}
\usepackage[hang,small,bf]{caption}
\usepackage{fancyhdr}
\usepackage{indentfirst}

\usepackage[toc,page]{appendix}
\theoremstyle{definition}

\theoremstyle{plain}
\newtheorem{note}{Remark}
\newtheorem{theo}{Theorem}[section]
\newtheorem{lemma}[theo]{Lemma}
\newtheorem{prop}[theo]{Proposition}

\numberwithin{equation}{section}
\numberwithin{defn}{section}
\numberwithin{note}{section}

\stackMath

\let\oldker\ker
\let\oldforall\forall
\let\oldexists\exists
\renewcommand\widehat[1]{%
\savestack{\tmpbox}{\stretchto{%
    \scaleto{%
        \scalerel*[\widthof{\ensuremath{#1}}]{\kern.1pt\mathchar"0362\kern.1pt}%
        {\rule{0ex}{\textheight}}
    }{\textheight}%
}{2.4ex}}%
\stackon[-6.9pt]{#1}{\tmpbox}%
}

\NewDocumentCommand \fromToParser {m m m}{%
    \IfValueTF{#3}{%
        \fromTo{#1}{#2}[#3]%
    }{%
        \fromTo{#1}{#2}%
    }%
}
\NewDocumentCommand \fromTo {m m o}{%
    #1\! \IfValueTF{#3}{%
            \to[#3]
        }{%
            \to[\!\quad\!]
        } \!#2%
}
\newcommand{\integrand}[2]{\!#2\; #1}
\NewDocumentCommand \integralLimits {m m}{%
    _{#1}\IfValueT{#2}{^{\,#2}\!}%
}
\newcommand{\defineSymbol}[2]{#1 \vcentcolon = #2}
\newcommand{\enifedSymbol}[2]{#1 = \vcentcolon #2}

\renewcommand{\newline}{\hfill\\}
\renewcommand{\-}{\mspace{-1.5mu}}
\newcommand{\+}{\mspace{1.5mu}}
\newcommand{\nquad}{\mspace{-18mu}}

\newcommand{\n}{\noindent}
\newcommand{\vs}{\vspace{0.5cm}}

\renewcommand{\Im}{\mathrm{Im}}
\renewcommand{\Re}{\mathrm{Re}}
\newcommand{\pInfty}{ {\scriptstyle +}\+\infty}
\newcommand{\mInfty}{ {\scriptstyle -}\+\infty}
\newcommand{\Id}{\mathds{1}}
\renewcommand{\forall}{\oldforall\,}
\RenewDocumentCommand \exists {s}{%
    \IfBooleanTF{#1}{\oldexists!}{\oldexists} \;%
}
\ProvideDocumentCommand \define {s >{\SplitArgument{1}{;}} m}{%
    \IfBooleanTF{#1}{\enifedSymbol #2}{\defineSymbol #2}%
}
\RenewDocumentCommand \to {o}{%
    \IfValueTF{#1}{%
        \xrightarrow[\,#1\,]{\;}%
    }{%
        \,\rightarrow\,%
    }%
}
\NewDocumentCommand \maps {m >{\SplitArgument{2}{;}} m}{%
    #1\!: \fromToParser #2%
}
\ProvideDocumentCommand \conjugate {s m}{%
    \IfBooleanTF{#1}{%
        \overline{#2}%
    }{%
        \bar{#2}%
    }%
}
\providecommand{\abs}[1]{\lvert#1\rvert}
\newcommand*{\vect}[1]{\boldsymbol{#1}}
\newcommand*{\vers}[1]{\hat{\vect{#1}}}
\NewDocumentCommand \ball {s m m}{%
    \IfBooleanTF{#1}{%
        \mathcal{B}^{\+ c}_{#2}(#3)%
    }{%
        \mathcal{B}_{#2}(#3)%
    }%
}
\NewDocumentCommand \Char { m o }{
    \Id_{#1}
    \IfValueT{#2}{( #2 )}
}
\newcommand*{\oSmall}[1]{o\!\left(#1\right)}
\newcommand*{\oBig}[1]{\mathcal{O}\!\left(#1\right)}
\NewDocumentCommand \integrate { >{\SplitArgument{1}{;}} o >{\SplitArgument{1}{;}} m}{%
    \int \IfValueT{#1}{\integralLimits #1 \!} \integrand #2%
}

\NewDocumentCommand \hilbert{s}{%
    \IfBooleanTF{#1}{\mathscr{H}}{\mathfrak{H}}%
}
\NewDocumentCommand \X{s}{%
    \IfBooleanTF{#1}{\mathcal{X}}{\mathfrak{X}}%
}
\NewDocumentCommand \scalar { m m o }{%
    \langle#1,\,#2\rangle\IfValueT{#3}{_{#3}}%
}
\NewDocumentCommand \norm { m o }{%
    \left\lVert#1\right\rVert \IfValueT{#2}{_{#2}}%
}
\NewDocumentCommand \normConverge {o m m o}{
    \IfValueTF{#4}{
        \fromTo{ \IfValueTF{#1}{\norm{#2 - #3}[#1]\! }{\norm{#2 - #3}} }{\,0}[#4]
    }{
        \fromTo{ \IfValueTF{#1}{\norm{#2 - #3}[#1]\! }{\norm{#2 - #3}} }{\,0}
    }
}
\NewDocumentCommand \weakConverge {m m o}{%
    #1 \!\IfValueTF{#3}{%
            \xrightharpoonup[#3]{\quad}
        }{%
            \xrightharpoonup{\quad}%
        } \- #2
}
\newcommand*{\dom}[1]{\mathscr{D}(#1)}
\newcommand*{\ran}[1]{\mathrm{ran}(#1)}
\renewcommand*{\ker}[1]{\oldker(#1)}
\newcommand*{\adj}[1]{{#1}^{\ast}}

\newcommand*{\bounded}[1]{\mathscr{B}\mspace{-1mu}\left(#1\right)}
\NewDocumentCommand \resolvent {m o}{%
    \mathcal{R}_{#1}\IfValueT{#2}{(#2)}%
}
\NewDocumentCommand \spectrum {o m}{%
    \IfValueTF{#1}{%
        \sigma_{\mathrm{#1}}(#2)%
    }{%
        \sigma(#2)%
    }%
}

\NewDocumentCommand \Lp {s m o} { L^{\- #2}\IfBooleanT{#1}{_{\,\mathrm{loc}}}\IfValueT{#3}{(#3)} }
\NewDocumentCommand \LpS {s m o}{
\IfBooleanTF{#1}{\Lp{#2}_{+}}{\Lp{#2}_{\mathrm{sym}}}\IfValueT{#3}{(#3)}
}
\NewDocumentCommand \LpA {s m o}{
\IfBooleanTF{#1}{\Lp{#2}_{-}}{\Lp{#2}_{\mathrm{asym}}}\IfValueT{#3}{(#3)}
}
\NewDocumentCommand \LpSA {m o}{
\Lp{#1}_{\pm}\IfValueT{#2}{(#2)}
}

\NewDocumentCommand \FT { s m o }{
    \IfBooleanTF{#1}{ \IfValueTF{#3}{ (\mathcal{F} \+  #2)\-(#3) }{%
    \mathcal{F} \+  #2 } }{ \hat{#2}\IfValueT{#3}{(#3)} }
}

\newcommand{\C}{\mathbb{C}}
\newcommand{\R}{\mathbb{R}}
\newcommand{\Rplus}{\mathbb{R}_{+}}
\newcommand{\N}{\mathbb{N}}
\newcommand{\No}{\mathbb{N}_0}
\NewDocumentCommand \Z {o}{%
    \IfValueTF{#1}{%
        \mathbb{Z}_{\,#1}%
    }{%
        \mathbb{Z}%
    }%
}

\newcommand{\be}{\begin{equation}}
\newcommand{\ee}{\end{equation}}

\title[Regularized zero-range Hamiltonian for a Bose gas  with an impurity]{Regularized zero-range Hamiltonian \\for a Bose gas  with an impurity}

\author[D. Ferretti]{Daniele Ferretti}
\address[D. Ferretti]{Department of Mathematics G. Castelnuovo, University of Rome ``La Sapienza'', Piazzale  Aldo Moro, 5, 00185 Rome, Italy}
\email{d.ferretti@uniroma1.it}

\author[A. Teta]{Alessandro Teta}
\address[A. Teta]{Department of Mathematics G. Castelnuovo, University of Rome ``La Sapienza'', Piazzale  Aldo Moro, 5, 00185 Rome, Italy}
\email{teta@mat.uniroma1.it}
\date{}

\thanks{The authors acknowledge the support of the GNFM Gruppo Nazionale per la Fisica Matematica - INdAM}

\begin{document}
\pagenumbering{Alph}
\makeatletter
\renewcommand\subsection{%
  \@startsection{subsection}%
    {2}
    {0em}
    {-1ex \@plus 0.1ex \@minus -0.05ex}
    {-1em \@plus 0.2em}
    {\scshape}
  }

\makeatother

\begin{abstract}
We study the Hamiltonian for a system of $N$ identical bosons  interacting with an impurity, i.e., a different particle, via zero-range forces in dimension three. It is well known that, following the standard approach, one obtains the Ter-Martirosyan Skornyakov Hamiltonian which is unbounded from below. In order to avoid such instability problem, we introduce a three-body force acting at short distances. The effect of this force is to  reduce to zero the strength of the zero-range interaction between two particles, i.e., the impurity and a boson, when another boson approaches the common position of the first two particles. We show that the Hamiltonian defined with such regularized interaction is self-adjoint and bounded from below if the strength of the three-body force is sufficiently large. The method of the proof is based on a careful analysis of the corresponding  quadratic form.
\end{abstract}
\maketitle

\begin{footnotesize}
\emph{Keywords: Zero-range interactions; Many-body Hamiltonians; Schr\"odinger operators.} 
 
\emph{MSC 2020: 
    81Q10; 
    81Q15; 
    70F07; 
    46N50; 
    81V70;
}  
\end{footnotesize}



\vs\vs

\pagenumbering{arabic}
\section{Introduction}\label{intro}
\vs

Hamiltonians with zero-range interactions are often used in Quantum Mechanics as toy models to describe the low energy behavior of a particle system.
The advantage is that zero-range interactions are structurally simple, allow in many cases to make explicit computations and, at least formally, are characterized by a single physical parameter known as two-body scattering length.
The mathematical construction of such Hamiltonians as self-adjoint (s.a.) and, possibly, lower bounded operators in the appropriate Hilbert space requires some care.
In the one-body case a complete theory is available (\cite{Albeverio}), while in the $n$-body case the situation depends on the space dimensions.
In dimension one perturbation theory applies and the model is well understood (see, e.g., \cite{BCFT1}, \cite{GHL} for recent contributions).
In dimension two, using the same type of renormalization used in the one-body case, the s.a. and bounded from below Hamiltonian can be constructed (\cite{DFT}, \cite{DR}) and analysed in detail (see, e.g., \cite{GH}).
On the other hand, in dimension three for $n=3$ it is known  (\cite{MF}, see also \cite{MF2}) that the same procedure leads to a symmetric but non s.a. operator and that all its s.a. extensions are unbounded from below.
Such instability property is known as Thomas effect and it is due to the fact that the interaction becomes too singular when all the three particles are close to each other (see also the recent papers \cite{Miche}, \cite{GM} and the references therein).
The Hamiltonian defined in \cite{MF} is therefore unsatisfactory from the physical point of view and the construction of other physically reasonable $n$-body Hamiltonians in dimension three with zero-range interactions can be considered as an open problem (we just mention that the situation is rather different for systems made of two species of fermions, see, e.g., \cite {CDFMT}, \cite{CDFMT1}, \cite{MS}, \cite{MS1}). 

\n Following a suggestion of Minlos and Faddeev contained in \cite{MF}, it has been recently proposed (\cite{FT}, \cite[section 9]{Miche}, \cite{BCFT}, \cite[section 6]{GM}) a regularized version of the Hamiltonian for a system of three bosons.
Roughly speaking, the idea is to introduce an effective scattering length which decreases to zero when the position of two particles coincides and the third particle is close to the common position of the first two.
In the other cases the effective scattering length remains constant.
In this sense, one introduces a three-body interaction that reduces to zero the strength of the interaction between two particles (only) when the third particle approaches the common position of the first two.  

\n In this paper we exploit the same idea to construct the Hamiltonian in dimension three for a gas of bosons interacting with an impurity.  
More precisely, we consider a quantum system of $N$ identical spinless bosons of mass $\frac{1}{2}$ and we assume that the bosons interact only with an impurity, i.e. a different particle of mass $\frac{M}{2}$, via a zero-range, two-body interaction.
Let us denote by
\begin{equation}
    \define{\hilbert*_{N+1};\Lp{2}[\R^3]\-\otimes\-\LpS{2}[\R^{3N}]}\subset\Lp{2}[\R^{3(N+1)}],\quad N\-\geq 2
\end{equation}
the Hilbert space of the system.
At a formal level, the Hamiltonian reads
\begin{equation}\label{formalH}
    \hat{\mathcal{H}}=\mathcal{H}_0+\nu\sum_{i=1}^N\delta(\vect{x}_i-\vect{x}_0),
\end{equation}
where $\nu$ is a coupling constant and $\mathcal{H}_0$ is the free Hamiltonian, given by
\begin{equation}\label{freeH}
    \begin{split}
    &\dom{\mathcal{H}_0}=\hilbert*_{N+1}\cap H^2(\R^{3(N+1)}), \quad\quad\define{\mathcal{H}_0;-\frac{1}{M}\Delta_{\vect{x}_0}-\sum_{i=1}^N\Delta_{\vect{x}_i}}\+.
    \end{split}
\end{equation}
We want to define a rigorous counterpart of the formal operator~\eqref{formalH} as a s.a. and bounded from below operator $\mathcal{H}$ in $\hilbert*_{N+1}\+$.
By definition, such an operator must be a proper singular perturbation of $\mathcal{H}_0\+$ supported on the coincidence hyperplanes
\begin{equation}\label{incidenceHyperplanes}
    \define{\pi;\bigcup_{i=1}^N\+\pi_i}\,, \qquad \define{\pi_i; \left\{(\vect{x}_0, \vect{x}_1,\ldots,\vect{x}_N)\in\R^{3(N+1)} \,\big| \; \vect{x}_i\-=\vect{x}_0\right\}}\-.
\end{equation}
This means that $\mathcal{H}$ must satisfies the property
\be\label{h=h0}
 \mathcal{H}\psi = \mathcal{H}_0\+\psi\;\;\;\; \forall \psi \in \dom{\mathcal{H}_0}\;\;\; \text{s.t.}\;\;\psi |_{\pi} = 0 \,.
 \ee
 
 \n
Motivated by this observation, we define the operator
\begin{equation}\label{symmetricHamiltonianToBeExtended}
    \define{\dot{\mathcal{H}}_0; \mathcal{H}_0|_{\dom{\dot{\mathcal{H}}_0}}}\+, \qquad \define{\dom{\dot{\mathcal{H}}_0};\hilbert*_{N+1}\cap H^2_0(\R^{3(N+1)}\setminus \pi)} \end{equation}
which is symmetric and closed according to the graph norm of $\mathcal{H}_0\+$.
Our goal is to find the Hamiltonian $\mathcal{H}$ as a s.a. and bounded from below extension of $\dot{\mathcal{H}}_0\+$.

\n A typical class of extensions is obtained by requiring that an element $\psi$ of the domain of $\mathcal H$ satisfies the following boundary condition on each hyperplane $\pi_i$
\begin{equation}\label{stmBC}
    \begin{split}
    \psi(\vect{x}_0,\vect{x}_1,\ldots,\vect{x}_N)=&\;\frac{\xi\!\left(\frac{\vect{x}_i+M\vect{x}_0}{M+1},\vect{x}_1,\ldots\check{\vect{x}}_i\ldots,\vect{x}_N\right)}{\abs{\vect{x}_i-\vect{x}_0}}\,+\\
    &+\-\alpha\+\xi(\vect{x}_0,\vect{x}_1,\ldots\check{\vect{x}}_i\ldots,\vect{x}_N)\-+\oSmall{1}\-,\quad \text{for }\,\vect{x}_i\to\vect{x}_0\+,
    \end{split}
\end{equation}
for some $\xi\in\hilbert*_N$, where $\alpha$ is a real parameter and the notation $\check{\vect{x}}_i$ indicates the omission of the variable $\vect{x}_i\+$.
The boundary condition~\eqref{stmBC} is a natural generalization of the boundary condition satisfied in the one-body case (\cite{Albeverio}). 
The parameter $\alpha$ is related to the two-body scattering length $\mathfrak{a}$ between the impurity and a boson via the relation
\begin{equation}\label{2BodyScatteringLenght}
    \mathfrak{a} = -\frac{1}{\alpha}\,.
\end{equation}
Notice that, by~\eqref{2BodyScatteringLenght}, the strength of the point interaction between the impurity and a boson goes to zero as $\abs{\alpha} \!\to\- \pInfty$. 

\n The s.a. extensions obtained requiring the boundary condition~\eqref{stmBC} lead to the Hamiltonian unbounded from below studied in~\cite{MF}.
As already mentioned, in order to obtain an energetically stable system, we  introduce a suitable regularization in~\eqref{stmBC}.
More precisely, we replace the parameter $\alpha$ by a new, position dependent, coupling constant on each coincidence plane $\pi_i\+$
\begin{equation*}
    \alpha\longmapsto\beta_i\+,
\end{equation*}
where the function $\maps{\beta_i}{\pi_i;\R}$ is given by
\begin{equation}\label{runningCoupling}
    \beta_i:(\vect{x}_0,\vect{x}_1,\ldots\check{\vect{x}}_i\ldots,\vect{x}_N)\longmapsto \alpha+\gamma\mspace{-9mu}\sum_{\substack{1\+\leq \,j\,\leq\+ N\\ j\+\neq\+ i}}\mspace{-9mu}\frac{\theta(\abs{\vect{x}_j-\vect{x}_0})}{\abs{\vect{x}_j-\vect{x}_0}},
\end{equation}
with $\gamma>0$ and $\maps{\theta}{\Rplus;\R}$ a measurable function satisfying
\begin{subequations}\label{assumptionsTheta}
\begin{gather}
    \label{compactSuppTheta}\mathrm{supp}\,\theta \text{ is a compact},\\
    \label{positiveBoundedCondition}
    1-\tfrac{r}{b}\leq \theta(r)\leq 1+\tfrac{r}{b}\+,\quad\text{for some }b>0.
\end{gather}
\end{subequations}
We observe that the function $\theta$, by assumption~\eqref{positiveBoundedCondition}, is positive in a neighborhood of the origin and it is continuous at zero, with $\theta(0)=1$.
We also stress that the simplest choice for the function $\theta$ is the characteristic function $\Char{b}$ of the ball of radius $b$ centered in the origin.

\n With the above replacement, we define the modified  boundary condition 
\begin{equation}\label{mfBC}
\begin{split}
    \psi(\vect{x}_0,\vect{x}_1,\ldots,\vect{x}_N)=&\;\frac{\xi\!\left(\frac{\vect{x}_i+M\vect{x}_0}{M+1},\vect{x}_1,\ldots\check{\vect{x}}_i\ldots,\vect{x}_N\right)}{\abs{\vect{x}_i-\vect{x}_0}}\,+\\
    &+\-(\Gamma_{\!\mathrm{reg}}^i \+\xi)(\vect{x}_0,\vect{x}_1,\ldots\check{\vect{x}}_i\ldots,\vect{x}_N)\-+\oSmall{1}\-,\quad \text{for }\,\vect{x}_i\to\vect{x}_0\+,
\end{split}
\end{equation}
where $\Gamma^i_{\!\mathrm{reg}}$ acts as follows
\begin{equation}\label{regGamma}
    \Gamma_{\!\mathrm{reg}}^i: \xi\longmapsto \beta_i(\vect{x}_0,\vect{x}_1,\ldots\check{\vect{x}}_i\ldots,\vect{x}_N)\+\xi(\vect{x}_0,\vect{x}_1,\ldots\check{\vect{x}}_i\ldots,\vect{x}_N).
\end{equation}
Notice that the function $\beta_i$ is symmetric under the exchange of any couple $\vect{x}_k\-\longleftrightarrow\-\vect{x}_\ell$ with $\ell\-\neq k\-\in\-\{1,\ldots,N\}\-\smallsetminus\-\{i\}$.
In analogy with~\eqref{stmBC}, the boundary condition~\eqref{mfBC} characterizes the point interaction between the impurity and the $i$-th boson. 
The function $\beta_i$ diverges if $\vect{x}_j \!\to\-\vect{x}_0\+$, for any $j \neq i$ and this means that the strength of the point interaction between the impurity and the $i$-th boson decreases to zero when a third particle, in our case another boson, approaches the common position of the first two particles. 
In other words, as already pointed out, we are introducing a three-body interaction meant to regularize the ultraviolet singular behavior occurring when the positions of more than two particles coincide.
We also stress that, by~\eqref{compactSuppTheta}, $\theta$ is chosen to be compactly supported and therefore the usual two-body point interaction between the impurity and the $i$-th boson is restored when the other particles are far enough.

\n The aim of this paper is to show that the modified boundary condition \eqref{mfBC}  allows to give a rigorous construction of a s.a. and bounded from below Hamiltonian $\mathcal{H}$. 

%


\n
The approach is based on the theory of quadratic forms.
More precisely, by a heuristic procedure based on the conditions \eqref{h=h0}, \eqref{mfBC}, we  arrive at the definition of a quadratic form in $\hilbert*_{N+1}$ (see~\eqref{QF}) which is the starting point of the rigorous analysis.
Our main result is the proof that for any $\gamma$ larger than a threshold value $\gamma_c\+$, the quadratic form is closed and bounded from below.
The threshold value is explicitly given (see~\eqref{criticalGamma}) and it is uniformly bounded in $N$ and $M$.
Furthermore, we characterize the s.a. and bounded from below operator $\mathcal{H}$ uniquely defined by the quadratic form.
Such operator, by definition, is our Hamiltonian for the boson gas interacting with an impurity via regularized zero-range interactions.

\vs
\n For the convenience of the reader, we collect here some of the notation used in the paper. 


\n - Given the Euclidean space $(\R^n,\+\cdot\+)$, $\vect{x}$ is a vector in $\R^n$ and $x=\abs{\vect{x}}$.

\n - $\mathcal{S}(\R^n)$ denotes the space of Schwartz functions.

\n
-  $\FT{\psi}$ is the  Fourier transform of $\psi$.

\n - For any $p\geq \-1$ and $\Omega$ open set in $\R^n$, $\Lp{p}(\Omega,\mu)$ is the Banach space of $p\+$-integrable functions with respect to the Borel measure $\mu$.
We  use $\Lp{p}(\Omega)$ in case $\mu$ is the Lebesgue measure and we denote  $\define{\norm{\cdot}[p]\-; \-\norm{\cdot}[\Lp{\+p}[\R^n]]}$.

\n - If $\hilbert$ is a complex Hilbert space, we denote by $\scalar{\cdot\+}{\!\cdot}_{\hilbert}$,  $\define{\norm{\+\cdot\+}[\hilbert]\!; \!\sqrt{\scalar{\cdot\+}{\!\cdot}_{\hilbert}} }\;$ the inner product and the induced norm.

%
%
 \n
 -   if $\hilbert=\Lp{2}(\R^n)$, we simply denote by $\scalar{\cdot\+}{\!\cdot}$, $\|\cdot\|$ the inner product and the norm.
 

\n - $H^s(\R^n)$ is the standard  Sobolev space of order $s>0$ in $\R^n$.

\n -$f|_{\pi_i}\in H^s(\R^{3N})$ is the trace of $f\in H^{s+\frac{3}{2}}(\R^{3(N+1)})$ on the hyperplane $\pi_i\+$.

\n - $\bounded{X,Y}$ is the Banach space of the linear bounded operators from $X$ to $Y$, where $X$ and $Y$ are Hilbert spaces,  and $\define{\bounded{X};\bounded{X,X}}$.


\vs\vs

\section{Main results and strategy of the proof}\label{mainResults}
\vs

In this section we introduce some definitions and then we formulate our main results.

\n
Let us define the bounded operator $\maps{G^\lambda}{\hilbert*_N;\Lp{2}(\R^{3(N+1)})}$ whose Fourier representation is given by
\begin{equation}\label{potentialDef}
    \define{(\+\widehat{G^\lambda\xi}\+)(\vect{p},\vect{k}_1,\ldots,\vect{k}_N); \frac{1}{\mu}\,\sqrt{\frac{2}{\pi}}\,\frac{\sum_{j=1}^N\hat{\xi}(\vect{p}+\vect{k}_j,\vect{k}_1,\ldots\check{\vect{k}}_j\ldots,\vect{k}_N)}{\frac{1}{M} p^2+\sum_{m=1}^N k_m^2+\lambda}}\,,
\end{equation}
where $\lambda>0$ and 
\begin{equation}
    \define{\mu;\frac{M}{M+1}}
\end{equation}
denotes the reduced mass of the two-particle subsystem composed by a boson and the impurity.
We shall refer to $G^{\lambda} \xi$ as the potential produced by the charge $\xi$ distributed on $\pi$.
A more detailed discussion on the properties of the potential is postponed to the appendix (section~\ref{computingPotential}).
Here we only mention that $G^\lambda$ is injective, $\ran{G^{\lambda}}\-\subset\- \hilbert*_{N+1}$ and $G^{\lambda} \xi \-\notin\! H^{1}(\R^{3(N+1)})$ (see remarks~\ref{potentialPropertiesInheritance},~\ref{wellPosednessQFDomain}).

\n
Then, let us define the following hermitian quadratic form in $\Lp{2}[\R^{3N}]$
\begin{equation}\label{defPhi}
\dom{\Phi^{\lambda}} = H^\frac{1}{2}(\R^{3N}), \qquad   \define{\Phi^\lambda;\Phi^\lambda_{\mathrm{diag}}\!+\Phi^\lambda_{\mathrm{off}}\-+\Phi_{\mathrm{reg}}}\- + \Phi_0\+,
\end{equation}
where
\begin{subequations}\label{componentsPhi}
\begin{gather}
    \define{\Phi^\lambda_{\mathrm{diag}}[\xi]\-;\tfrac{4\pi N}{\sqrt{\mu\,}\+}\!\- 
    \integrate[\R^{3N}]{\!\textstyle{\sqrt{\frac{p^2}{M+1}\- +\!\sum_{m=1}^{N-1} k_m^2+\lambda\+}\,} \abs{\FT{\xi}(\vect{p},\vect{k}_1,\ldots,\vect{k}_{N-1})}^2; 
    \mspace{-10mu}d\vect{p}\+d\vect{k}_1\-\cdots d\vect{k}_{N-1}}}\+,\label{diagPhi}\\
    \define{\Phi^\lambda_{\mathrm{off}}[\xi]\-;\! -\tfrac{2N(N-1)\!}{\pi\,\mu^2}\!\! \integrate[\R^{3(N+1)}]{\-\frac{\conjugate*{\FT{\xi}(\vect{p}\-+\-\vect{k}_1, \vect{k}_2,\ldots, \vect{k}_N)\-}\,\FT{\xi}(\vect{p}\-+\-\vect{k}_2, \vect{k}_1, \vect{k}_3,\ldots, \vect{k}_N)}{\frac{1}{M}p^2+\sum_{m=1}^N k_m^2+\lambda};\mspace{-33mu}d\vect{p}\+d\vect{k}_1\- \cdots d\vect{k}_N}},\-\label{offPhi}\\
    \define{\Phi_{\mathrm{reg}}[\xi];\tfrac{4\pi N(N-1)\,\gamma}{\mu}\!\-\integrate[\R^{3N}]{\frac{\abs{\xi(\vect{y},\vect{x}_1,\ldots,\vect{x}_{N-1})}^2\!\-}{\abs{\vect{y}\--\vect{x}_1}};\mspace{-10mu}d\vect{y} d\vect{x}_1\-\cdots d\vect{x}_{N-1}}},\label{regPhi}\\
    \Phi_0[\xi]:= \tfrac{4\pi N}{\mu}\!\-\integrate[\R^{3N}]{\tilde{\alpha}(\abs{\vect{y}\--\vect{x}_1})\+\abs{\xi(\vect{y},\vect{x}_1,\ldots,\vect{x}_{N-1})}^2;\mspace{-10mu}d\vect{y} d\vect{x}_1\-\cdots d\vect{x}_{N-1}}\label{notePhi}
    \end{gather}
\end{subequations}
with
\begin{equation}\label{actualAlpha}
\tilde{\alpha}: r\longmapsto\alpha+(N\!-\!1)\+\gamma\,\frac{\theta(r)-1}{r}\+ ,\qquad r>0\+. 
\end{equation}
Notice that assumption~\eqref{positiveBoundedCondition} implies   $\tilde{\alpha}\in\Lp{\infty}[\Rplus]$ and therefore the quadratic form $\Phi_0$ is bounded in $\Lp{2}[\R^{3N}]$. 
We are now in position to define the main object of our analysis, i.e., the  quadratic form in $\hilbert*_{N+1}$ given by 
\begin{equation}\label{QF}\begin{split}
    &\define{\dom{Q}\-; \!\left\{\psi \in\hilbert*_{N+1} \,\big |\;\psi=w^\lambda\!+G^\lambda\xi\+,\: w^\lambda\-\in\- H^1(\R^{3(N+1)}),\: \xi \-\in\- \hilbert*_N\cap H^{\frac{1}{2}}(\R^{3N})\right\}}\-,\\
    &\mspace{7.5mu}\define{Q[\psi];\mathscr{F}_\lambda[w^\lambda]-\lambda\- \norm{\psi}^2\-+\Phi^\lambda[\xi]}
\end{split}
\end{equation}
where
\begin{equation}\label{trivialQF}
\begin{split}
    &\maps{\mathscr{F}_\lambda}{H^1(\R^{3(N+1)}); \R_+},\\
    &\qquad\varphi \longmapsto \scalar{(\mathcal{H}_0\-+\-\lambda)^{\frac{1}{2}}\+\varphi}{(\mathcal{H}_0\-+\-\lambda)^{\frac{1}{2}}\+\varphi}.
\end{split}
\end{equation}
The heuristic motivation leading to the above definition is discussed in the appendix  (section~\ref{actuallyBuildingQF}).  
We observe that $\dom{Q}$ is an extension of the form domain of $\mathcal{H}_0\+$, since $H^1(\R^{3(N+1)})\cap\hilbert*_{N+1}$ is a proper subset of $\dom{Q}$ and \begin{equation}
    Q[\psi]=\scalar{\mathcal{H}_0^{\frac{1}{2}}\+\psi}{\mathcal{H}_0^{\frac{1}{2}}\+\psi},\qquad \text{for }\,\psi\in H^1(\R^{3(N+1)})\cap\hilbert*_{N+1}\+.
\end{equation}
This is due to the injectivity of $G^\lambda$ that implies $\psi\in \dom{Q} \cap H^1(\R^{3(N+1)})$ if and only if $\xi\equiv 0$.

\n Moreover, for any fixed $M>0$ and $N\geq 2$, we introduce the critical parameter
\begin{equation}\label{criticalGamma}
    \define{\gamma_c;\frac{2(M\-+\-1)}{\pi}\arcsin\!\left(\tfrac{1}{M+1}\right)-\frac{2\sqrt{M(M\!+\-2)}}{\pi(N\!-\!1)(M\-+\-1)}\+.}
\end{equation}
It is easy to see that $\gamma_c$ is positive and
\begin{align*}
    \inf_{M>\+0}\gamma_c=\frac{2}{\pi}\frac{N\--\-2}{N\--\-1}, && \sup_{M>\+0}\gamma_c=1.
\end{align*}
We stress that our main results hold for any $\gamma >\gamma_c\+$.
We also notice that $\gamma_c$ is uniformly bounded in $N\-\geq\- 2$ and $M\->\-0$. In the special case $M=1$, $N=2$ we have $\gamma_c= \frac{2}{3} - \frac{\sqrt{3}}{\pi} \simeq 0.115$.
It is worth to observe that in the case of three interacting bosons (discussed in~\cite{BCFT}), a larger critical value $\gamma_c^{\mathrm{3 bos}}\! =  \frac{4}{3} - \frac{\sqrt{3}}{\pi} \simeq 0.782$ is found.
The difference is due to the fact that, in our case, the two bosons are non interacting and therefore the singular negative contribution to be compensated, given by~\eqref{offPhi}, is smaller by a factor $2$.

\n Our first  result concerns the quadratic form $\Phi^{\lambda}$ and it  is formulated in the next proposition.
\begin{prop}\label{coercivityLemma}
\newline
    i)\quad For any $\gamma>0$ and $\lambda>0$ one has
    \begin{equation}
        \Phi^{\lambda}[\xi] \,\leq\, C_1  \, \Phi^\lambda_{\mathrm{diag}} [\xi]\+, \qquad \xi \in H^{\frac{1}{2}}(\R^{3N}) 
    \end{equation}
    where $C_1$ is a positive constant.
    
    \n ii)\quad Let us assume $\gamma>\gamma_c\+$.
    Then, there exists $\lambda_0\->0\+$ s.t. for any $\lambda>\-\lambda_0$ one has
    \begin{equation}\label{Philo}
        C_2\, \Phi^\lambda_{\mathrm{diag}}[\xi] \, \leq \, \Phi^{\lambda}[\xi]\+, \qquad \xi \in H^{\frac{1}{2}}(\R^{3N}) 
\end{equation}
where $C_2$ is a positive constant.
In particular, the quadratic form $\Phi^{\lambda}$, $\dom{\Phi^{\lambda}}$ in $\Lp{2}[\R^{3N}]$ is closed and bounded from below by a positive constant.
\end{prop}
\n Proposition~\ref{coercivityLemma} implies that $\Phi^\lambda$ uniquely defines a s.a. and invertible operator $\Gamma^\lambda$ in $\Lp{2}[\R^{3N}]$ for any $\lambda\->\-\lambda_0\+$, as long as $\gamma>\gamma_c\+$.
Moreover, we shall see that the domain $D$ of $\Gamma^\lambda$ is independent of $\lambda\->\-\lambda_0\+$.
We shall also extend the definition of $\Gamma^\lambda,D$ to all $\lambda\-\in\--\rho(\mathcal{H}_0)$, preserving its invertibility for any $\lambda\in\C\smallsetminus(\mInfty,\lambda_0]$.

\n Using the above proposition, we can prove our main results that are summarized in the following two theorems.
In the first one, we characterize the quadratic form $Q$.
\begin{theo}\label{closedBoundedQF}\newline
Let us assume $\gamma>\gamma_c\+$.
Then, the quadratic form $Q$, $\dom{Q}$ in $\hilbert*_{N+1}$ is closed and bounded from below.
In particular, $Q>-\lambda_0\+$.
\end{theo}
\n In the second theorem we characterize the Hamiltonian $\mathcal{H}$ defined by the quadratic form $Q$.
\begin{theo}\label{hamiltonianCharacterization}$ $

\n
Let us assume $\gamma>\gamma_c\+$.
Then, the quadratic form $Q$, $\dom{Q}$ uniquely defines the self-adjoint and bounded from below operator $\mathcal{H}$, $\dom{\mathcal{H}}$ characterized as follows
\begin{equation}\label{hamiltonianAction}
    \begin{split}
    \dom{\mathcal{H}}&=\!\left\{\psi\in\dom{Q}\,\big|\; w^\lambda\!\in\-  H^2(\R^{3(N+1)}),\; \xi\- \in\- D,\,\Gamma^\lambda\xi\- =  \tfrac{4\pi}{\mu}\+\textstyle{\sum_{i=1}^N} w^\lambda\big|_{\pi_i}\+,\,\lambda>\!\lambda_0\right\}\-,\\
    \mathcal{H}\psi&=\mathcal{H}_0\+w^\lambda\--\lambda\+  G^\lambda \xi.
\end{split}
\end{equation}
Moreover, the resolvent $\define{\resolvent{\mathcal{H}}[z]; (\mathcal{H} -z)^{-1}}$ is given by 
\begin{equation}\label{resolventHamiltonian}
    \resolvent{\mathcal{H}}[z]\psi=\resolvent{\mathcal{H}_0}[z]\psi+G^{-z}\xi,\qquad \forall z\in\C\smallsetminus[-\lambda_0,\pInfty),
\end{equation}
where $\psi\in\hilbert*_{N+1}\+$, $\define{\resolvent{\mathcal{H}_0}[z]; (\mathcal H_0 - z)^{-1}}$ and $\xi\-\in\-D$ solves the  equation
\begin{equation}\label{resolventConstraint}
    \Gamma^{-z}\xi=\tfrac{4\pi}{\mu}\sum_{i=1}^N\-\Big(\resolvent{\mathcal{H}_0}[z]\+\psi\Big)\-\Big|_{\pi_i}.
\end{equation}
\end{theo}
\n Let us comment on our lower bound $-\lambda_0$ of the quadratic form $Q$ (and then of the infimum of the spectrum of $\mathcal{H}$).
In the course of the proof of proposition~\ref{coercivityLemma} we explicitly find (see section~\ref{closure&Boundedness}) 
\begin{equation}
    -\-\lambda_0 = \begin{dcases}
        \; -\frac{(N\!-\!1)^2\,\gamma^2}{\mu\+ \Lambda_\gamma(N,M)^2\,b^2}\,,\quad&\text{if }\alpha\geq 0\+,\\
        \; -\frac{\left[(N\!-\!1)\+ \gamma+\abs{\alpha}\+b\+ \right]^2}{\mu\+ \Lambda_\gamma(N,M)^2\,b^2},\quad&\text{if }\alpha<0
        \end{dcases}
\end{equation}
where
\begin{equation}\label{betterLambda}
    \Lambda_\gamma(N,M) =\min\!\left\{1,\,\tfrac{\pi(N-1)}{2}\+\tfrac{M+1}{\!\-\sqrt{M(M+2)}\,}(\gamma-\gamma_c)\right\}\in (0,1].
\end{equation}
We notice that $\Lambda_\gamma(N,M) \to[\!\quad\!] 0$ for $\gamma \to[\!\quad\!] \gamma_c^+$ for any choice of $N\-\geq\- 2$ and $M\->\-0$ and therefore we have $-\lambda_0 \to[\!\quad\!] \mInfty$ for $\gamma \to[\!\quad\!] \gamma_c^+$.
It is also worth mentioning that, regardless of the sign of $\alpha$, $\lambda_0$ grows as $N^2$ for $N$ large.

\n We conclude this section with an outline of the strategy of the proof.
We stress that the main technical point is proposition~\ref{coercivityLemma}, where we estimate the quadratic form $\Phi^{\lambda}$.
It is worth to notice that the hard part of the work is devoted to finding the estimate from below.

\n In section~\ref{3bodyReduction} we introduce suitable changes of coordinates and we rewrite the quadratic form $\Phi^{\lambda}$ in $\Lp{2}[\R^{3N}]$ in terms of the quadratic form $\Theta^{\+\zeta}$ in $\Lp{2}[\R^3]$ (see~\eqref{tetaz}, \eqref{3bodyPhi} and~\eqref{Phi-Phi3bodyConjuction}), that is of the type studied in~\cite[section 3]{BCFT} for the three-particle case.
This allows us to reduce the analysis to the estimate of $\Theta^{\+\zeta}$.
As a first result, we prove an estimate from above for $\Theta^{\+\zeta}$ (see proposition~\ref{boundednessPhiH-half}).

\n In section~\ref{partialWaves} we consider the expressions~\eqref{3bodyPhi} of $\Theta^{\+\zeta}$ in the Fourier space and we expand the quadratic form in partial waves (see~\eqref{SphericalHarmonicsDec}, \eqref{partialWavesDecQF}).
We also recall some known results on the terms of the expansion $F_\ell^{\,\zeta}$, $\ell\-\in\-\No\+$.

\n In section~\ref{mainEstimates} we prove some estimates that are crucial to control $F_\ell^{\,\zeta}$.
We stress that we perform a careful analysis for each value of $\ell\-\in\-\No$ that leads to a detailed control of the upper bound and the lower bound.
In particular, the result of lemma~\ref{auxiliaryPositivity} allows us to prove proposition~\ref{coercivityLemma} by introducing the threshold value $\gamma_c$ that is uniformly bounded in $M\->\-0$ and $N\-\geq\- 2$.
It is worth to mention that~\cite[lemma 3.5]{BCFT} provides an analogous result for the $s$-wave, i.e., for the case $\ell=0$.
In our framework such a result is not sufficient to obtain a uniform control on $\gamma_c$ 
since a technical constraint on the mass (depending on $N$) would be required.

\n In section~\ref{estimateTheta} we use the above technical results to obtain the key estimate from below of $\Theta^{\+\zeta}$ (see~\eqref{ThetaLowerBound}, \eqref{LambdaDef}).
We also prove another estimate from above of $\Theta^{\+\zeta}$ (see~\eqref{Thetaup2}, \eqref{LambdaPrime}) which improves the result in proposition~\ref{boundednessPhiH-half}. 

\n In section~\ref{closure&Boundedness} we show that the estimates of $\Theta^{\+\zeta}$ imply those on $\Phi^{\lambda}$ and thus we conclude the proof of proposition~\ref{coercivityLemma}.
Then, following a standard procedure, we also prove theorems~\ref{closedBoundedQF} and~\ref{hamiltonianCharacterization}. 


\n We conclude the paper with an appendix.
In section~\ref{kreinRecall} we recall some basic facts of Birman-Kre\u{\i}n-Vishik's theory of s.a. extensions.
In section~\ref{computingPotential} we discuss some useful properties of the potential $G^{\lambda}$.
In section~\ref{actuallyBuildingQF} we give a heuristic derivation of the quadratic form $Q$.

\vs\vs

\section{Reduction to a three-body problem}\label{3bodyReduction}

\vs

We start the study of  $\Phi^\lambda$, defined by~\eqref{defPhi}, introducing  suitable changes of variables that reduce the analysis to a quadratic form of the type studied in~\cite[section 3]{BCFT} for the three-particle case.
In the end we shall prove that $\Phi^{\lambda}$ is bounded in $H^{\frac{1}{2}}(\R^{3N})$.

\n Let us define the change of coordinates
\begin{equation}\label{2bodyChangeCoordinates}
    \begin{cases}
        \vect{R}=\-\sqrt{\eta\+}\-\left(\-\sqrt{\-\frac{M}{\mu}}\+\vect{x}_0\-+\!\sqrt{\-\frac{\mu}{M}}\+\vect{x}_i\-\right)\!,\\
        \vect{r}=\-\sqrt{\mu\+}\+(\vect{x}_0-\vect{x}_i)
    \end{cases}\nquad\iff
    \begin{cases}
        \vect{x}_0=\-\sqrt{\frac{\mu\+\eta}{M}}\vect{R}+\frac{\sqrt{\mu}\+\eta}{M}\vect{r}\+,\\
        \vect{x}_i=\-\sqrt{\frac{\mu\+\eta}{M}}\vect{R}-\frac{\eta}{\!\-\sqrt{\mu\,}\+}\vect{r}
    \end{cases}
\end{equation}
where $\eta$ is the modified reduced mass
\begin{equation}
    \define{\eta;\frac{M\-+\-1}{M\-+\-2}}\,.
\end{equation}
Then, let us define a unitary operator that encodes the transformation~\eqref{2bodyChangeCoordinates}
\begin{subequations}
\begin{equation}\label{changeCoordinates2}
\begin{split}
    &\mspace{135mu}\maps{\mathscr{K}_i}{\Lp{2}[\R^{3N}];\Lp{2}[\R^{3N}\!,\,d\vect{r}d\vect{R}\+d\vect{x}_1\-\cdots d\check{\vect{x}}_i\cdots d\vect{x}_{N-1}]},\\
    &\define{(\mathscr{K}_i\psi)(\vect{r},\vect{R}\-,\vect{x}_1,\ldots\check{\vect{x}}_i\ldots,\vect{x}_{N-1})\-;\\
    &\mspace{90mu}\left(\-\tfrac{\eta}{M}\-\right)^{\!\frac{3}{4}}\!\psi\bigg(\!\sqrt{\tfrac{\mu\+\eta}{M}}\vect{R}\-+\tfrac{\sqrt{\mu}\+\eta}{M}\vect{r}\-,\vect{x}_1,\ldots,\vect{x}_{i-1},\sqrt{\tfrac{\mu\+ \eta}{M}}\vect{R}\--\tfrac{\eta}{\!\-\sqrt{\mu\,}\+}\+\vect{r},\vect{x}_{i+1},\ldots,\vect{x}_{N-1}\!\bigg)}\-.
\end{split} 
\end{equation}
Notice that $\mathscr{K}_i$ is unitary since $\left(\frac{\eta}{M}\right)^{\!\frac{3}{2}}$ is the Jacobian associated to~\eqref{2bodyChangeCoordinates}.
Similarly, one has
\begin{equation}\label{antiChangeCoordinates2}
\begin{split}
    &\mspace{135mu}\maps{\adj{\mathscr{K}}_i}{\Lp{2}[\R^{3N}\!,\,d\vect{r}d\vect{R}\+d\vect{x}_1\-\cdots d\check{\vect{x}}_i\cdots d\vect{x}_{N-1}];\Lp{2}[\R^{3N}]},\\
    &(\adj{\mathscr{K}}_i\psi)(\vect{x}_0,\vect{x}_1,\ldots,\vect{x}_{N-1})\-=\\
    &\mspace{126mu}\left(\-\tfrac{M}{\eta}\-\right)^{\!\-\frac{3}{4}}\!\-\psi\bigg(\!\sqrt{\mu\+}\+(\vect{x}_0\--\vect{x}_i), \sqrt{\-\tfrac{M\+\eta}{\mu}}\+\vect{x}_0\-+\!\sqrt{\-\tfrac{\mu\+\eta}{M}}\+\vect{x}_i\+,\vect{x}_1,\ldots\check{\vect{x}}_i\ldots,\vect{x}_{N-1}\!\bigg)\-.
\end{split} 
\end{equation}
\end{subequations}
Next, for any given $\xi\in H^{\frac{1}{2}}(\R^{3N})$, we define $\chi\-\in\- H^{\frac{1}{2}}(\R^{3N}\!,\,d\vect{r}d\vect{R}\+d\vect{x}_2\cdots d\vect{x}_{N-1})$ given by 
\begin{equation}\label{modifiedCharge}
    \define{\chi(\vect{r}\-,\vect{R},\vect{x}_2,\ldots,\vect{x}_{N-1});(\mathscr{K}_1\+\xi)(\vect{r}\-,\vect{R},\vect{x}_2,\ldots,\vect{x}_{N-1})}.
\end{equation}
Moreover we define 
\begin{equation}\label{rescaledModifiedCharge}
\begin{split}
    &\define{\FT{\phi}(\vect{\sigma},\tilde{\vect{k}});(\tilde{k}^2\-+\lambda)^{\frac{3}{4}}\FT{\chi}\!\left(\-\sqrt{\tilde{k}^2\-+\lambda\+}\,\vect{\sigma},\tilde{\vect{k}}\right)}\\
    & =\! (\tilde{k}^2\!+\-\lambda)^{\frac{3}{4}} \! \left( \-\tfrac{M}{\eta}\-\right)^{\!\-\frac{3}{4}}\! \FT{\xi}\! \left(\-\sqrt{\mu}\+\sqrt{ \tilde{k}^2\!+\-\lambda} \,\vect{\sigma} + \-\sqrt{\-\tfrac{M\+\eta}{\mu}} \vect{k}_1 , 
    \sqrt{\-\tfrac{\mu\+\eta}{M}} \vect{k}_1\--\sqrt{\mu}\+\sqrt{ \tilde{k}^2\!+\!\lambda}\, \vect{\sigma},\vect{k}_2,\ldots , \vect{k}_{N-1} \!
    \right)
    \end{split}
\end{equation}
where $\tilde{\vect{k}} = (\vect{k}_1,\ldots,\vect{k}_{N-1})$.
In the next lemma we rewrite $\Phi^{\lambda}[\xi]$ in terms of $\FT{\phi}$.
\begin{lemma}
    For any $\xi\in H^{\frac{1}{2}}(\R^{3N})$ one has
    \begin{equation}\label{Phi3b}
        \begin{split}
            \Phi^\lambda[\xi]\-=&\:\Phi_0[\xi]+\tfrac{4\pi N}{\sqrt{\mu}}\!\-\integrate[\R^{3(N-1)}]{\sqrt{\tilde{k}^2\-+\lambda\+};\mspace{-33mu}d\tilde{\vect{k}}}\+\bigg[\-\integrate[\R^3]{\!\- \sqrt{\tfrac{\mu}{\eta}\+ \sigma^2\- +1\+}\,\abs{\FT{\phi}(\vect{\sigma},\tilde{\vect{k}})}^2;\-d\vect{\sigma}}+\\
            &+\frac{(N\!-\!1)\gamma}{2\pi^2}\!\-\integrate[\R^6]{\frac{\conjugate*{\FT{\phi}(\vect{\sigma},\tilde{\vect{k}})}\, \FT{\phi}(\vect{\tau},\tilde{\vect{k}})}{\abs{\vect{\sigma}-\vect{\tau}}^2};\-d\vect{\sigma}d\vect{\tau}}-\frac{N\!-\!1}{2\pi^2}\!\-\integrate[\R^6]{\frac{\conjugate*{\FT{\phi} (\vect{\sigma},\tilde{\vect{k}})}\, \FT{\phi}(\vect{\tau},\tilde{\vect{k}})}{\sigma^2\- +\tau^2\-+\frac{2\,\vect{\sigma}\cdot\vect{\tau}}{M+1}+1};\-d\vect{\sigma}d\vect{\tau}}\+\bigg] .
        \end{split}
    \end{equation}
    \begin{proof}
    Let us define a new hermitian quadratic form $\maps{\tilde{\Phi}^\lambda}{\Lp{2}[\R^{3N}];\R}$, given by
    \begin{subequations}\label{componentsPhi2}
    \begin{gather}
        \define{\tilde{\Phi}^\lambda;\tilde{\Phi}_{\mathrm{diag}}^\lambda\-+\tilde{\Phi}_{\mathrm{off}}^\lambda\-+\tilde{\Phi}_{\mathrm{reg}}}\+,\qquad \dom{\tilde{\Phi}^\lambda}=H^{\frac{1}{2}}(\R^{3N}),\nonumber\\
        \define{\tilde{\Phi}^\lambda_{\mathrm{diag}}[\chi];\tfrac{4\pi \+N}{\sqrt{\mu}}\!\-\integrate[\R^{3N}]{\!\textstyle{\sqrt{\frac{\mu}{\eta}\+q^2\-+\-\sum_{m=1}^{N-1} k_m^2\-+\-\lambda\+}\,}\abs{\FT{\chi}(\vect{q},\vect{k}_1,\ldots,\vect{k}_{N-1})}^2;\mspace{-10mu}d\vect{q}d\vect{k}_1\-\cdots d\vect{k}_{N-1}}},\label{diagPhi2}\\
        \define{\tilde{\Phi}^\lambda_{\mathrm{off}}[\chi]; -\+ \tfrac{2N(N-1)}{\pi\+\sqrt{\mu}}\!\-\integrate[\R^{3(N+1)}]{\frac{\conjugate*{\FT{\chi}(\vect{q},\vect{k}_1,\ldots,\vect{k}_{N-1})}\+\FT{\chi}(\vect{p},\vect{k}_1,\ldots,\vect{k}_{N-1})}{q^2+p^2+\frac{2\,\vect{q}\,\cdot\+\vect{p}}{M+1}+\-\sum_{m=1}^{N-1} k_m^2+\lambda}; \mspace{-33mu}d\vect{q}d\vect{p}\+d\vect{k}_1\-\cdots d\vect{k}_{N-1}}},\label{offPhi2}\\
        \define{\tilde{\Phi}_{\mathrm{reg}}[\chi];\tfrac{2N(N-1)\,\gamma}{\pi\+\sqrt{\mu}}\!\-\integrate[\R^{3(N+1)}]{\frac{\conjugate*{\FT{\chi}(\vect{q},\vect{k}_1,\ldots,\vect{k}_{N-1})}\+\FT{\chi}(\vect{p},\vect{k}_1,\ldots,\vect{k}_{N-1})}{\abs{\vect{q}-\vect{p}}^2};\mspace{-33mu}d\vect{q}d\vect{p}\+d\vect{k}_1\-\cdots d\vect{k}_{N-1}}}\+ \label{regPhi2}
    \end{gather}
    \end{subequations}
    where $\chi$ is given by \eqref{modifiedCharge}.
    We want to show that
    \begin{equation*}
        \Phi^\lambda_{\mathrm{diag}}[\xi]=\tilde{\Phi}^\lambda_{\mathrm{diag}}[\chi],\quad\Phi^\lambda_{\mathrm{off}}[\xi]=\tilde{\Phi}^\lambda_{\mathrm{off}}[\chi],\quad\Phi_{\mathrm{reg}}[\xi]=\tilde{\Phi}_{\mathrm{reg}}[\chi].
    \end{equation*}
Let us show that $\tilde{\Phi}_{\mathrm{reg}}[\chi]=\Phi_{\mathrm{reg}}[\xi]$.
    To this end, we adopt the change of variables associated to $\mathscr{K}_1$ in~\eqref{regPhi}, yielding
    \begin{align*}
        \Phi_{\mathrm{reg}}[\xi]&=\tfrac{4\pi N(N-1)\,\gamma}{\mu}\!\-\integrate[\R^{3N}]{\frac{\abs{\xi(\vect{x}_0,\vect{x}_1,\ldots,\vect{x}_{N-1})}^2}{\abs{\vect{x}_1\--\vect{x}_0}};\mspace{-10mu}d\vect{x}_0d\vect{x}_1\-\cdots d\vect{x}_{N-1}}\\
        &=\tfrac{4\pi N(N-1)\,\gamma}{\mu}\!\-\integrate[\R^{3N}]{\frac{\!\-\sqrt{\mu\,}\+}{r}\,\abs{\chi(\vect{r},\vect{R},\vect{x}_2,\ldots,\vect{x}_{N-1})}^2;\mspace{-10mu}d\vect{r}d\vect{R}\+d\vect{x}_2\cdots d\vect{x}_{N-1}}\\[-2pt]
        &=\tfrac{2\+N(N-1)\,\gamma}{\pi\+\sqrt{\mu}}\!\-\integrate[\R^{3(N+1)}]{\frac{\conjugate*{\FT{\chi}(\vect{q},\vect{Q},\vect{k}_2,\ldots,\vect{k}_{N-1})}\+\FT{\chi}(\vect{p},\vect{Q},\vect{k}_2,\ldots,\vect{k}_{N-1})\-}{\abs{\vect{q}-\vect{p}}^2};\mspace{-33mu}d\vect{q}d\vect{p}\+d\vect{Q}\+d\vect{k}_2\cdots d\vect{k}_{N-1}}\\
        &=\tilde{\Phi}_{\mathrm{reg}}[\chi].
    \end{align*}
    Notice that in the last step we have used the identity
    \begin{equation}\label{distributionalFourier}
        \integrate[\R^3]{\frac{\abs{f(\vect{r})}^2\!\-}{r};\-d\vect{r}}=\frac{1}{2\pi^2\!\-}\-\integrate[\R^6]{\frac{\conjugate*{\FT{f}(\vect{p})}\+\FT{f}(\vect{q})}{\abs{\vect{p}-\vect{q}}^2};\-d\vect{p}\+d\vect{q}},\qquad\forall f\in H^{\frac{1}{2}}(\R^3).
    \end{equation}
    In order to prove the same result for $\tilde{\Phi}^\lambda_{\mathrm{diag}}$ and $\tilde{\Phi}^\lambda_{\mathrm{off}}\+$, it is helpful to deal with $\mathscr{K}_i$ in Fourier space.
    If we denote with $(\vect{q},\+\vect{Q})$ the couple of conjugate variables of $(\vect{r},\+\vect{R})$, while $(\vect{p},\+\vect{k}_i)$ are conjugate to $(\vect{x}_0,\+\vect{x}_i)$, in the space of momenta~\eqref{2bodyChangeCoordinates} reads
    \begin{equation}\label{2bodyFourierChangeCoordinates}
        \begin{cases}
            \vect{Q}=\-\sqrt{\-\frac{\mu\+\eta}{M}}(\vect{p}+\vect{k}_i),\\
            \vect{q}=\frac{\sqrt{\mu}\+\eta}{M}\+\vect{p}-\frac{\eta}{\!\-\sqrt{\mu\,}\+}\+\vect{k}_i
        \end{cases}\nquad\iff
        \begin{cases}
            \vect{p}=\-\sqrt{\mu}\, \vect{q}+\-\sqrt{\-\frac{M\+\eta}{\mu}}\+\vect{Q}\+,\\
            \vect{k}_i=\--\sqrt{\mu}\,\vect{q} +\-\sqrt{\-\frac{\mu\+\eta}{M}}\+\vect{Q}\+.
        \end{cases}
    \end{equation}
    In particular, one can verify that
    \begin{equation}\label{fourierChi}
        \FT{\chi}(\vect{q},\vect{Q},\vect{k}_2,\ldots,\vect{k}_{N-1})=\left(\-\tfrac{M}{\eta}\-\right)^{\!\frac{3}{4}}\!\FT{\xi}\!\left(\-\sqrt{\mu}\, \vect{q}+\-\sqrt{\-\tfrac{M\+\eta}{\mu}}\+\vect{Q}\+,\sqrt{\-\tfrac{\mu\+\eta}{M}}\+\vect{Q}-\-\sqrt{\mu}\,\vect{q},\vect{k}_2,\ldots,\vect{k}_{N-1}\right)\!.
    \end{equation}
    Hence, considering that $\frac{p^2}{M+1}+k_1^2=Q^2\-+\frac{\mu}{\eta}\+q^2$, one has
    \begin{align*}
        \Phi^\lambda_{\mathrm{diag}}[\xi]&=\tfrac{4\pi N}{\!\-\sqrt{\mu\,}\+}\!\-\integrate[\R^{3N}]{\!\textstyle{\sqrt{\frac{p^2}{M+1}+\-\sum_{m=1}^{N-1} k_m^2\-+\-\lambda\+}\,}\abs{\FT{\xi}(\vect{p},\vect{k}_1,\ldots,\vect{k}_{N-1})}^2;\mspace{-10mu}d\vect{p}\+d\vect{k}_1\-\cdots d\vect{k}_{N-1}}\\
        &=\tfrac{4\pi N}{\!\-\sqrt{\mu\,}\+}\!\-\integrate[\R^{3N}]{\!\textstyle{\sqrt{Q^2\-+\frac{\mu}{\eta}\+q^2\-+\-\sum_{m=2}^{N-1} k_m^2\-+\-\lambda\+}\,}\abs{\FT{\chi}(\vect{q},\vect{Q},\vect{k}_2,\ldots,\vect{k}_{N-1})}^2;\mspace{-10mu}d\vect{Q}d\vect{q}d\vect{k}_2\cdots d\vect{k}_{N-1}}\\
        &=\tilde{\Phi}^\lambda_{\mathrm{diag}}[\chi].
    \end{align*}
    In the remaining case, we need the change of coordinates
    \begin{equation*}
        \begin{cases}
            \vect{p}=\!\sqrt{\mu}\+ (\vect{q}_1\-+\vect{q}_2)\- +\- \sqrt{\mu\+\eta\+M}\+\vect{Q},\\
            \vect{k}_1=\- -\sqrt{\mu}\,\vect{q}_2\- +\!\sqrt{\-\frac{\mu\+\eta}{M}}\+\vect{Q},\\
            \vect{k}_2=\- -\sqrt{\mu}\,\vect{q}_1\-+\!\sqrt{\-\frac{\mu\+\eta}{M}}\+\vect{Q}
        \end{cases}\nquad\implies \begin{cases}
            \vect{p}+\vect{k}_1=\!\sqrt{\mu}\+ \vect{q}_1\-+\- \sqrt{\frac{\eta \+M}{\mu}}\+\vect{Q},\\
            \vect{p}+\vect{k}_2=\!\sqrt{\mu}\+ \vect{q}_2\- +\- \sqrt{\frac{\eta \+M}{\mu}}\+\vect{Q},\\
            \frac{p^2}{M}+k_1^2\-+k_2^2=\- q_1^2\-+q_2^2\-+Q^2\!+\frac{2}{M+1}\+ \vect{q}_1\!\cdot\mspace{-0.75mu}\vect{q}_2\+.
        \end{cases}
    \end{equation*}
    Indeed, notice that this substitution of Jacobian $\left(\- \frac{\mu M}{\eta}\- \right)^{\!\-\frac{3}{2}}\!$, together with~\eqref{fourierChi} lead to
    \begin{align*}
    \Phi^\lambda_{\mathrm{off}}[\xi]&= -\tfrac{2N(N-1)\-}{\pi\,\mu^2}\!\- \integrate[\R^{3(N+1)}]{\-\frac{\conjugate*{\FT{\xi}(\vect{p}\-+\-\vect{k}_1, \vect{k}_2,\ldots, \vect{k}_N)}\+\FT{\xi}(\vect{p}\-+\-\vect{k}_2, \vect{k}_1, \vect{k}_3,\ldots, \vect{k}_N)}{\frac{1}{M}p^2+\sum_{m=1}^N k_m^2+\lambda};\mspace{-33mu}d\vect{p}\+d\vect{k}_1\- \cdots d\vect{k}_N}\\
    &=-\+ \tfrac{2N(N-1)}{\pi\+\sqrt{\mu}}\!\-\integrate[\R^{3(N+1)}]{\frac{\conjugate*{\FT{\chi}(\vect{q}_1,\vect{Q},\vect{k}_3,\ldots,\vect{k}_N)}\+\FT{\chi}(\vect{q}_2,\vect{Q},\vect{k}_3,\ldots,\vect{k}_N)}{q_1^2+q_2^2+\frac{2\,\vect{q}_1\cdot\+\vect{q}_2}{M+1}+Q^2\-+\-\sum_{m=3}^N k_m^2+\lambda}; \mspace{-33mu}d\vect{q}_1d\vect{q}_2d\vect{Q}\+d\vect{k}_3\cdots d\vect{k}_N}\\
    &=\tilde{\Phi}_{\mathrm{off}}^\lambda[\chi].
    \end{align*}
    Next, let us define the unitary scaling operator in $\Lp{2}[\R^{3N}]$
    \begin{equation}\label{scalingOperator}
    \begin{split}
        &\maps{U}{\Lp{2}[\R^3]\-\otimes\-\Lp{2}[\R^{3(N-1)}];\Lp{2}[\R^3]\-\otimes\-\Lp{2}[\R^{3(N-1)}]},\\
        \psi(\vect{p},\vect{k}_1,\ldots,\vect{k}_{N-1})&\longmapsto \textstyle{\left(\sum_{m=1}^{N-1} k_m^2+\lambda\right)^{\!\- \frac{3}{4}}\!\psi\!\left(\!\sqrt{\sum_{m=1}^{N-1}k_m^2+\lambda\+}\,\vect{p},\vect{k}_1,\ldots,\vect{k}_{N-1}\-\right)}.
    \end{split}
    \end{equation}
    We stress that $\ran{U|_{H^{1/2}(\R^{3N})}}\-\subset\- H^{\frac{1}{2}}(\R^{3N})$.
    Using the scaling defined by  $\FT{\phi}=U\FT{\chi}$ in $\tilde{\Phi}^{\lambda}$, we conclude the proof. 
    
    \end{proof}
\end{lemma}
\n Formula~\eqref{Phi3b} suggests to define the hermitian quadratic form
$\Theta^{\+\zeta}$, for $\zeta\geq 0$, with
\begin{equation}\label{tetaz}
    \dom{\Theta^{\+\zeta}}=H^{\frac{1}{2}}(\R^3),\quad \define{\Theta^{\+\zeta};\Theta^{\+\zeta}_{\mathrm{diag}}\!+\Theta^{\+\zeta}_{\mathrm{off}}\-+\Theta_{\mathrm{reg}}}\+,
\end{equation}
where, for a given $\varphi\in H^{\frac{1}{2}}(\R^3)$,
\begin{subequations}\label{3bodyPhi}
\begin{gather}
    \define{\Theta^{\+\zeta}_{\mathrm{diag}}[\varphi];\!\- \integrate[\R^3]{\!\- \textstyle{\sqrt{\frac{\mu}{\eta}\+ \sigma^2\- +\zeta\+}}\,\abs{\FT{\varphi}(\vect{\sigma})}^2;\-d\vect{\sigma}}},\label{diag3bodyPhi}\\
    \define{\Theta_{\mathrm{reg}}[\varphi];\- \frac{(N\!-\!1)\+ \gamma}{2\pi^2}\!\- \integrate[\R^6]{\frac{\conjugate*{\FT{\varphi}(\vect{\sigma})}\, \FT{\varphi}(\vect{\tau})}{\abs{\vect{\sigma}-\vect{\tau}}^2};\-d\vect{\sigma}d\vect{\tau}}}\+ ,\label{reg3bodyPhi}\\
    \define{\Theta^{\+\zeta}_{\mathrm{off}}[\varphi];\- -\frac{N\!-\!1}{2\pi^2}\!\- \integrate[\R^6]{\frac{\conjugate*{\FT{\varphi} (\vect{\sigma})}\, \FT{\varphi}(\vect{\tau})}{\sigma^2\- +\tau^2\-+\frac{2\,\vect{\sigma}\cdot\vect{\tau}}{M+1}+\zeta};\-d\vect{\sigma}d\vect{\tau}}}\+ .\label{off3bodyPhi}
\end{gather}
\end{subequations}
Observe that
\begin{equation}\label{Phi-Phi3bodyConjuction}
    \Phi^{\lambda} [\xi]= \Phi_0[\xi]+\!\tfrac{4\pi N}{\!\-\sqrt{\mu\,}\+}\!\-\integrate[\R^{3(N-1)}]{\!\textstyle{\sqrt{\sum_{m=1}^{N-1} k_m^2 +\lambda\+}\,}\Theta^{\+1}[\+\phi\+](\vect{k}_1,\ldots,\vect{k}_{N-1}); \mspace{-33mu}d\vect{k}_1\-\cdots d\vect{k}_{N-1}}\+
\end{equation}
where $\phi$ is given by~\eqref{rescaledModifiedCharge}.
Equation~\eqref{Phi-Phi3bodyConjuction} shows that the analysis of $\Phi^{\lambda}$ in $\Lp{2}[\R^{3N}]$ can be reduced to the analysis of $\Theta^{\+\zeta}$ in $\Lp{2}[\R^3]$.
In the rest of this section and in sections~\ref{partialWaves}, \ref{mainEstimates} and~ \ref{estimateTheta} we shall concentrate on the estimate from above and from below of $\Theta^{\+\zeta}$. 

\n We notice that $\varphi\in H^{\frac{1}{2}}(\R^3)$ is equivalent to $\Theta^{\+\zeta}_{\mathrm{diag}}[\varphi] <\pInfty$. 
As a first result we prove that  the whole quadratic form $\Theta^{\+\zeta}$ is well defined in $H^{\frac{1}{2}}(\R^3)$.
\begin{prop}\label{boundednessPhiH-half}
    Given $\varphi \in H^{\frac{1}{2}}(\R^{3N})$, $\gamma>0$  and $\zeta \geq 0$, we have
    \begin{equation}\label{Thetaup1}
       \abs{\Theta^{\+\zeta}[\varphi]}\leq \left[ 1+\tfrac{(N-1)(M+1)}{\sqrt{M(M+2)}} \left( \tfrac{M+1}{M} + \tfrac{\pi}{2}\+ \gamma \right) \-\right]\!
       \Theta^{\+\zeta}_{\mathrm{diag}}[\varphi]\,.
    \end{equation}
    \begin{proof}
        Concerning $\Theta^{\+\zeta}_{\mathrm{off}}\+$, we have
        \begin{equation*}
            \abs{\Theta^{\+\zeta}_{\mathrm{off}}[\varphi]}\leq\frac{N\- -\- 1}{2\pi^2} \!\- \integrate[\R^6]{\frac{\abs{\FT{\varphi}(\vect{p})}\,\abs{\FT{\varphi}(\vect{q})}}{p^2\-+q^2\-+\frac{2\,\vect{p}\,\cdot\+ \vect{q}}{M+1}}; \- d\vect{p}\+d\vect{q}}
            \leq \frac{N\- -\- 1}{2\mu\+\pi^2} \!\- \integrate[\R^6]{\frac{\abs{\FT{\varphi}(\vect{p})}\,\abs{\FT{\varphi}(\vect{q})}}{p^2\-+q^2}; \- d\vect{p}\+d\vect{q}},
        \end{equation*}
        where we have used
        $$x^2\!+y^2\!-\tfrac{2}{M+1}\,xy\geq (x^2\!+y^2)\!\left(1-\tfrac{1}{M+1}\right)\!=\mu\+(x^2\!+y^2).$$
        Let us define the integral operator $\maps{Q}{\Lp{2}(\R^3);\Lp{2}(\R^3)}$ acting as
        \begin{equation}
            \define{\left(Q\,\psi\right)\!(\vect{p})\-;\!\- \integrate[\R^3]{\frac{\psi(\vect{k})}{\!\-\sqrt{p\+}\+ (p^2\-+k^2)\sqrt{k\+}\,};\-d\vect{k}}}\+.
        \end{equation}
        Thanks to~\cite[lemma 2.1]{FT}, we know that $Q$ is bounded with norm $2\pi^2$.
        
        \n Hence, denoted $\define{g(\vect{k});\sqrt{k\+}\,\abs{\FT{\varphi}(\vect{k})}}$, we have 
        \begin{align*}
            \abs{\Theta^{\+\zeta}_{\mathrm{off}}[\varphi]} &\leq\frac{N\- -\- 1}{2\mu\+\pi^2}\+ \scalar{g}{Q\,g}[\Lp{2}(\R^3)]\leq \frac{N\- -\- 1}{\mu} \norm{g}[\Lp{2}(\R^3)]^2 = \frac{N\--\-1}{\mu}\integrate[\R^3]{\abs{\vect{k}} \abs{\FT{\varphi}(\vect{k})}^2;\-d\vect{k}}\\
            &\leq (N\--\-1)\+\sqrt{\frac{\eta}{\mu^3}} \,\Theta^{\+\zeta}_{\mathrm{diag}}[\varphi]\+.
        \end{align*}
        Let us consider $\Theta_{\mathrm{reg}}$. 
        According to~\cite{Y}\footnote{There is a typo in~\cite[equation~(1.4)]{Y}: the power $2$ on the Euler Gamma function in the numerator is missing.}, we have the following sharp Hardy-Rellich inequality
        for any $u\in H^{\alpha}(\R^d)$ and $\alpha<\frac{d}{2}$
        \begin{equation}\label{inequalityHardyRellich}
            \integrate[\R^d]{\frac{\abs{u(\vect{x})}^2\!\-}{\;\abs{\vect{x}}^{2\alpha}}; \-d\vect{x}}\leq \frac{1}{2^{2\alpha}}\frac{\Gamma^2\!\left(\frac{d}{4}-\frac{\alpha}{2}\right)}{\Gamma^2\!\left(\frac{d}{4}+\frac{\alpha}{2}\right)}\-  \integrate[\R^d]{\abs{\vect{k}}^{2\alpha}\abs{\FT{u}(\vect{k})}^2; \-d\vect{k}},
        \end{equation}
        where $\maps{\Gamma}{\R\smallsetminus \!-\No;\R}$ denotes the Euler Gamma function.
        In case $\alpha=\frac{1}{2}$ and $d=3$, using equation~\eqref{distributionalFourier}, inequality~\eqref{inequalityHardyRellich} reads
        \begin{equation*}
            \frac{1}{2\pi^2\!\-}\- \integrate[\R^6]{\frac{\conjugate*{\FT{u}(\vect{p})\-}\, \FT{u}(\vect{q})}{\abs{\vect{p}-\vect{q}}^2};\-d\vect{p}\+d\vect{q}}
            =\!\integrate[\R^3]{\frac{\abs{u(\vect{x})}^2\!\-}{\abs{\vect{x}}};\- d\vect{x}}\leq\frac{\pi}{2}\- \integrate[\R^3]{\abs{\vect{k}}\+ \abs{\FT{u}(\vect{k})}^2;\- d\vect{k}}\,. 
        \end{equation*}
        Therefore
        \begin{equation}
            0\leq\Theta_{\mathrm{reg}}[\varphi]\-\leq \- \frac{\pi}{2}(N\--\-1)\+ \gamma\, \sqrt{\frac{\eta}{\mu}} \, \Theta^{\+\zeta}_{\mathrm{diag}}[\varphi]\,.
        \end{equation}
        and the proof is complete.
    \end{proof}
\end{prop}

\vs\vs

\section{Partial waves decomposition}\label{partialWaves}

\vs

In order to establish a lower bound for $\Theta^{\+\zeta}$, in this section we start by following the first steps of~\cite[section 3]{BCFT} properly adapted to our case, i.e., we study the quadratic form decomposed in partial waves.
Given $\FT{\varphi}\in\Lp{2}[\R^3,\,\sqrt{p^2+1}\,d\vect{p}]$, one has
\begin{equation}\label{SphericalHarmonicsDec}
    \FT{\varphi}(\vect{p})=\!\sum_{\ell\+ \in\+ \No}\sum_{m\+ =\+ -\ell}^{\ell}\FT{\varphi}_{\ell,m}(p)\,Y^m_\ell(\vers{\omega}).
\end{equation}
Here, $\maps{Y_\ell^m}{\mathbb{S}^2;\C}$ denotes the Spherical Harmonic of order $\ell, m$, while $(p, \vers{\omega})\-\in\-\Rplus\!\-\times\-\mathbb{S}^2$ represents $\vect{p}\-\in\-\R^3$ in spherical coordinates and $\FT{\varphi}_{\ell,m}\in\Lp{2}[\Rplus,\,p^2\sqrt{p^2+1}\,dp]$ are the Fourier coefficients of $\FT{\varphi}$.
Accordingly, we decompose the quadratic form $\Theta^{\+\zeta}$ for any $\zeta\geq 0$
\begin{gather}\label{partialWavesDecQF}
    \Theta^{\+\zeta}[\varphi]=\!\sum_{\ell\+ \in\+ \No}\sum_{m\+ =\+ -\ell}^\ell F^{\,\zeta}_\ell[\FT{\varphi}_{\ell,m}],\\
    \maps{F^{\,\zeta}_{\ell}}{\Lp{2}[\Rplus,\,p^2dp];\R},\qquad     \dom{F^{\,\zeta}_\ell}=\Lp{2}[\Rplus,\,p^2\sqrt{p^2+1}\,dp].
\end{gather}
As usual, we consider the three-components
\begin{equation}\label{sphericalDecQF}
    \define{F^{\,\zeta}_{\ell}\!;F^{\,\zeta}_{\mathrm{diag}}\!+\- F^{\,\zeta}_{\mathrm{off};\,\ell}\-+\- F_{\mathrm{reg};\,\ell\,}},
\end{equation}
each of which is described in the following lemma.
From now on, let $P_\ell(y)=\frac{1}{2^\ell \ell!}\frac{\mathrm{d}^\ell}{\mathrm{d}y^\ell}(y^2-1)^\ell$ be the Legendre polynomial of degree $\ell\-\in\-\No\+$.
\begin{lemma}\label{partialWavesQF}
    For any $\psi\-\in\-\Lp{2}[\Rplus,\, p^2\sqrt{p^2+1}\, dp]$, taking into account decomposition~\eqref{partialWavesDecQF} and definition~\eqref{sphericalDecQF}, we have the following expressions for any $\zeta\geq 0,\, \ell\in\N_0$
    \begin{subequations}
    \begin{gather}\label{diagPartialWavesQF}
        F^{\,\zeta}_{\mathrm{diag}}[\psi]\- =\!\- \integrate[0;\pInfty]{k^2\!\sqrt{\tfrac{\mu}{\eta}\+ k^2\- +\zeta\+}\,\abs{\psi(k)}^2;\nquad dk},\\
        \label{regPartialWavesQF}
        F_{\mathrm{reg};\,\ell}[\psi]\- =\tfrac{(N-1)\gamma}{\pi}\-\!\integrate[0;\pInfty]{p^2;\nquad dp}\!\- \integrate[0;\pInfty]{q^2\,\conjugate*{\psi(p)}\+\psi(q);\nquad dq}\-\!\integrate[-1;1]{\frac{P_{\ell}(y)}{p^2\-+q^2\--2\+p\+q\,y\,};\-dy}\+,\\
        \label{offPartialWavesQF}
        F^{\,\zeta}_{\mathrm{off};\,\ell}[\psi]\- =\- -\tfrac{N-1}{\pi}\-\!\integrate[0;\pInfty]{p^2;\nquad dp}\!\- \integrate[0;\pInfty]{q^2\,\conjugate*{\psi(p)}\+\psi(q);\nquad dq}\-\!\integrate[-1;1]{\frac{P_{\ell}(y)}{p^2\-+q^2\-+\frac{2}{M+1}\+p\+q\,y+\zeta\,};\-dy}\+.
    \end{gather}
    \end{subequations}
    \begin{proof}
        The result is proved in~\cite[lemma 3.1]{CDFMT} and here we 
        give some details for reader's convenience.
       Let us consider the diagonal contribution. Using spherical coordinates in~\eqref{diag3bodyPhi}
        so that any $\FT{\varphi}(\vect{p})\-\in\-\Lp{2}[\R^3,\,\sqrt{p^2+1}\,d\vect{p}]$ can be replaced in $\Theta^{\+\zeta}_{\mathrm{diag}}[\varphi]$ via decomposition~\eqref{SphericalHarmonicsDec}, one gets
        \begin{equation}
           \Theta^{\+\zeta}_{\mathrm{diag}}[\varphi]=\!\sum_{\ell\+ \in\+ \No}\sum_{m\+ =\+ -\ell}^\ell F^{\,\zeta}_{\mathrm{diag}}[\FT{\varphi}_{\ell,m}],
        \end{equation}
        thanks to the orthonormality of $Y^m_\ell\-$, with $F^{\,\zeta}_{\mathrm{diag}}$ given by~\eqref{diagPartialWavesQF}.
        
        \n Regarding the regularizing contribution, the same procedure yields
        \begin{equation}\label{regPartialWavesInProgress}
           \begin{split}
           \Theta_{\mathrm{reg}}[\varphi]=\!\sum_{\substack{\ell_1\+ \in\+ \No\\\ell_2\+ \in\+ \No}}\+\sum_{m_1\+ =\+ -\ell_1}^{\ell_1}
           \+\sum_{m_2\+ =\+ -\ell_2}^{\ell_2}\!\!\! \tfrac{(N-1)\+\gamma}{2\pi^2}\!\-\integrate[0;\pInfty]{p_1^2; \nquad dp_1}\!\-\integrate[0;\pInfty]{p_2^2\,\conjugate*{\FT{\varphi}_{\ell_1,m_1}(p_1)}\+\FT{\varphi}_{\ell_2,m_2}(p_2); \nquad dp_2}\times\\[-12.5pt]
           \times\!\-\integrate[\mathbb{S}^2]{;d\vers{\omega}_1}\!\!\-\integrate[\mathbb{S}^2]{\frac{\conjugate*{Y^{m_1}_{\ell_1}(\vers{\omega}_1)}\+ Y^{m_2}_{\ell_2}(\vers{\omega}_2)}{p_1^2+p_2^2-2p_1p_2\,\vers{\omega}_1\!\cdot\vers{\omega}_2};d\vers{\omega}_2}\+.
           \end{split}
        \end{equation}
Let us  decompose the function $f_{p_1\+p_2}\-: x\longmapsto\frac{1}{p_1^2+\+p_2^2-2\,p_1 p_2\,x\,}$ in terms of Legendre polynomials, for almost every $(p_1,p_2)\-\in\-\Rplus\!\-\times\-\Rplus$
        \begin{equation}\label{decLegendreP}
        f_{p_1\+p_2}\- (x)=\!\sum_{\ell\+\in\+\N_0}\tfrac{2\ell+1}{2}\scalar{P_\ell\+}{f_{p_1\+p_2}}[\Lp{2}(-1,\,1)]\,P_\ell(x),\quad \forall x\in[-1,\,1].
        \end{equation}
        Making use of the addition formula for the Spherical Harmonics
        \begin{equation}\label{additionFormulaSH}
            P_\ell(\vers{\omega}_1\!\cdot\vers{\omega}_2)=\frac{4\pi}{2\ell+1}\!\sum_{m=-\ell}^\ell Y^m_\ell(\vers{\omega}_1)\, \conjugate*{Y^m_\ell(\vers{\omega}_2)},
        \end{equation}
        one has the following decomposition
        \begin{equation}
        f_{p_1\+p_2}(\vers{\omega}_1\!\cdot\vers{\omega}_2)\!=\!\sum_{\ell\+\in\+\N_0}2\pi\!\- \integrate[-1;1]{\frac{P_\ell(y)}{p_1^2+p_2^2-2p_1p_2\,y};\-dy}\sum_{m=-\ell}^{\ell}Y^m_\ell(\vers{\omega}_1)\,\conjugate*{Y^m_\ell(\vers{\omega}_2)}.
        \end{equation}
Replacing the previous expression in \eqref{regPartialWavesInProgress}, we find
        \begin{equation}
           \Theta_{\mathrm{reg}}[\varphi]=\!\sum_{\ell\+ \in\+ \No}\sum_{m\+ =\+ -\ell}^\ell F_{\mathrm{reg;\,\ell}}[\FT{\varphi}_{\ell,m}],
        \end{equation}
        with $F_{\mathrm{reg};\,\ell}$ given by~\eqref{regPartialWavesQF}.
        
        \n The computation related to the off-diagonal contribution is exactly the same.
        
    \end{proof}
\end{lemma}
\n In the next lemma we characterize the sign of $F^{\,\zeta}_{\mathrm{off};\,\ell}$ and $F_{\mathrm{reg};\,\ell}\+$.
\begin{lemma}\label{FsignLemma}
Let $\psi \in\Lp{2}(\Rplus,\,p^2\sqrt{p^2+1}\,dp)$.
Then, for all $\ell\in\No$ one has
$F_{\mathrm{reg};\,\ell}\geq 0$ and
\begin{equation*}
    \begin{cases}
       F^{\+0}_{\mathrm{off};\,\ell}[\psi]\geq F^{\,\zeta}_{\mathrm{off};\,\ell}[\psi]\geq 0,\quad & \text{if }\ell \text{ is odd,}\\
       F^{\+0}_{\mathrm{off};\,\ell}[\psi]\leq F^{\,\zeta}_{\mathrm{off};\,\ell}[\psi]\leq 0,\quad & \text{if }\ell \text{ is even.}
    \end{cases}
\end{equation*}
\begin{proof}
The result concerning $F^{\,\zeta}_{\mathrm{off};\,\ell}$ has been proved in~\cite[lemma 3.3]{CDFMT} and the same procedure can be adapted to obtain $F_{\mathrm{reg};\,\ell}\geq 0$.

\end{proof}
\end{lemma}
\n Notice that, thanks to lemma~\ref{FsignLemma}, for the sake of a lower bound, we can neglect $F^{\,\zeta}_{\mathrm{off};\,\ell}$ with $\ell$ odd and focus on $F^{\+0}_{\mathrm{off};\,\ell}$ that represents a lower estimates for $F^{\,\zeta}_{\mathrm{off};\,\ell}$ in case $\ell$ is even.
We conclude this section presenting the diagonalization of the quadratic form $F^{\+0}_\ell$.
\begin{lemma}\label{offregDiagonalization}
Given $\psi\in\Lp{2}(\Rplus,\,k^2\sqrt{k^2+1}\,dk)$, let $g_\psi\in\Lp{2}(\R)$ be defined by
\begin{equation}
    g_\psi(p)\vcentcolon=\frac{1}{\!\-\sqrt{2\pi\,}\+}\!\- \integrate[\R]{e^{-i p t}e^{2t}\psi(e^t); dt}.
\end{equation}
Then, considering the quantities computed in lemma~\ref{partialWavesQF}, one has
\begin{subequations}
\begin{align}
    F^{\+0}_{\mathrm{diag}}[\psi]&=\sqrt{\frac{\mu}{\eta}}
    \integrate[\R]{\abs{g_\psi(p)}^2; dp},\label{diagDiagonalizedQF}\\
    F^{\+0}_{\mathrm{off};\,\ell}[\psi]&=\frac{N\! -\! 1}{2}\!\- \integrate[\R]{\abs{g_\psi(p)}^2\+ S_{\mathrm{off};\,\ell}(p); dp},\label{offDiagonalizedQF}\\
    F_{\mathrm{reg};\,\ell}[\psi]&=\frac{N\! -\! 1}{2}\!\- \integrate[\R]{\abs{g_\psi(p)}^2\+ S_{\mathrm{reg};\,\ell}(p); dp},\label{regDiagonalizedQF}
\end{align}
\end{subequations}
where
\begin{subequations}
\begin{equation}\label{offS}
    S_{\mathrm{off};\,\ell}(p)=\begin{dcases}
    -\!\-\integrate[-1;1]{P_\ell(y)\;\frac{\cosh\!\left(p\arcsin\frac{y}{M+1}\right)}{\!\-\sqrt{1-\frac{y^2}{(M+1)^2\mspace{-9mu}}\,} \,\cosh\!\left(\frac{\pi}{2} p\right)};\-dy}\+,\quad &\text{if }\ell \text{ is even},\\
    \integrate[-1;1]{P_\ell(y)\;\frac{\sinh\!\left(p\arcsin\frac{y}{M+1}\right)}{\!\-\sqrt{1-\frac{y^2}{(M+1)^2\mspace{-9mu}}\,} \,\sinh\!\left(\frac{\pi}{2} p\right)};\-dy}\+,\quad &\text{if }\ell \text{ is odd}.
    \end{dcases}
\end{equation}
\begin{equation}\label{regS}
    S_{\mathrm{reg};\,\ell}(p)=\begin{dcases}
    \gamma\!\integrate[-1;1]{P_\ell(y)\;\frac{\cosh\!\left(p\arcsin{y}\right)}{\!\-\sqrt{1-y^2} \,\cosh\!\left(\frac{\pi}{2} p\right)};\-dy}\+,\quad &\text{if }\ell \text{ is even},\\
    \gamma\!\integrate[-1;1]{P_\ell(y)\;\frac{\sinh\!\left(p\arcsin{y}\right)}{\!\-\sqrt{1-y^2} \,\sinh\!\left(\frac{\pi}{2} p\right)};\-dy}\+,\quad &\text{if }\ell \text{ is odd}.
    \end{dcases}
\end{equation}
\end{subequations}
Moreover,
\begin{subequations}\label{SregoffSigns}
\begin{equation}\label{SoffSign}
    \begin{cases}
        S_{\mathrm{off};\,\ell}(p)\leq S_{\mathrm{off};\,\ell+2}(p)\leq 0,\quad \text{if }\ell\text{ is even};\\
        S_{\mathrm{off};\,\ell}(p)\geq S_{\mathrm{off};\,\ell+2}(p)\geq 0,\quad \text{if }\ell\text{ is odd},
    \end{cases}
\end{equation}
while 
\begin{equation}\label{SregSign}
    S_{\mathrm{reg};\,\ell}(p)\geq S_{\mathrm{reg};\,\ell+2}(p)\geq 0,\quad\forall \ell\in\N_0\+ .
\end{equation}
\end{subequations}
\begin{proof}
The results about the diagonal and off-diagonal terms are proved in~\cite[lemma 3.4]{CDFMT}, whereas the statements regarding the regularizing contribution have been shown in~\cite[lemma 3.4]{BCFT}.

\end{proof}
\end{lemma}

\vs\vs

\section{Key estimates}\label{mainEstimates}

\vs

In the following, we obtain some key estimates useful to control $F^{\,\zeta}_{\ell}$.
We first provide the upper bound.
\begin{lemma}\label{SoffregAboveBound}
    For any $\psi\in\Lp{2}(\Rplus,\,k^2\sqrt{k^2+1}\,dk), \zeta\geq 0$ and $\ell\in\No\+ $, one has
    \begin{subequations}
    \begin{equation}
        F_{\mathrm{reg};\,\ell}[\psi]\leq \begin{cases}
            \frac{\pi \gamma}{2}\,\frac{(N-1)(M+1)}{\sqrt{M(M+2)}}\+ F^{\,\zeta}_{\mathrm{diag}}[\psi],\quad&\text{if }\ell\text{ is even},\\
            \frac{2 \gamma}{\pi}\,\frac{(N-1)(M+1)}{\sqrt{M(M+2)}}\+ F^{\,\zeta}_{\mathrm{diag}}[\psi],\quad&\text{if }\ell\text{ is odd},
        \end{cases}
    \end{equation}
    \begin{equation}
        F^{\+0}_{\mathrm{off};\,\ell}[\psi]\leq\begin{cases}
            0,\quad&\text{if }\ell\text{ is even},\\
            \frac{2(N-1)(M+1)^2\!\-}{\pi}\- \left[\+\frac{1}{\!\-\sqrt{M(M+2)}\,}-\arcsin{\frac{1}{1+M}}\right]\!F^{\,\zeta}_{\mathrm{diag}}[\psi],\quad&\text{if }\ell\text{ is odd}.
        \end{cases} 
    \end{equation}
    \end{subequations}
    \begin{proof}
        Let us first consider the case $\ell$ even.
        According to~\eqref{SoffSign} we have $S_{\mathrm{off};\,\ell}\leq 0$ while, from~\eqref{regS} and~\eqref{SregSign}, we know that
        \begin{equation*}
            S_{\mathrm{reg};\,\ell}(p)\leq S_{\mathrm{reg};\,0}(p)=2\gamma\,\frac{\tanh\!\left(\frac{\pi}{2}p\right)}{p}\leq \pi\gamma.
        \end{equation*}
        Hence, it is straightforward to see that
        \begin{equation*}
            F_{\mathrm{reg};\,\ell}[\psi]\leq \tfrac{N-1}{2}\sqrt{\tfrac{\eta}{\mu}}\,S_{\mathrm{reg};\,0}(0)\+  F^{\+0}_{\mathrm{diag}}[\psi]\leq \tfrac{\pi\gamma}{2}\,\tfrac{(N-1)(M+1)}{\sqrt{M(M+2)\,}}\+  F^{\,\zeta}_{\mathrm{diag}}[\psi].
        \end{equation*}
        Next, we consider $\ell$ odd.
        In light of~\eqref{SregoffSigns}, both $S_{\mathrm{reg};\,\ell}(p)$ and $S_{\mathrm{off};\,\ell}(p)$ are maximised at $\ell=1$, uniformly in $p\in\R$.
        Furthermore, thanks to the parity of the integrand, we get
        \begin{align*}
            S_{\mathrm{off};\,1}(p)=2\!\- \integrate[0;1]{\frac{y\+ \sinh\!\left(p\arcsin\frac{y}{1+M}\right)}{\sqrt{1-\frac{y^2}{(1+M)^2\mspace{-9mu}}\,}\,\sinh\!\left(\frac{\pi}{2}p\right)};\-dy}=\frac{2(M\-+\-1)^2\!\-}{\sinh\!\left(\frac{\pi}{2}p\right)}\! \integrate[0;\arcsin{\frac{1}{M+1}}]{\sin{u}\,\sinh(p\+u);\mspace{-57mu}du}\+,\\
            S_{\mathrm{reg};\,1}(p)=2\gamma\!\- \integrate[0;1]{\frac{y\+ \sinh\!\left(p\arcsin{y}\right)}{\sqrt{1-y^2\,}\sinh\!\left(\frac{\pi}{2}p\right)};\-dy}=\frac{2\gamma}{\sinh\!\left(\frac{\pi}{2}p\right)}\!\integrate[0;\frac{\pi}{2}]{\sin{u}\,\sinh(p\+u);\!du}\+.
        \end{align*}
        In both cases, we are dealing with decreasing functions in $p>0$.
        Indeed, one has that the function $p\longmapsto\frac{\sinh(p\+u)}{\sinh\left(\-\frac{\pi}{2}p\-\right)}$ is non-increasing, uniformly in $u\in[0,\frac{\pi}{2}]$.
        Thus, performing the derivatives $\frac{\mathrm{d}}{\mathrm{d}p}S_{\mathrm{off};\,1}(p)$ and $\frac{\mathrm{d}}{\mathrm{d}p}S_{\mathrm{reg};\,1}(p)$, one gets negative quantities for any $p>0$, since the integrands are negative for almost every $u\in[0,\frac{\pi}{2}]$.
        
        \n Hence, since $S_{\mathrm{off};\,1}(p)$ and $S_{\mathrm{reg};\,1}(p)$ are even and decreasing in $p>0$, the maximum is attained at the point $p=0\+$.
        Therefore,
        \begin{align*}
            S_{\mathrm{off};\,\ell}(p)\-&\leq\- S_{\mathrm{off};\,1}(0)\-=\frac{4(M\!+\!1)^2\mspace{-6mu}}{\pi}\!\integrate[0;\arcsin{\frac{1}{1+M}}]{u\,\sin(u);\mspace{-60mu}du}\-=\frac{4(M\!+\!1)}{\pi}\-\left[1\--\-\sqrt{M(M\!+\- 2)}\+\arcsin\frac{1}{1\!+\!M}\right]\!,\\[-5pt]
            S_{\mathrm{reg};\,\ell}(p)\-&\leq\- S_{\mathrm{reg};\,1}(0)\-=\frac{4\gamma}{\pi}\- \integrate[0;\frac{\pi}{2}]{u\,\sin(u);du}\-=\frac{4\gamma}{\pi}.
        \end{align*}
        In conclusion, for $\ell$ odd we have
        \begin{align*}
            F_{\mathrm{reg};\,\ell}[\psi]&\leq \tfrac{N-1}{2}\sqrt{\tfrac{\eta}{\mu}}\,S_{\mathrm{reg};\,1}(0)\+  F^{\+0}_{\mathrm{diag}}[\psi]\leq \tfrac{2\gamma}{\pi}\,\tfrac{(N-1)(M+1)}{\sqrt{M(M+2)\,}}\+  F^{\,\zeta}_{\mathrm{diag}}[\psi];\\
            F_{\mathrm{off};\,\ell}[\psi]&\leq \tfrac{N-1}{2}\sqrt{\tfrac{\eta}{\mu}}\,S_{\mathrm{off},\,1}(0)\+  F^{\+0}_{\mathrm{diag}}[\psi]\leq \tfrac{2(N-1)(M+1)^2\!\-}{\pi}\- \left[\+\tfrac{1}{\!\-\sqrt{M(M+2)}\,}-\arcsin\tfrac{1}{1+M}\right]\!F^{\,\zeta}_{\mathrm{diag}}[\psi].
        \end{align*}
    \end{proof}
\end{lemma}
\n Let us introduce some further notation.
For any $a\in\C$ and $n\in\No\+$, let $(a)_n$ be the Pochhammer symbol, also known as rising factorial, given by
\begin{subequations}
\begin{equation}\label{risingFactorial}
    \define{(a)_n;\begin{cases}
        a(a+1)\cdots(a+n-\- 1),\quad&\text{if }n\in\N,\\
        1,\quad&\text{if }n=0.
    \end{cases}}
\end{equation}
It is easy to see that for any $n\in\No$
\begin{equation}\label{symbolPocchammer}
    (a)_n=\begin{cases}
        (-1)^n\+ n!\,\binom{\abs{a}}{\+\abs{a}-n\+},\quad&\text{if }\,a\in-\No\+,\\
        \frac{\Gamma(a\++\+n)}{\Gamma(a)},\quad&\text{otherwise.}
    \end{cases}
\end{equation}
\end{subequations}
In particular, notice that if $a\in-\No\+$, then $(a)_n=0$ for all $n>\abs{a}$. 
Next, we recall the definition of the Gauss hypergeometric function
\begin{equation}\label{hypergeometricGauss}
    {}_2 F_1(a,b\+;\+c\+;z)\vcentcolon= \!\sum_{k\+\in\+\No} \frac{(a)_k(b)_k}{(c)_k\!}\+\frac{z^k}{k!}.
\end{equation}
Representation~\eqref{hypergeometricGauss} is well defined for $a,b\in\C$, $c\in\C\smallsetminus\--\No$ and its radius of convergence is $1$.
However, if $a$ or $b$ is a non-positive integer, then the Gauss hypergeometric function reduces to a polynomial in $z$.
In this case, $c$ can also assume non-positive integer values, provided that $\abs{c}$ is greater than or equal to the degree of the polynomial. 
We also remind the Gauss' summation theorem
\begin{equation}\label{summationGauss}
        {}_2 F_1(a,b\+ ;\+ c\+ ;1)=\frac{\Gamma(c)\+\Gamma(c-a-b)}{\Gamma(c-a)\+\Gamma(c-b)}\+,\qquad \text{if }\, \Re(c-a-b)>0.
    \end{equation}
In the following lemma  we give the explicit computation of the integrals appearing in $ S_{\mathrm{off};\,\ell}$ and $ S_{\mathrm{reg};\,\ell}$ for $\ell$ even (see \eqref{offS}, \eqref{regS}).
\begin{lemma}\label{hypergeometricLegendre} For any $x\in[0,1], p\in\R$ and $\ell$ even we have \begin{equation*}
    \integrate[-1;1]{P_\ell(y)\+ \frac{\cosh\!\left[p\arcsin(xy)\right]}{\sqrt{1-x^2y^2}};\-dy}\- =\frac{2^{\ell+1}\ell\+ !\,x^\ell}{(2\ell\- +\!1)!}\prod_{k=1}^{\frac{\ell}{2}}\!\left[p^2\mspace{-5mu}+\!(2k\- -\mspace{-2.5mu}1)^2\right]\-  {}_2 F_1\!\- \left(\tfrac{\ell+1+ip}{2},\tfrac{\ell+1-ip}{2};\ell\!+\- \tfrac{3}{2};x^2\right)\!.
\end{equation*}
\begin{proof}
    First let $z\in(-1,1)$ and take into account ~\cite[p. 1007, 9.121.32]{GR}, so that
    \begin{equation}\label{hardTaylorExpansion}
    \begin{split}
        \frac{\cosh(p\arcsin{z})}{\sqrt{1-z^2}}&={}_2 F_1\!\left(\tfrac{1+ip}{2},\tfrac{1-ip}{2};\tfrac{1}{2}\+ ;z^2\right)\!=\!\- \sum_{k\+\in\+\No} \frac{\left(\frac{1+ip}{2}\right)_{\- k}\!\left(\frac{1-ip}{2}\right)_{\- k}}{\left(\frac{1}{2}\right)_{\- k}}\,\frac{\:z^{2k}\!\- }{k!}\\[-10pt]
        &=\!\- \sum_{k\+\in \+\No}\frac{\:z^{2k}\!\- }{(2k)!\!}\,\prod_{n=1}^k \!\- \left[\+ p^2\!+\- (2n-1)^2\right]\!,
    \end{split}
    \end{equation}
    where the last identity is given by the following simple computations
    \begin{gather}
        \left(\tfrac{1+ip}{2}\right)_{\- k}\!\left(\tfrac{1-ip}{2}\right)_{\- k}\! =\-\left\lvert\tfrac{1+ip}{2}\right\rvert^2\!\left\lvert\tfrac{1+ip+2}{2}\right\rvert^2\!\cdots\+ \left\lvert\tfrac{1+ip+2k-2}{2}\right\rvert^2\!=\frac{1}{\,2^{2k}\!\- }\,\prod_{n=1}^k\!\- \left[\+ p^2\!+\- (2n-1)^2\right]\!,\label{productPocchammer}\\
        \left(\tfrac{1}{2}\right)_{\- k}\- =\frac{1}{\!\- \sqrt{\pi\,}\,}\+ \Gamma\!\left(k+\tfrac{1}{2}\right)=\frac{1}{\,2^{2k}\!\- }\,\frac{(2k)!}{k!}\label{gammaHalfInteger}.
    \end{gather}
    Notice that~\eqref{gammaHalfInteger} is a particular case of the Legendre's duplication formula
    \begin{equation}\label{duplicationLegendre}
        \Gamma(z)\+\Gamma\!\left(z+\tfrac{1}{2}\right)=2^{1-2z}\sqrt{\pi}\+\Gamma(2z),\qquad z\in\C\smallsetminus\- -\tfrac{1}{2}\+\No.
    \end{equation}
    Using the Rodrigues’ formula for $P_\ell$ and integrating by parts $\ell$ times, one gets
    \begin{equation*}
        \integrate[-1;1]{P_\ell(y)\,\frac{\cosh[\+ p\arcsin(xy)]}{\sqrt{1-x^2y^2}};\-dy}=\frac{1}{2^\ell \ell\+ !}\integrate[-1;1]{(1-y^2)^\ell\frac{\partial^{\+\ell}}{\partial y^\ell\!\-} \,\frac{\cosh[\+ p\arcsin(xy)]}{\sqrt{1-x^2y^2}};\-dy}.
    \end{equation*}
    By~\eqref{hardTaylorExpansion}, the function $y\longmapsto \frac{\cosh[\+p\arcsin(xy)]}{\sqrt{1-x^2y^2}}$ is analytic in $(-1,1)$ for all $x\in[0,1]$ and $p\in\R$, thus one can compute the $\ell$-th derivative:
    \begin{equation*}
        \frac{\partial^{\+\ell}}{\partial y^\ell\!\-}\,\frac{\cosh\!\left[\+ p\arcsin(xy)\right]}{\sqrt{1-x^2y^2}}=\sum_{k=\frac{\ell}{2}}^{\pInfty}\frac{\:x^{2k}\!}{(2k)!\!\- }\;a_k(p^2)\,\frac{(2k)!\,y^{2k-\ell}}{(2k-\ell)!}=\sum_{k\+\in\+\No}\frac{\;x^{\ell+2k}\!\- }{(2k)!}\;a_{k+\frac{\ell}{2}}(p^2)\,y^{2k},
    \end{equation*}
    where we have set $a_k(p^2)\vcentcolon=\prod_{n=1}^k [\+p^2\!+\- (2n-1)^2]$ for the sake of notation.
    Using Tonelli's theorem to interchange the integral with the summation, one obtains
    \begin{equation*}
        \integrate[-1;1]{P_\ell(y)\,\frac{\cosh[\+ p\arcsin(xy)]}{\sqrt{1-x^2y^2}};\-dy}=\frac{x^\ell}{2^\ell \ell\+ !}\sum_{k\+\in\+\No}\!\frac{\:x^{2k}\!\- }{(2k)!\!}\;a_{k+\frac{\ell}{2}}(p^2)\!\integrate[-1;1]{(1-y^2)^\ell y^{2k};\-dy}.
    \end{equation*}
    The last integral can be explicitly computed, namely
    \begin{equation}\label{betaIntegralStep}
        \integrate[-1;1]{(1-y^2)^\ell\+  y^{2k};\-dy}=\frac{\Gamma(\ell+1)\+\Gamma\!\left(k+\tfrac{1}{2}\right)\!}{\Gamma\!\left(\ell+k+\tfrac{3}{2}\right)}= \frac{2^{2\ell+2}\,\ell\+ !\,(\ell+k+1)!\,(2k)!}{(2\ell+2k+2)!\,k!},
    \end{equation}
    where in the last equality we have used~\eqref{duplicationLegendre}.
    Therefore,
    \begin{equation*}
        \integrate[-1;1]{P_\ell(y)\,\frac{\cosh[\+ p\arcsin(xy)]}{\sqrt{1-x^2y^2}};\-dy}=2^{\ell+2}x^\ell\sum_{k\+\in\+\No}\frac{(\ell+k+1)!}{(2\ell+2k+2)!}\,a_{k+\frac{\ell}{2}}(p^2)\,\frac{\:x^{2k}\!}{k!}.
    \end{equation*}
    Using~\eqref{gammaHalfInteger} and~\eqref{productPocchammer}, the last expression can be rewritten in terms of the Pochhammer symbols
    \begin{align*}
    \integrate[-1;1]{P_\ell(y)\,\frac{\cosh[\+ p\arcsin(xy)]}{\sqrt{1-x^2y^2}};\-dy}&=\frac{x^\ell}{2^\ell}\sum_{k\+\in\+\No} \frac{\:x^{2k}\!}{k!}\,\frac{a_{k+\frac{\ell}{2}}(p^2)}{2^{2k}\-\left(\frac{1}{2}\right)_{\-\ell+k+1}}\\
    &=x^\ell\-\sum_{k\+\in\+\No} \frac{\:x^{2k}\!}{k!}\,\frac{\left(\frac{1+ip}{2}\right)_{\- k+\frac{\ell}{2}}\!\left(\frac{1-ip}{2}\right)_{\- k+\frac{\ell}{2}}}{\left(\frac{1}{2}\right)_{\-\ell+k+1}}.
    \end{align*}
    By definition~\eqref{risingFactorial}, one has
    \begin{equation}
        (\cdot)_{n+m}=(\cdot)_m\+ (\,\cdot+m)_n\+,\qquad \forall n,m\in\No\+,
    \end{equation}
    hence
    \begin{equation*}
        \integrate[-1;1]{P_\ell(y)\,\frac{\cosh[\+ p\arcsin(xy)]}{\sqrt{1-x^2y^2}};\-dy}=
        \frac{x^\ell\!\left(\frac{1+ip}{2}\right)_{\- \frac{\ell}{2}}\!\left(\frac{1-ip}{2}\right)_{\- \frac{\ell}{2}}}{\left(\frac{1}{2}\right)_{\-\ell+1}}\, \sum_{k\+\in\+\No}\frac{\:x^{2k}\!}{k!}\,\frac{\left(\frac{\ell+1+ip}{2}\right)_{\- k}\!\left(\frac{\ell+1-ip}{2}\right)_{\- k}}{\left(\ell+\frac{3}{2}\right)_{\-k}}.
    \end{equation*}
    Using again~\eqref{productPocchammer}, \eqref{gammaHalfInteger} and definition~\eqref{hypergeometricGauss} one concludes the proof.
    
    \end{proof}
\end{lemma}
\begin{note}
    We point out that the integral evaluated in lemma~\ref{hypergeometricLegendre} considerably simplifies in case $x=1$ or $p=0$.
    Indeed, making use of~\eqref{summationGauss} and~\eqref{gammaHalfInteger}, one gets
    \begin{align*}
        \integrate[-1;1]{\!P_\ell(y)\,\frac{\cosh(p\arcsin{y})}{\sqrt{1-y^2}};\-dy}=\frac{2^{\ell+1}\-\sqrt{\pi}\,\ell\+ !\,\Gamma\!\left(\ell+\frac{3}{2}\right)\!}{(2\ell\- +\!1)!\left\lvert\Gamma\!\left(\frac{\ell+2+ip}{2}\right)\- \right\rvert^2\!\-}\+\prod_{n=1}^{\frac{\ell}{2}}\!\left[p^2\mspace{-5mu}+\!(2n\- -\mspace{-2.5mu}1)^2\right]\\
        =\frac{2^{\ell}\-\sqrt{\pi}\,\ell\+ !\,\Gamma\!\left(\ell+\frac{1}{2}\right)\!}{(2\ell)!\left\lvert\Gamma\!\left(\frac{\ell+2+ip}{2}\right)\- \right\rvert^2\!\-}\,\prod_{n=1}^{\frac{\ell}{2}}\!\left[p^2\mspace{-5mu}+\!(2n\- -\mspace{-2.5mu}1)^2\right]=\frac{\pi}{2^\ell\left\lvert\Gamma\!\left(\frac{\ell+2+ip}{2}\right)\- \right\rvert^2\!\-}\,\prod\limits_{n=1}^{\frac{\ell}{2}} \!\left[\+ p^2\!+\- (2n\- -\!1)^2\right]\!.
    \end{align*}
    Now, exploiting the identity
    \begin{equation}\label{gammaComplexPlusInteger}
        \left\lvert\Gamma\!\left(n+\-1\-+ib\right)\- \right\rvert^2\!=\frac{\pi b}{\sinh\!\left(\pi b\right)}\, \prod_{k=1}^n \+ (k^2\!+b^2),\qquad \forall b\in\R,\,n\in\No\+,
    \end{equation}
    one obtains
    \begin{equation}\label{hypergeometricLegendreXOne}
        \integrate[-1;1]{\!P_\ell(y)\,\frac{\cosh(p\arcsin{y})}{\sqrt{1-y^2}};\-dy}=\frac{2\sinh\!\left(\frac{\pi}{2}p\right)\!}{p}\prod_{k=1}^{\frac{\ell}{2}}\frac{p^2\!+\- (2k\- -\!1)^2\!}{p^2+4k^2}.
    \end{equation}
    Let us consider the case $p=0$.
    Taking into account that
    \begin{equation*}
        \prod_{k=1}^{\frac{\ell}{2}}(2k-1)^2=(\ell-1)!!^{\+2}=\frac{\ell\+ !^{\+ 2}}{2^\ell\- \left(\frac{\ell}{2}\right)!^{\+ 2}\!\-}\:,
    \end{equation*}
    where $(\cdot)!!$ denotes the double factorial, i.e.
    \begin{equation}\label{doubleFactorial}
        n!!\vcentcolon=\!\prod_{k=0}^{\left\lceil\! \frac{n}{2}\!\right\rceil-1}\!\-(n-2k)=
        \begin{dcases}
            2^{\frac{n}{2}}\!\left(\tfrac{n}{2}\right)!\+ ,\quad&\text{if }n\text{ is even},\\
            \frac{(n\- +\!1)!}{2^{\frac{n+1}{2}}\!\left(\frac{n+1}{2}\right)!\! }\,,\quad&\text{if }n\text{ is odd},
        \end{dcases}
    \end{equation}
    one obtains
    \begin{equation}\label{hypergeometricLegendrePZero}
        \integrate[-1;1]{\frac{P_\ell(y)}{\sqrt{1-x^2y^2}\,};\-dy}=\frac{2\+ x^\ell\, \ell\+ !^{\+ 3}}{(2\ell+1)!\left(\frac{\ell}{2}\right)!^{\+ 2}\!}\; {}_{2}F_1\!\left(\tfrac{\ell+1}{2},\tfrac{\ell+1}{2};\ell+\tfrac{3}{2}\+ ;x^2\right)\!.
    \end{equation}
    In the special case $x=1$ and $p=0$, one has
    \begin{equation}\label{hypergeometricLegendrePZeroXOne}
        \integrate[-1;1]{\frac{P_\ell(y)}{\sqrt{1-y^2}\,};\-dy}=\frac{2\sqrt{\pi}\+\ell\+ !^{\+ 3}\,\Gamma\!\left(\ell+\frac{3}{2}\right)\!}{(2\ell+1)!\left(\frac{\ell}{2}\right)!^{\+ 4}\!}=\frac{\sqrt{\pi}\+\ell\+ !^{\+ 3}\,\Gamma\!\left(\ell+\frac{1}{2}\right)\!}{(2\ell)!\left(\frac{\ell}{2}\right)!^{\+ 4}\!}=\frac{\pi\+\ell\+ !^{\+ 2}}{2^{2\ell}\left(\frac{\ell}{2}\right)!^{\+ 4}\!}
    \end{equation}
    where we have used~\eqref{summationGauss} and~\eqref{gammaHalfInteger}.
\end{note}

\n
Now, let $\Xi^{\,\zeta}_{\+\ell,\,s_\ell}$ be a sequence of auxiliary quadratic forms defined on $\Lp{2}[\Rplus,\,p^2dp]$ for any given $\ell\in\No$ and for some parameter $s_\ell\in(0,1)$ as follows
\begin{equation}\label{auxiliaryQF}
 \Xi^{\,\zeta}_{\+\ell,\,s_\ell}\- \vcentcolon=s_\ell F^{\,\zeta}_{\mathrm{diag}}\!+F^{\+0}_{\mathrm{off};\,\ell}\-+F_{\mathrm{reg};\,\ell}\+,\qquad\dom{\Xi^{\,\zeta}_{\+\ell,\,s_\ell}}=\Lp{2}[\Rplus,\,p^2\sqrt{p^2+1}\,dp]\+.
\end{equation}
These quadratic forms will be useful to obtain a lower bound for $F^{\,\zeta}_\ell$.

\n The next lemma is the key technical ingredient for the proof of proposition~\ref{coercivityLemma}.
\begin{lemma}\label{auxiliaryPositivity}
    Let $\psi\in\Lp{2}(\Rplus,\,p^2\sqrt{p^2+1}\,dp)$ and $\gamma_c$ given by~\eqref{criticalGamma}.
    Then, for $\gamma\->\-\gamma_c\,$, there exists $\{s^\ast_\ell\}_{\ell\+\in\+\No}\!\- \subset\!(0,1)$ such that each quadratic form $\Xi^{\,\zeta}_{\+\ell,\,s^\ast_\ell}$ defined by~\eqref{auxiliaryQF}, is non-negative for any $\zeta\geq 0$ and $\ell\in\No\+$.
\begin{proof}
    Taking into account the diagonalization derived in lemma~\ref{offregDiagonalization}, one has
    \begin{align*}
        \Xi^{\,\zeta}_{\+\ell,\,s_\ell}[\psi]\- \geq \!\left(s_\ell F^{\+0}_{\mathrm{diag}}\! +\-F^{\+0}_{\mathrm{off};\,\ell}\- +\-F_{\mathrm{reg};\,\ell} \right)\![\psi]\- =\! &\integrate[\R]{\abs{g_\psi(p)}^2;dp}\!\left[s_\ell\sqrt{\- \tfrac{\mu}{\eta}}\- +\- \tfrac{N-1}{2}\left(S_{\mathrm{off};\,\ell}\- +\- S_{\mathrm{reg};\,\ell}\right)\!(p)\right]\\
        =\vcentcolon \!\- &\integrate[\R]{\abs{g_\psi(p)}^2f^{N}_{\ell,\,s_\ell}(p);dp}.
    \end{align*}
    The lemma is proved if we show that for each order $\ell$, there exists $s_\ell\in(0,1)$ such that the function $f^N_{\ell,\,s_\ell}$ is non-negative uniformly in $N\-\geq\- 2$.
    Notice that this is actually the case for $\ell$ odd, in light of~\eqref{SregoffSigns}, so from now on we focus on the case $\ell$ even.
    
    \n Moreover we have
    \begin{equation*}
        \lim_{p\to\pInfty} f^N_{\ell,\,s_\ell}(p)=s_\ell\sqrt{\tfrac{\mu}{\eta}}>0.
    \end{equation*}
    We notice that  $S_{\mathrm{off};\,\ell}$  and $S_{\mathrm{reg};\,\ell}$, and then $ f^{N}_{\ell,\,s_\ell}$, are written in terms of the Gauss hypergeometric function ${}_{2}F_1$ (see \eqref{offS}, \eqref{regS} and lemma \ref{hypergeometricLegendre}) and therefore the main point is a careful control of such a function. 
    
    \n
    The proof  will be constructed in two steps: first we show that $f^N_{\ell,\,s_\ell}$ evaluated at zero is positive uniformly in $N\-\geq\- 2$ for a proper choice of $\{s_\ell\}_{\ell\+\in\+\No}\-\subset\-(0,1)$, then we prove that $f^N_{\ell,\,s_\ell}$ is bounded from below by a monotonic function $h^N_{\ell,\,s_\ell}$ that shares the same values with $f^N_{\ell,\,s_\ell}$ at zero and infinity.
    Once these statements are proven, we will have $f^N_{\ell,\,s_\ell}\geq h^N_{\ell,\,s_\ell}>0$ as long as $s_\ell$ is such that $f^N_{\ell,\,s_\ell}(0)>0$ for all $\ell\-\in\-\No$ and uniformly in $N\geq 2$.
    
    \vs
    \n\emph{Step 1.}\quad
    We observe that $f^N_{\ell,\,s_\ell}(0)$ is positive if and only if
    \begin{equation}\label{s-ellCondition}
    \begin{split}
        s_\ell&>-\tfrac{N-1}{2}\sqrt{\tfrac{\eta}{\mu}}\left(S_{\mathrm{off};\,\ell}+S_{\mathrm{reg};\,\ell}\right)\!(0)\\[-5pt]
        &=\tfrac{N-1}{2}\sqrt{\tfrac{\eta}{\mu}}\left[\integrate[-1;1]{\frac{P_\ell(y)}{\!\-\sqrt{1-\frac{y^2}{(1+M)^2\mspace{-9mu}}\,}\,};\-dy}-\gamma\!\integrate[-1;1]{\frac{P_\ell(y)}{\!\-\sqrt{1\--y^2}\,};\-dy}\right]\!.
    \end{split}
    \end{equation}
    The requirement $s_\ell\-\in\-(0,1)$ implies a constraint for the parameter $\gamma$, since we need the right hand side of~\eqref{s-ellCondition} to be strictly less than $1$.
    Therefore
    \begin{equation}\label{gamma-ellCriticalDef}
        \gamma>\gamma_M^\ell\vcentcolon=\left[\integrate[-1;1]{\frac{P_\ell(y)}{\!\-\sqrt{1\--y^2}\,};\-dy}\right]^{-1}\!\left[\integrate[-1;1]{\frac{P_\ell(y)}{\!\-\sqrt{1-\frac{y^2}{(1+M)^2\mspace{-9mu}}\,}\,};\-dy}-\frac{2}{N\!-\!1}\sqrt{\frac{\mu}{\eta}}\+\right]\!.
    \end{equation}
    Let us show that
    \begin{equation}\label{gammaDecreasesToShow}
        \gamma_c=\max_{k\+\in\+\No}\{\gamma^{2k}_M\}\-=\gamma_M^0.
    \end{equation}
    Taking into account  equations~\eqref{hypergeometricLegendrePZero} and~\eqref{hypergeometricLegendrePZeroXOne}, condition~\eqref{gamma-ellCriticalDef} reads
    \begin{equation*}
        \gamma>\gamma_M^\ell\-=\gamma^\ell_{M,\+1}\--\+\gamma^\ell_{M,\+2}\+,
    \end{equation*}
    with
    \begin{gather}\label{fractionMassLegendre}
        \define{\gamma_{M,\+1}^\ell;}\frac{2^{2\ell+1}\ell\+ !\- \left(\frac{\ell}{2}\right)!^{\+ 2}}{\pi\,(2\ell+1)!\,(M\-+\-1)^\ell\!\-}\;\,{}_2 F_1\!\left(\tfrac{\ell+1}{2},\tfrac{\ell+1}{2};\ell+\tfrac{3}{2}\+ ;\tfrac{1}{(M+1)^2\mspace{-9mu}}\,\right)\!,\\
        \define{\gamma_{M,\+2}^\ell;\frac{2^{2\ell+1}\!\left(\frac{\ell}{2}\right)!^{\+ 4}\sqrt{M(M\-+\-2)}}{\pi\,\ell\+ !^{\+ 2}(N\!-\!1)(M\-+\-1)}}.
    \end{gather}
    We observe that $\gamma_{M,\+2}^\ell$ is increasing in $\ell$, since
    \begin{equation*}
        \frac{2^{2(\ell+2)}\!\left(\frac{\ell+2}{2}\right)!^{\+ 4}\!\-}{(\ell+2)!^{\+ 2}}=\frac{2^{2\ell+4}\!\left(\frac{\ell}{2}+\mspace{-1mu}1\right)!^{\+ 4}\!\-}{(\ell+2)!^{\+ 2}}=\frac{2^{2\ell}\,2^4\!\left(\frac{\ell}{2}+\mspace{-1mu}1\right)^4\!\left(\frac{\ell}{2}\right)!^{\+ 4}\!\-}{(\ell+2)^2(\ell+1)^2\+ \ell\+ !^{\+ 2}}=\frac{2^{2\ell}\!\left(\ell+2\right)^2\!\left(\frac{\ell}{2}\right)!^{\+ 4}\!\-}{(\ell+1)^2\+ \ell\+ !^{\+ 2}}>\frac{2^{2\ell}\!\left(\frac{\ell}{2}\right)!^{\+ 4}\!\-}{\ell\+ !^{\+ 2}}\+.
    \end{equation*}
    Therefore
    \begin{equation}\label{gammaSecondTermOk}
        \gamma^\ell_M<\gamma^\ell_{M,\+1}\--\gamma^0_{M,\+2}=\gamma^\ell_{M,\+1}\--\frac{2}{\pi\+(N\!-\!1)}\+\frac{\sqrt{M(M\-+\-2)}}{M\-+\-1}.
    \end{equation}
    Let us consider $\gamma^\ell_{M,\+1}$.
Using the Euler's integral representation of the Gauss hypergeometric function
    \begin{equation}
        {}_2 F_1(a,b\+;\+c\+;z)=\frac{\Gamma(c)}{\Gamma(b)\+\Gamma(c-b)}\integrate[0;1]{\frac{t^{b-1}(1-t)^{c-b-1}}{(1-z\+t)^a};dt},\qquad \Re(c)>\Re(b)>0,
    \end{equation}
    one has for any $x\in[0,1]$
    \begin{align*}
        {}_2 F_1\!\left(\tfrac{\ell+1}{2},\tfrac{\ell+1}{2};\ell+\tfrac{3}{2}\+ ;x^2\right)\!=\frac{\Gamma\!\left(\ell+\frac{3}{2}\right)}{\Gamma\!\left(\frac{\ell+1}{2}\right)\!\-\left(\frac{\ell}{2}\right)!}\-\integrate[0;1]{\frac{t^{\frac{\ell-1}{2}}(1-t)^{\frac{\ell}{2}}\!\-}{(1-x^2t)^{\frac{\ell+1}{2}}};dt}
        =\frac{2\,\Gamma\!\left(\ell+\frac{3}{2}\right)}{\Gamma\!\left(\frac{\ell+1}{2}\right)\!\-\left(\frac{\ell}{2}\right)!}\-\integrate[0;1]{\frac{u^\ell(1-u^2)^{\frac{\ell}{2}}}{(1-x^2u^2)^{\frac{\ell+1}{2}}\!\-};du}\\
        =\frac{2^\ell(2\ell+1)\,\Gamma\!\left(\ell+\frac{1}{2}\right)}{\sqrt{\pi}\,\ell\+!}\!\integrate[0;1]{\frac{u^\ell(1-u^2)^{\frac{\ell}{2}}}{(1-x^2u^2)^{\frac{\ell+1}{2}}\!\-};du}=\frac{(2\ell\- +\- 1)!}{2^\ell\,\ell\+!^{\+ 2}}\!\integrate[0;1]{\frac{u^\ell(1-u^2)^{\frac{\ell}{2}}}{(1-x^2u^2)^{\frac{\ell+1}{2}}};du}.
    \end{align*}
    Exploiting the inequality $1-u^2\leq 1-x^2u^2$, we obtain an estimate from above
    \begin{equation}\label{hypergeometricUpperBound}
        {}_2F_1\!\left(\tfrac{\ell+1}{2},\tfrac{\ell+1}{2};\ell+\tfrac{3}{2}\+ ;x^2\right)\!\leq\frac{(2\ell\- +\- 1)!}{2^\ell\,\ell\+ !^{\+ 2}}\!\integrate[0;1]{\frac{u^\ell}{\sqrt{1-x^2u^2}\,};du},
    \end{equation}
    where equality holds if $\ell=0 \lor x=1$.
    From~\eqref{hypergeometricUpperBound} one gets
    \begin{equation}\label{gammasInequality}
        \gamma_{M,\+1}^\ell\leq\frac{\,2^{\ell+1}\!\left(\frac{\ell}{2}\right)!^{\+ 2}}{\pi\, \ell\+ !\,(M\- +\- 1)^\ell\mspace{-6mu}}\+ \integrate[0;1]{\frac{u^\ell}{\!\-\sqrt{1-\frac{u^2}{(1+M)^2\mspace{-9mu}}\,}\,};du}\define*{;\bar{\gamma}^\ell_M}\+,
    \end{equation}
    where equality holds if $\ell=0\lor M=0$.
    Hence, in particular we know that
    \begin{equation}
        \gamma_{M,\+1}^0=\bar{\gamma}^0_M=\tfrac{2(M+1)}{\pi}\arcsin\!\left(\tfrac{1}{M+1}\right)\!.\label{equalityHoldsAtZero}
    \end{equation}
    In the following computations we set $x=\frac{1}{M+1}$ for the sake of notation.
    Let us prove that $\left\{\bar{\gamma}^\ell_M\-\right\}_{\ell\+\in\+2\No}\-$ is a decreasing sequence for all fixed $M\->\-0$. We have
    \begin{align*}
        \bar{\gamma}^{\ell}_M-\bar{\gamma}^{\ell+2}_M&=\frac{2^{\ell+1} x^\ell\!\left(\frac{\ell}{2}\right)!^{\+ 2}}{\pi\, \ell\+ !}\!\integrate[0;1]{\frac{u^\ell}{\sqrt{1-x^2u^2}}\!\left[1-\frac{4x^2\!\left(\frac{\ell}{2}+1\right)^{\- 2}\!u^2}{(\ell+2)(\ell+1)}\right];du}\\
        &=\frac{2^{\ell+1} x^\ell\!\left(\frac{\ell}{2}\right)!^{\+ 2}}{\pi\, \ell\+ !}\!\integrate[0;1]{\frac{u^\ell}{\sqrt{1-x^2u^2}}\!\left[1-\frac{\left(\ell+2\right)\- x^2u^2}{\ell+1}\right];du}\!.
    \end{align*}
    Our goal is to show that the last integral is positive for any given $x\in(0,1)$ and $\ell$ even, so that $\bar{\gamma}_M^{\ell+2}\!<\bar{\gamma}_M^\ell\+$.
    To this end, we first point out that the integral is manifestly positive at $x=0$, whereas the evaluation of the integral at $x=1$ yields
    \begin{equation*}
        \integrate[0;1]{\frac{2u^\ell}{\sqrt{1-u^2}}\!\left[1-\frac{\ell+2}{\ell+1} \,u^2\right];du}=\frac{\pi\, \ell\+ !}{2^\ell\! \left(\frac{\ell}{2}\right)!^{\+ 2}}-\frac{\ell+2}{\ell+1}\,\frac{\pi\, (\ell+2)!}{2^{\ell+2}\! \left(\frac{\ell}{2}+1\right)!^{\+ 2}}=0.
    \end{equation*}
We observe that, in order to obtain $\inf\!\left\{\bar{\gamma}^\ell_M\--\bar{\gamma}^{\ell+2}_M\,|\;M\->\-0\right\}\!\geq 0$, it is sufficient to prove that the integral is a monotonic decreasing function in $x$ for any $\ell$.
    In other words, we want to show
    \begin{equation}\label{negativeIntegralWTS}
        \frac{\mathrm{d}}{\mathrm{d}x}\!\integrate[0;1]{\frac{u^\ell}{\sqrt{1-x^2u^2}\+}\!\left[1-\frac{\ell+2}{\ell+1} \,x^2u^2\right];du}\!<0,\quad \forall x\in(0,1),\,\ell \text{ even}.
    \end{equation}
    By the Leibniz integral rule, the derivative with respect to $x$ can be computed inside the integral.  Therefore, for any $x\in(0,1), u\in[0,1]$ and $\ell$ even, one has
    \begin{align*}
        \frac{\partial}{\partial x}\,\frac{u^\ell}{\sqrt{1-x^2u^2}\+}\!\left[1-\frac{\ell+2}{\ell+1} \,x^2u^2\right]\!=\frac{x\,u^{\ell+2}}{(1-x^2u^2)^{\frac{3}{2}}\mspace{-10.5mu}}\:\left[1-\frac{\ell+2}{\ell+1}\,x^2u^2\right]\!-\frac{\ell+2}{\ell+1}\,\frac{2\+x\,u^{\ell+2}}{\sqrt{1-x^2u^2}\+}\\
        =\frac{x\,u^{\ell+2}}{(1-x^2u^2)^{\frac{3}{2}}\mspace{-10.5mu}}\:\left[1-\frac{\ell+2}{\ell+1}\left(x^2u^2+2-2\+x^2u^2\right)\right]\!
        =\frac{x\,u^{\ell+2}}{(1-x^2u^2)^{\frac{3}{2}}\mspace{-10.5mu}}\:\left[\frac{\ell+2}{\ell+1}\,x^2u^2-\frac{\ell+3}{\ell+1}\right]\!<0.
    \end{align*}
    Since the integral of a negative function obviously yields a negative quantity,~\eqref{negativeIntegralWTS} is proven.
    This means that $\!\left\{\bar{\gamma}^\ell_M\-\right\}_{\ell\+\in\+2\No}\!$ is decreasing for any fixed $M\->\-0$. Thus, taking into account \eqref{gammasInequality} and~\eqref{equalityHoldsAtZero}, we finally get
    \begin{equation*}
        \gamma^\ell_{M,\+1}\leq\bar{\gamma}^\ell_M\leq\bar{\gamma}^0_M=\gamma^0_{M,\,1}\+,\qquad \forall \ell \text{ even}.
    \end{equation*}
    Hence, thanks to~\eqref{gammaSecondTermOk}, equation~\eqref{gammaDecreasesToShow} is proved.
    
    \vs
    \n \emph{Step 2.}\quad Let us define the following function
    \begin{equation}
        h^N_{\ell,\,s_\ell}(p)\vcentcolon=s_\ell\,\frac{\sqrt{M(M\!+\- 2)}}{M+1}+(N\! -\- 1)(\gamma-\gamma^\ell_{M,\+1})\,\frac{\tanh\!\left(\frac{\pi}{2}p\right)}{p}\prod_{k=1}^{\frac{\ell}{2}}\frac{p^2+(2k\- -\- 1)^2\mspace{-6mu}}{p^2+4k^2}\,,
    \end{equation}
    where $\gamma^\ell_{M,\+1}$ has been defined in~\eqref{fractionMassLegendre}.
    We shall prove that $h^N_{\ell,\,s_\ell}$ satisfies 
    \begin{subequations}\label{hPositivityConditions}
    \begin{gather}
        h^N_{\ell,\,s_\ell}\leq f^N_{\ell,\,s_\ell}\+ ,\label{hPositivityConditionA}\\
        h^N_{\ell,\,s_\ell}(0)=f^N_{\ell,\,s_\ell}(0),\label{hPositivityConditionB}\\
        \lim_{p\to\pInfty}h^N_{\ell,\,s_\ell}(p)=\!\-\lim_{p\to\pInfty}f^N_{\ell,\,s_\ell}(p),\label{hPositivityConditionC}\\
        h^N_{\ell,\,s_\ell}(p)\text{ is monotonic in } p\in\Rplus.\label{hPositivityConditionD}
    \end{gather}
    \end{subequations}
    
    \n Starting with~\eqref{hPositivityConditionA},
    we take into account  lemma~\ref{hypergeometricLegendre} and equation~\eqref{hypergeometricLegendreXOne} to obtain an explicit expression for $f^N_{\ell,\,s_\ell}$ 
    \begin{equation}\label{fExplicitFlag}
        \begin{split}
        f^N_{\ell,\,s_\ell}(p)&=s_\ell\,\frac{\sqrt{M(M+2)}}{M+1}+\frac{N\- -\!1}{2}\left(S_{\mathrm{off};\,\ell}+S_{\mathrm{reg};\,\ell}\right)\!(p)\\[-5pt] &=s_\ell\,\frac{\sqrt{M(M+2)}}{M+1}+\frac{(N\- -\!1)\+ \bar{h}_\ell(p)}{\cosh\!\left(\frac{\pi}{2}p\right)}\+ \prod_{k=1}^{\frac{\ell}{2}}\!\left[p^2+(2k\- -\!1)^2\right]\!,
        \end{split}
    \end{equation}
    where we have introduced, for the sake of notation, the function
    \begin{equation}
        \bar{h}_\ell(p)\vcentcolon=\gamma\,\frac{\sinh\!\left(\frac{\pi}{2}p\right)}{p}\prod_{k=1}^{\frac{\ell}{2}}\frac{1}{p^2+4k^2}-\frac{2^\ell\,\ell\+ !\,{}_2 F_1\!\left(\frac{\ell+1+ip}{2},\frac{\ell+1-ip}{2};\ell+\frac{3}{2}\+ ;\frac{1}{(1+M)^2\mspace{-9mu}}\,\right)\!}{(2\ell\- +\!1)!\,(M\!+\!1)^\ell}\+.
    \end{equation}
    To achieve the result, consider the Euler's transformation formula
    \begin{equation}\label{transformationEuler}
        {}_2 F_1(a,b\+ ;c\+ ;z)=(1-z)^{c-a-b}{}_2F_1(c-a,c-b\+ ;c\+ ;z)
    \end{equation}
    and the inequality
    \begin{equation}\label{gammaComplex-RealInequality}
        \left\lvert\Gamma\!\left(a+ib\right)\- \right\rvert^2\!\leq \left\lvert\Gamma\!\left(a\right)\- \right\rvert^2,\quad \forall a,b\in\R\+.
    \end{equation}
    Indeed, one can write
    \begin{align*}
        {}_2F_1\!\left(\tfrac{\ell+1+ip}{2},\right.\!\!&\left.\tfrac{\ell+1-ip}{2};\ell+\tfrac{3}{2}\+ ;x^2\right)\!=\sqrt{1-x^2\,}\+  {}_2F_1\!\left(\tfrac{\ell+2-ip}{2},\tfrac{\ell+2+ip}{2};\ell+\tfrac{3}{2}\+ ;x^2\right)\\
        &=\sqrt{1-x^2\,}\sum_{k\+\in\+\No}\frac{\;x^{2k}\!\- }{k!}\,\frac{\left(\frac{\ell+2-ip}{2}\right)_{\- k}\!\left(\frac{\ell+2+ip}{2}\right)_{\- k}}{\left(\ell+\frac{3}{2}\right)_{\- k}}\leq \frac{\sqrt{1-x^2\,}\left(\frac{\ell}{2}\right)!^{\+ 2}\!\- }{\left\lvert\Gamma\!\left(\frac{\ell+2+ip}{2}\right)\- \right\rvert^2}\sum_{k\in\No}\frac{\;x^{2k}\!\- }{k!}\,\frac{\left(\frac{\ell}{2}\- +\!1\right)^{\- 2}_{\- k}}{\left(\ell+\frac{3}{2}\right)_{\- k}}\\
        &=\frac{\sqrt{1-x^2\,}\left(\frac{\ell}{2}\right)!^{\+ 2}\!\- }{\left\lvert\Gamma\!\left(\frac{\ell+2+ip}{2}\right)\- \right\rvert^2}\,{}_2F_1\!\left(\tfrac{\ell}{2}+1,\tfrac{\ell}{2}+1;\ell+\tfrac{3}{2}\+ ;x^2\right)\!,
    \end{align*}
    where we have used inequality~\eqref{gammaComplex-RealInequality}, according to which
    \begin{equation*}
        \frac{\left\lvert\Gamma\!\left(\frac{\ell+2+ip}{2}+\-k\right)\- \right\rvert^2\!\- }{\,\left\lvert\Gamma\!\left(\frac{\ell+2+ip}{2}\right)\- \right\rvert^2}\leq\frac{\Gamma^{\+ 2}\!\left(\frac{\ell}{2}\- +\!1\right)}{\left\lvert\Gamma\!\left(\frac{\ell+2+ip}{2}\right)\- \right\rvert^2\mspace{-6mu}}\;\frac{ \Gamma^{\+ 2}\!\left(\tfrac{\ell}{2}\-+\!1\-+\-k\right)\!}{\Gamma^{\+ 2}\!\left(\frac{\ell}{2}\- +\!1\right)}=\frac{\left(\frac{\ell}{2}\right)!^{\+ 2}}{\left\lvert\Gamma\!\left(\frac{\ell+2+ip}{2}\right)\- \right\rvert^2\mspace{-6mu}}\;\left(\tfrac{\ell}{2}+\-1\right)^{\-2}_{\-k}\-.
    \end{equation*}
    Using again~\eqref{transformationEuler} to the right hand side, one obtains
    \begin{equation}\label{p-behaviorExtractHypergeometric}
        {}_2F_1\!\left(\tfrac{\ell+1+ip}{2},\tfrac{\ell+1-ip}{2};\ell+\tfrac{3}{2}\+ ;x^2\right)\!\leq\frac{\left(\frac{\ell}{2}\right)!^{\+ 2}\!\- }{\left\lvert\Gamma\!\left(\frac{\ell+2+ip}{2}\right)\- \right\rvert^2\mspace{-6mu}}\:\:{}_2F_1\!\left(\tfrac{\ell+1}{2},\tfrac{\ell+1}{2};\ell+\tfrac{3}{2}\+ ;x^2\right)\!.
    \end{equation}
    Making use of identity~\eqref{gammaComplexPlusInteger} in the previous inequality, one has
    \begin{equation*}
    \begin{split}
        \bar{h}_\ell(p)&\geq\frac{\sinh\!\left(\frac{\pi}{2}p\right)}{p}\prod_{k=1}^{\frac{\ell}{2}}\frac{1}{p^2+4k^2}\!\left[\gamma-\frac{2^{2\ell+1}\,\ell\+ !\left(\frac{\ell}{2}\right)!^{\+ 2}}{\pi\,(2\ell\- +\!1)!\,(M\!+\!1)^\ell}\,{}_2 F_1\!\left(\tfrac{\ell+1}{2},\tfrac{\ell+1}{2};\ell+\tfrac{3}{2}\+ ;\tfrac{1}{(1+M)^2\mspace{-9mu}}\,\right)\!\right]\\
        &=\frac{\sinh\!\left(\frac{\pi}{2}p\right)}{p}\,(\gamma-\gamma^\ell_{M,\+1})\prod_{k=1}^{\frac{\ell}{2}}\frac{1}{p^2+4k^2}\,.
    \end{split}
    \end{equation*}
    Exploiting this lower bound in~\eqref{fExplicitFlag}, one finds out that $h^N_{\ell,\,s_\ell}$ satisfies condition~\eqref{hPositivityConditionA}.
    Furthermore, we stress that we have obtained this estimate by using only inequality~\eqref{gammaComplex-RealInequality}, according to which the equality sign holds in case $p\-=\-0$.
    In other words, we have also proved~\eqref{hPositivityConditionB}.

    \n Next, we show~\eqref{hPositivityConditionC}.
    Since $\frac{p^2+(2k-1)^2}{p^2+4k^2}<1$ for all $k$,
    \begin{equation*}
        \left\lvert h^N_{\ell,\,s_\ell}(p)\--s_\ell\,\frac{\sqrt{M(M+2)}}{M+1}\right\rvert\leq (N-1)\abs{\gamma-\gamma^\ell_{M,\+1}}\,\frac{\tanh\!\left(\frac{\pi}{2}p\right)}{p}
    \end{equation*}
    where the right hand side vanishes as $p$ goes to infinity.
    Therefore,
    \begin{equation*}
        \lim_{p\to\pInfty}h^N_{\ell,\,s_\ell}(p)=s_\ell\,\frac{\sqrt{M(M+2)}}{M+1}=\lim_{p\to\pInfty}f^N_{\ell,\,s_\ell}(p).
    \end{equation*}
    It remains to prove the monotonicity of $h^N_{\ell,\,s_\ell}$ in $\Rplus\+$.
    In particular, it suffices to show that the function
    \begin{equation}\label{functionToBeDecreasing}
        p\longmapsto\frac{\tanh\!\left(\frac{\pi}{2}p\right)\-}{p}\prod\limits_{k=1}^{\ell/2}\frac{p^2\-+(2k\- -\!1)^2\!\-}{p^2\-+4k^2}
    \end{equation}
    is decreasing in $\Rplus\+$.
    Let us remind the product representation of the hyperbolic tangent 
    \begin{equation}
        \tanh\!\left(z\right)=z\prod_{k\+\in\+\N}\frac{1+\frac{z^2}{\pi^2k^2\!\- }}{1+\frac{4z^2}{\pi^2(2k-1)^2\mspace{-6mu}}}\,.
    \end{equation}
    Denoting $z=\frac{\pi}{2}p$, one has 
    \begin{align*}
        \frac{\tanh\!\left(\frac{\pi}{2}p\right)}{p}\prod_{k=1}^{\frac{\ell}{2}}\frac{p^2\!+\- (2k\!-\!1)^2}{p^2\!+\- 4k^2}=\frac{\pi}{2}\prod_{k=1}^{\frac{\ell}{2}}\frac{1+\frac{p^2}{4k^2\!\- }}{1+\frac{p^2}{(2k-1)^2\mspace{-6mu}}}\,\frac{p^2\!+\- (2k\!-\!1)^2}{p^2\!+\- 4k^2}\prod_{k=\frac{\ell}{2}+1}^{\pInfty}\frac{1+\frac{p^2}{4k^2\!\- }}{1+\frac{p^2}{(2k-1)^2\mspace{-6mu}}}\\
        =\frac{\pi}{2}\prod_{k=1}^{\frac{\ell}{2}}\frac{(2k\!-\!1)^2\!\- }{4k^2}\prod_{k=\frac{\ell}{2}+1}^{\pInfty}\frac{1+\frac{p^2}{4k^2\!\- }}{1+\frac{p^2}{(2k-1)^2\mspace{-6mu}}}=\frac{\pi}{2}\,\frac{(\ell-1)!!^{\+ 2}\!\- }{\ell\+ !!^{\+ 2}}\prod_{k=\frac{\ell}{2}+1}^{\pInfty}\frac{1+\frac{p^2}{4k^2\!\- }}{1+\frac{p^2}{(2k-1)^2\mspace{-6mu}}}\,.
    \end{align*}
    In order to prove that the function~\eqref{functionToBeDecreasing} is decreasing, we consider
    \begin{equation}\label{lnFunctionToBeDecreasing}
    \begin{split}
    \ln\!\left(\frac{\tanh\!\left(\frac{\pi}{2}p\right)}{p}\prod_{k=1}^{\frac{\ell}{2}}\frac{p^2\!+\- (2k\!-\!1)^2\!\-}{p^2\!+\- 4k^2}\+\right)\!\-=&\,\ln\!\left(\frac{\pi}{2}\right)\!+2\ln\!\left[\frac{(\ell-1)!!}{\ell\+ !!}\right]\!+\\[-10pt]
    &+\mspace{-6mu}\sum_{k=\frac{\ell}{2}+1}^{\pInfty}\!\!\ln\!\left(1+\frac{p^2}{4k^2}\right)\!-\ln\!\left[1+\frac{p^2}{(2k-1)^2\mspace{-6mu}}\,\right]\!.
    \end{split}
    \end{equation}
    We notice that for $p>0$
    \begin{align*}
    \frac{\partial}{\partial p}\left\{\ln\!\left(1+\frac{p^2}{4k^2}\right)\!-\ln\!\left[1+\frac{p^2}{(2k-1)^2\mspace{-6mu}}\,\right]\!\right\}&=\frac{2p}{p^2+4k^2}-\frac{2p}{p^2+(2k-1)^2\mspace{-6mu}}\\
    &=\frac{2p\,(1-4k)}{(p^2+4k^2)[p^2+(2k-1)^2]}<0,\quad \forall k\geq 1.
    \end{align*}
    Hence,~\eqref{lnFunctionToBeDecreasing} is decreasing in $p\->\-0$ since it is a sum of decreasing functions.
    Therefore, also~\eqref{functionToBeDecreasing} is decreasing and~\eqref{hPositivityConditionD} is proven.
    
    \n In conclusion, we know that whenever $\gamma\->\-\gamma_c\+$, there exists $s^\ast_\ell\in(0,1)$ for any $\ell\in\No\+$, such that $f^N_{\ell,\,s^\ast_\ell}(0)> 0$ uniformly in $N\geq 2$.
    Since we also know that $f^N_{\ell,\,s_\ell}$ is eventually positive, conditions~\eqref{hPositivityConditions} imply that $f^N_{\ell,\,s^\ast_\ell}\geq h^N_{\ell,\,s^\ast_\ell}\!> 0$ and the proof is completed.

\end{proof}
\end{lemma}
\begin{note}\label{partialWaveUniformity}
    In lemma~\ref{auxiliaryPositivity}, we have shown that, if $\gamma\-\geq\-\gamma^\ell_{M,\+1}\+$, any $s^\ast_\ell\-\in\-(0,1)$ is such that $f^N_{\ell,\,s^\ast_\ell}\geq 0$, whereas in case $\gamma\-\in\!\left(\gamma_c\+ ,\,\gamma^\ell_{M,\+1}\right)\!,$ the function $f^N_{\ell,\,s^\ast_\ell}$ is still non negative for all $s^\ast_\ell$ s.t.
    \begin{equation*}
        \frac{\pi\,\ell\+ !^{\+ 2}}{2^{2\ell+1}\!\left(\frac{\ell}{2}\right)!^{\+ 4}}\,\frac{(N\!-\!1)(M\!+\!1)}{\sqrt{M(M\!+\- 2)\,}}\,(\gamma^\ell_{M,\+1}\--\gamma)<s^\ast_\ell<\-1.
    \end{equation*}
    Notice that the lower bound is non-increasing in $\ell$, hence the sequence $\{s^\ast_\ell\}$ that makes $\Xi^{\,\zeta}_{\+\ell,\,s^\ast_\ell}$ non-negative for all $\zeta\geq 0$ and $\ell\in\No$ can be chosen within an interval that does not depend on $\ell$, namely
    $$\max\!\left\{\!0,\,\frac{\pi}{2}\,\frac{(N\--\-1)(M\-+\-1)}{\sqrt{M(M\-+\-2)\,}}\,(\gamma^0_{M,\+1}-\gamma)\!\right\}\!<s^\ast_\ell<\-1,\qquad\forall\ell\in\No\+.$$
\end{note}

\vs\vs

\section{Estimate of \texorpdfstring{$\,\Theta^{\+\zeta}$}{TEXT}}\label{estimateTheta}

\vs

Collecting the results obtained in the previous two sections, we can now establish detailed estimates for $\Theta^{\+\zeta}$.
Indeed, in the next proposition we prove a lower bound, which is the crucial ingredient for the proof of our main results.
We also prove an upper bound, which improves the result already obtained in proposition~\ref{boundednessPhiH-half}.

\begin{prop}\label{stimaTheta}
 Given $\varphi\in H^{\frac{1}{2}}(\R^{3})$ and $\zeta \geq 0$, we have
\begin{equation}\label{ThetaLowerBound}
         \Theta^{\+\zeta}[\varphi]\geq \Lambda_\gamma(N,M)\,\Theta^{\+\zeta}_{\mathrm{diag}}[\varphi], \qquad \text{for }\,\gamma>\gamma_c\+,
    \end{equation}
    where 
    \begin{equation}\label{LambdaDef}
        \define{\Lambda_\gamma(N,M);\min\!\left\{1,\,\tfrac{\pi(N-1)}{2}\+\tfrac{M+1}{\!\-\sqrt{M(M+2)}\,}(\gamma-\gamma_c)\right\}}\in (0,1].
    \end{equation}
Moreover
    \begin{equation}\label{Thetaup2}
      \Theta^{\+\zeta}[\psi] \leq  \Lambda'_\gamma(N,M)\,\Theta^{\+\zeta}_{\mathrm{diag}}[\psi]\,, \qquad \text{for }\,\gamma>0,
    \end{equation}
    where 
    \begin{equation}\label{LambdaPrime}
        \define{\Lambda'_{\gamma}(N,M);1+\tfrac{(N-1)(M+1)}{\!\-\sqrt{M(M+2)}\,}\max\!\Big\{\tfrac{\pi}{2}\+\gamma,\,\tfrac{1}{2}\+S_{\mathrm{off};\,1}(0)+\tfrac{2}{\pi}\+ \gamma\- \Big\}}.
    \end{equation}
    \begin{proof}
        Let $\varphi\!\in\! H^{\frac{1}{2}}(\R^{3})$ and consider decomposition~\eqref{partialWavesDecQF} and lemma~\ref{FsignLemma}.
        Then,
        \begin{align*}
            \Theta^{\+\zeta}[\varphi]&=\sum_{\ell\+ \in\+ \No}\sum_{m\+ =\+ -\ell}^\ell F^{\,\zeta}_\ell[ \FT{\varphi}_{\ell,m}]=\sum_{\ell\+ \in\+ \No}\sum_{m\+ =\+ -\ell}^\ell \!\left(\- F^{\,\zeta}_{\mathrm{diag}}\!+\- F^{\,\zeta}_{\mathrm{off};\,\ell}\-+\- F_{\mathrm{reg};\,\ell}\right)\![\FT{\varphi}_{\ell,m}]\\[-2.5pt]
            &\geq \sum_{\substack{\ell\+ \in\+ \N\\ \ell\, \text{odd}}}\sum_{m\+ =\+ -\ell}^\ell F^{\,\zeta}_{\mathrm{diag}}[\FT{\varphi}_{\ell,m}]+\!\sum_{\substack{\ell\+ \in\+ \No\\ \ell\, \text{even}}}\sum_{m\+ =\+ -\ell}^\ell\!\left(\- F^{\,\zeta}_{\mathrm{diag}}\!+\- F^{\+0}_{\mathrm{off};\,\ell}\-+\- F_{\mathrm{reg};\,\ell}\right)\![\FT{\varphi}_{\ell,m}].
        \end{align*}
        Taking account of definition~\eqref{auxiliaryQF}, for any choice of $\{s_\ell\}_{\ell\+\in\+\No}\!\-\subset\!(0,1)$, the previous inequality reads
        \begin{equation*}
            \Theta^{\+\zeta}[\varphi]\geq\sum_{\substack{\ell\+ \in\+ \N\\ \ell\, \text{odd}}}\sum_{m\+ =\+ -\ell}^\ell F^{\,\zeta}_{\mathrm{diag}}[\FT{\varphi}_{\ell,m}]+\!\sum_{\substack{\ell\+ \in\+ \No\\ \ell\, \text{even}}}\sum_{m\+ =\+ -\ell}^\ell(1-s_\ell)F^{\,\zeta}_{\mathrm{diag}}[\FT{\varphi}_{\ell,m}]+\- \Xi^{\,\zeta}_{\+\ell,\,s_\ell}[\FT{\varphi}_{\ell,m}].
        \end{equation*}
        According to lemma~\ref{auxiliaryPositivity}, there exists a sequence $\{s^\ast_\ell\}_{\ell\+\in\+\No}\!\-\subset\! (0,1)$ such that $\Xi^{\,\zeta}_{\+\ell,\,s^\ast_\ell}\-\geq\- 0$, hence
        \begin{align*}
            \Theta^{\+\zeta}[\varphi]&\geq \sum_{\substack{\ell\+ \in\+ \N\\ \ell\, \text{odd}}}\sum_{m\+ =\+ -\ell}^\ell F^{\,\zeta}_{\mathrm{diag}}[\FT{\varphi}_{\ell,m}]+\!\sum_{\substack{\ell\+ \in\+ \No\\ \ell\, \text{even}}}\sum_{m\+ =\+ -\ell}^\ell(1-s^\ast_\ell)F^{\,\zeta}_{\mathrm{diag}}[\FT{\varphi}_{\ell,m}]\\[-7.5pt]
            &\geq \sum_{\ell\+ \in\+ \No}\sum_{m\+ =\+ -\ell}^\ell (1-s^\ast_\ell)F^{\,\zeta}_{\mathrm{diag}}[\FT{\varphi}_{\ell,m}]\geq\inf_{k\+ \in\+ \No}(1-s_k^\ast)\sum_{\ell\+ \in\+ \No}\sum_{m\+ =\+ -\ell}^\ell F^{\,\zeta}_{\mathrm{diag}}[\FT{\varphi}_{\ell,m}]\\
            &=\!\inf_{k\+ \in\+ \No}(1-s_k^\ast)\,\Theta^{\+\zeta}_{\mathrm{diag}}[\varphi]
        \end{align*}
        where, according to remark~\ref{partialWaveUniformity}, each $s^\ast_k$ can be arbitrarily chosen within an interval in $(0,1)$ that does not shrink as $k$ varies.
        Exploiting this fact, we can optimize the inequality by choosing 
        $$s^\ast_k=\tfrac{\pi(N-1)}{2} \tfrac{M+1}{\!\-\sqrt{M(M+2)}\,}\max\!\left\{0,\,\gamma^0_{M,\+1}-\gamma\right\}\!,\qquad \forall k\in\No$$
        so that $\Theta^{\+\zeta}[\varphi]\geq \Lambda_\gamma(N,M)\,\Theta^{\+\zeta}_{\mathrm{diag}}[\varphi]$, 
        where $\Lambda_\gamma(N,M)$ is given by~\eqref{LambdaDef}.
    
        \n Let us consider the upper bound.
        By lemmata~\ref{FsignLemma} and~\ref{SoffregAboveBound}, we have 
        \begin{align*}
            \Theta^{\+\zeta}[\varphi]&=\sum_{\ell\+ \in\+ \No}\sum_{m\+ =\+ -\ell}^\ell F^{\,\zeta}_\ell[\FT{\varphi}_{\ell,m}]=\sum_{\ell\+ \in\+ \No}\sum_{m\+ =\+ -\ell}^\ell \!\left(\- F^{\,\zeta}_{\mathrm{diag}}\!+\- F^{\,\zeta}_{\mathrm{off};\,\ell}\-+\- F_{\mathrm{reg};\,\ell}\right)\![\FT{\varphi}_{\ell,m}]\\
            &\leq \sum_{\substack{\ell\+ \in\+ \No\\ \ell\, \text{even}}}\sum_{m\+ =\+ -\ell}^\ell \!\left(\- F^{\,\zeta}_{\mathrm{diag}}\!+F_{\mathrm{reg};\,\ell}\- \right)\![\FT{\varphi}_{\ell,m}]+\!\sum_{\substack{\ell\+ \in\+ \N\\ \ell\, \text{odd}}}\sum_{m\+ =\+ -\ell}^\ell\!\left(\- F^{\,\zeta}_{\mathrm{diag}}\!+\- F^{\+0}_{\mathrm{off};\,\ell}\-+\- F_{\mathrm{reg};\,\ell}\right)\![\FT{\varphi}_{\ell,m}]\\[-7.5pt]
            &\leq \left[1+\tfrac{N-1}{2}\sqrt{\tfrac{\eta}{\mu}}\max\!\Big\{S_{\mathrm{reg};\,0}(0),\,S_{\mathrm{off};\,1}(0)+S_{\mathrm{reg};\,1}(0)\- \Big\}\right]\!\sum_{\ell\+ \in\+ \No}\sum_{m\+ =\+ -\ell}^\ell \!F^{\,\zeta}_{\mathrm{diag}}[\FT{\varphi}_{\ell,m}]\\[-5pt]
            &=\Lambda'_\gamma(N,M)\,\Theta^{\+\zeta}_{\mathrm{diag}}[\varphi]
        \end{align*}
        where $\Lambda'_{\gamma}(N,M)$ is given by~\eqref{LambdaPrime}. 

    \end{proof}
\end{prop}
\begin{note}
    We stress that the upper bound obtained in~\eqref{Thetaup2} is an improvement of estimate~\eqref{Thetaup1}, since
    \begin{equation}
        \Lambda'_\gamma(N,M)\leq 1+\tfrac{(N-1)(M+1)}{\sqrt{M(M+2)}} \left( \tfrac{M+1}{M} + \tfrac{\pi}{2}\+ \gamma \right)\!.
    \end{equation}
    Indeed, in case $\gamma\geq \tfrac{\pi}{\pi^2-4}\+S_{\mathrm{off};\,1}(0)$ one  has
    \begin{equation*}
        1+\tfrac{(N-1)(M+1)}{\sqrt{M(M+2)}} \left( \tfrac{M+1}{M} + \tfrac{\pi}{2}\+ \gamma \right)\!\geq 1+\tfrac{(N-1)(M+1)}{\sqrt{M(M+2)}}\+\tfrac{\pi}{2}\+ \gamma = \Lambda'_\gamma(N,M).
    \end{equation*}
    On the other hand, consider $0<\gamma<\tfrac{\pi}{\pi^2-4}\+S_{\mathrm{off};\,1}(0)$.
    Taking into account  the following elementary estimate
    \begin{equation}\label{arcsinElementary}
        \arcsin(t)\geq\, t\+\geq t\,\sqrt{\frac{1-t}{1+t}}\,,\qquad 0\leq t\leq 1,
    \end{equation}
  we have   $\frac{1}{t}\+\sqrt{1-t^2}\+\arcsin(t)\geq 1-t$, or $\;1-\frac{1}{t}\+\sqrt{1-t^2}\+\arcsin(t)\leq t$.  Then 
    \begin{gather}
      1\--\sqrt{x^2-1}\+\arcsin{\tfrac{1}{x}}\leq \tfrac{1}{x}\+,\qquad x\geq 1.\label{arcsinNotSoElementary}
    \end{gather}
    Therefore, exploiting~\eqref{arcsinNotSoElementary} with $x=M+1$, one obtains 
    \begin{gather*}
        1+\tfrac{(N-1)(M+1)}{\sqrt{M(M+2)}} \left( \tfrac{M+1}{M} + \tfrac{\pi}{2}\+ \gamma \right)\!\geq 1+\tfrac{(N-1)(M+1)}{\sqrt{M(M+2)}} \left(1 + \tfrac{2}{\pi}\+ \gamma \right)\!\\
        \geq 1+\tfrac{(N-1)(M+1)}{\sqrt{M(M+2)}} \left[(M\-+\-1)\-\left(1\--\sqrt{M(M\-+\-2)}\+\arcsin{\tfrac{1}{M+1}}\right)\! + \tfrac{2}{\pi}\+ \gamma\+\right]\!\\
        \geq 1+\tfrac{2}{\pi}\+\tfrac{(N-1)(M+1)}{\sqrt{M(M+2)}} \left[(M\-+\-1)\-\left(1\--\sqrt{M(M\-+\-2)}\+\arcsin{\tfrac{1}{M+1}}\right)\- + \gamma\+ \right]\!=\Lambda'_\gamma(N,M).
    \end{gather*}
\end{note}

\vs\vs

\section{Proof of the main results}\label{closure&Boundedness}

\vs

In this section we complete the proof of the results stated in section~\ref{mainResults}.
\begin{proof}[Proof of proposition~\ref{coercivityLemma}]
     Let us recall that, for any charge $\xi\-\in\- H^{\frac{1}{2}}(\R^{3N})$, we have defined a  rescaled charge $\phi\-\in\- H^{\frac{1}{2}}(\R^{3N})$ given by \eqref{rescaledModifiedCharge}. 
     According to equations~\eqref{Phi-Phi3bodyConjuction} and~\eqref{ThetaLowerBound}, we can deduce a lower bound for the quadratic form $\Phi^\lambda$
    \begin{align*}
        \Phi^\lambda[\xi]&=\Phi_0[\xi]+ \tfrac{4\pi N}{\!\-\sqrt{\mu\,}\+}\!\-\integrate[\R^{3(N-1)}]{\!\textstyle{\sqrt{\sum_{j=2}^N k_j^2+\lambda\,}\,}\Theta^{\+1}[\phi](\vect{k}_2,\ldots,\vect{k}_N); \mspace{-33mu}d\vect{k}_2\cdots d\vect{k}_N}\\
        &\geq \Phi_0[\xi]+ \Lambda_\gamma(N,M)\,\tfrac{4\pi N}{\!\-\sqrt{\mu\,}\+}\!\-\integrate[\R^{3(N-1)}]{\!\textstyle{\sqrt{\sum_{j=2}^N k_j^2+\lambda\,}\,}\Theta^{\+1}_{\mathrm{diag}}[\phi](\vect{k}_2,\ldots,\vect{k}_N); \mspace{-33mu}d\vect{k}_2\cdots d\vect{k}_N}\\
        &=\Phi_0[\xi]+ \Lambda_\gamma(N,M)\, \Phi^\lambda_{\mathrm{diag}}[\xi]\+,\qquad\forall\lambda>0,\,\gamma>\gamma_c\+.
    \end{align*}
    Recalling definition~\eqref{notePhi} and assumption~\eqref{positiveBoundedCondition} (which implies that $\tilde{\alpha}$ is essentially bounded), we have 
    \begin{align*}
        \Phi_0[\xi]&\geq\tfrac{4\pi N}{\mu}\inf_{\;\Rplus\!}\{\tilde{\alpha}\}\,\scalar{\xi}{\xi}[\Lp{2}(\R^{3N})]=\tfrac{4\pi N}{\mu}\!\left(\alpha-\tfrac{(N-1)\,\gamma}{b}\right)\!\norm{\xi}^2\\
        &\geq \!\left(\min\{0,\alpha\}-\tfrac{(N-1)\,\gamma}{b}\right)\!\tfrac{4\pi N}{\mu}\+\lVert\FT{\xi}\rVert^2\\
        &\geq \frac{\min\{0,4\pi N\alpha\}\!-\- \tfrac{N(N-1)\,4\pi\+\gamma}{b}}{\sqrt{\lambda}\,\mu}\!\-\integrate[\R^{3N}]{\!\textstyle{\sqrt{\frac{k_1^2}{M+1}\- +\-\sum_{j=2}^N k_j^2\- +\!\lambda\,}\,}\abs{\FT{\xi}(\vect{k}_1,\ldots,\vect{k}_N)}^2;\mspace{-10mu}d\vect{k}_1\-\cdots d\vect{k}_N}\\
        &=\frac{\min\!\left\{0,\,\alpha\+ b\right\}-(N\!-\!1)\+ \gamma}{b\,\sqrt{\lambda\+\mu}}\,\Phi^\lambda_{\mathrm{diag}}[\xi].
    \end{align*}
    Collecting the results obtained so far, we get
    \begin{equation}\label{PhiLowerBound}
        \Phi^\lambda[\xi]\geq \-\left[\Lambda_\gamma(N,M)-\tfrac{\max\{(N-1)\,\gamma,\;(N-1)\,\gamma-\alpha \+b\}}{b\,\sqrt{\lambda\+\mu}}\right]\-\Phi^\lambda_{\mathrm{diag}}[\xi]\+.
    \end{equation}
    The last expression is positive if $\lambda$ is large enough, i.e. if $\lambda>\lambda_0\+$, with
    \begin{equation}\label{lambdanote}
    \define{\lambda_0;}\begin{dcases}
        \frac{(N\!-\!1)^2\,\gamma^2}{\mu\+ \Lambda_\gamma(N,M)^2\,b^2}\,,\quad&\text{if }\alpha\geq 0,\\
        \frac{\left[(N\!-\!1)\+ \gamma+\abs{\alpha}\+b\+ \right]^2}{\mu\+ \Lambda_\gamma(N,M)^2\,b^2},\quad&\text{if }\alpha<0.
        \end{dcases}
    \end{equation}
    Concerning the upper bound, we proceed in the same way and we find
    \begin{equation}
        \Phi^\lambda[\xi]\leq \Lambda'_{\gamma}(N,M)\,\Phi^\lambda_{\mathrm{diag}}[\xi]+\Phi_0[\xi]\+,\qquad \forall\lambda>0,\,\gamma>0.
    \end{equation}
Moreover,
    \begin{align*}
        \Phi_0[\xi]&\leq\tfrac{4\pi N}{\mu}\sup_{\;\Rplus\!}\{\tilde{\alpha}\}\norm{\xi}^2\!=\!\left(\alpha+\tfrac{(N-1)\,\gamma}{b}\right)\!\tfrac{4\pi N}{\mu}\lVert\FT{\xi}\rVert^2\\
        &\leq\frac{\max\{0,4\pi N\alpha\}\!+\- \tfrac{N(N-1)\,4\pi\+\gamma}{b}}{\sqrt{\lambda}\,\mu}\!\-\integrate[\R^{3N}]{\!\textstyle{\sqrt{\frac{k_1^2}{M+1}\- +\-\sum_{j=2}^N k_j^2\- +\!\lambda\,}\,}\abs{\FT{\xi}(\vect{k}_1,\ldots,\vect{k}_N)}^2;\mspace{-10mu}d\vect{k}_1\-\cdots d\vect{k}_N}\\
        &=\frac{\max\!\left\{0,\,\alpha\+ b\right\}+(N\!-\!1)\+ \gamma}{b\,\sqrt{\lambda\+\mu}}\,\Phi^\lambda_{\mathrm{diag}}[\xi].
    \end{align*}
    Using the last estimate, we obtain
    \begin{equation}\label{PhiUpperBound}
        \Phi^\lambda[\xi]\leq \left[\Lambda'_\gamma(N,M)+\tfrac{\max\{(N-1)\,\gamma,\;(N-1)\,\gamma+\alpha \+b\}}{b\,\sqrt{\lambda\+\mu}}\right]\-\Phi^\lambda_{\mathrm{diag}}[\xi].
    \end{equation}
\end{proof}

\n
The estimates~\eqref{PhiLowerBound} and~\eqref{PhiUpperBound} guarantee that $\Phi^\lambda$ is closed and bounded from below by a positive constant.
Then, if $\gamma\->\-\gamma_c\+$, the quadratic form $\Phi^\lambda$ uniquely defines a s.a. and positive operator $\Gamma^\lambda$ in $\Lp{2}[\R^{3N}]$ for all $\lambda\->\-\lambda_0\+ $.
Such operator is  characterized as follows
\begin{align*}
    \dom{\Gamma^\lambda}&=\!\left\{\xi\in H^{\frac{1}{2}}(\R^{3N})\,\big|\;\exists g\in \Lp{2}[\R^{3N}]\,\text{ s.t. }\Phi^\lambda[\varphi,\,\xi\+ ]=\scalar{\varphi}{g},\quad\forall\varphi\in H^{\frac{1}{2}}(\R^{3N})\right\}\!,\\
    \Gamma^\lambda \xi &= g,\quad \forall \xi\in\dom{\Gamma^\lambda}
\end{align*}
where $\Phi^\lambda[\+ \cdot\+ ,\+ \cdot\+ ]$ is the sesquilinear form associated to $\Phi^\lambda[\+ \cdot\+ ]$ via the polarization identity.
    
\n Moreover, $\Gamma^\lambda$ is invertible for all $\lambda\->\-\lambda_0$ and comparing~\eqref{QF} with equation~\eqref{inProgressQF1}, the following relation holds
\begin{equation}\label{invertibleGammaRelation}
    \Gamma^\lambda\xi=\tfrac{4\pi}{\mu}\sum_{i=1}^N(\Gamma_{\!\mathrm{diag}}^{i,\+\lambda}\-+\Gamma_{\!\mathrm{off}}^{i,\+\lambda}\-+\Gamma_{\!\mathrm{reg}}^i)\+\xi,\qquad \forall \xi\in\dom{\Gamma^\lambda}.
    \end{equation}
\begin{note}\label{resolventRemark}
    From equation~\eqref{theoreticalGammaConnection}, we have
    $\Gamma^\lambda\!=\Gamma_{\!\lambda}\+$.
    This means that,
    by~\eqref{firstResolventIdGamma}, we can extend the definition of $\,\Gamma^{-z}\!=\-\Gamma(z)$ to all $z\-\in\-\rho(\mathcal{H}_0)$ and condition~\eqref{firstResolventGammaTheory} of section~\ref{kreinRecall} holds true.
    This fact implies that the domain of $\,\Gamma(z)$ does not depend on $z\-\in\-\rho(\mathcal{H}_0)$ and it is denoted by $D$.
    Furthermore, using again~\eqref{firstResolventIdGamma} we deduce that condition~\eqref{adjointGammaTheory} of section~\ref{kreinRecall} is also satisfied.
    Indeed,
    \begin{gather*}
        \adj{\Gamma(\conjugate{z})}\!-\Gamma^\lambda=-(z+\lambda)\+\adj{G^\lambda}G(\conjugate{z})=-(z+\lambda)\+\adj{G(z)}G^\lambda=\Gamma(z)-\Gamma^\lambda\\
        \implies \adj{\Gamma(z)}\!=\Gamma(\conjugate{z}).
    \end{gather*}
    Finally, we claim that $\Gamma(z)$ is invertible for any $z\-\in\C\smallsetminus[-\lambda_0,\pInfty)$.
    We have already discussed the case $z\-<\!-\lambda_0\+$, thus let us consider $z\-\in\C\smallsetminus\R$.
    According to~\eqref{firstResolventGammaTheory}
    \begin{equation*}
        \Gamma(\conjugate{z})-\Gamma(z)=(z-\conjugate{z})\+\adj{G(z)}G(z),
    \end{equation*}
    whereas condition~\eqref{adjointGammaTheory} implies
    \begin{equation*}
        \scalar{\xi}{[\Gamma(\conjugate{z})-\Gamma(z)]\+\xi}[\hilbert*_N]\!=-2i\,\Im{\+\scalar{\xi}{\Gamma(z)\+\xi}[\hilbert*_N]}\+,\qquad\forall\xi\in D.
    \end{equation*}
    Hence, one obtains
    \begin{equation}
        \Im{\+\scalar{\xi}{\Gamma(z)\+\xi}[\hilbert*_N]}=-\Im(z)\norm{G(z)\+\xi}[\hilbert*_{N+1}]^2,\qquad\forall\xi\in D.
    \end{equation}
    Therefore
    \begin{align*}
        \abs{\scalar{\xi}{\Gamma(z)\+\xi}[\hilbert*_N]}^2&=|\scalar{\xi}{\tfrac{\Gamma(z)\++\+\Gamma(\conjugate{z})}{2}\+\xi}[\hilbert*_N]|^2+(\Im\+z)^2\norm{G(z)\+\xi}[\hilbert*_{N+1}]^4\\
        &\geq (\Im\+z)^2\norm{G(z)\+\xi}[\hilbert*_{N+1}]^4\!>0.
    \end{align*}
\end{note}

\n
We are now in position to conclude the proof of theorem~\ref{closedBoundedQF}.
\begin{proof}[Proof of theorem \ref{closedBoundedQF}]
    Taking into account  proposition~\ref{coercivityLemma}, $Q$ is bounded from below, since for any $\psi\in\dom{Q}$, one has
    \begin{equation*}
        Q[\psi]=\mathscr{F}_\lambda[w^\lambda]-\lambda\norm{\psi}^2\-+\Phi^\lambda[\xi]\geq -\lambda\norm{\psi}^2\!,\quad\forall \lambda>\lambda_0\+.
    \end{equation*}
    Now, let us fix $\lambda>\lambda_0\+$.
    By construction $Q$ is hermitian, hence, the associated sesquilinear form $Q[\+ \cdot,\cdot\+ ]$ is symmetric.
    In particular, this means that the sesquilinear form $s[\+ \cdot,\cdot\+ ]$ given by
    $$s[\psi,\varphi]\vcentcolon=Q[\psi,\varphi]+(1+\lambda)\scalar{\psi}{\varphi},\quad\forall\psi,\varphi\in\dom{Q}$$
    defines a scalar product in $\hilbert*_{N+1}\+$.
    Therefore, we equip $\dom{Q}\subset\hilbert*_{N+1}$ with the norm
    \begin{equation}\label{closednessStep}
        \norm{\psi}[Q]^2\vcentcolon=Q[\psi]+(1+\lambda)\- \norm{\psi}^2\-=\mathscr{F}_\lambda[w^\lambda]+\Phi^\lambda[\xi]+\norm{\psi}^2\!.
    \end{equation}
   We prove that $Q$ is closed by showing the completeness of $\dom{Q}$ with respect to $\norm{\cdot}[Q]\+$.
    To this end, let $\{\psi_n\}\subset\dom{Q}$ and $\psi\in\hilbert*_{N+1}$ be respectively a sequence and a vector s.t. $\normConverge[Q]{\psi_n}{\psi_m}$ as $n, m$ go to infinity and $\normConverge{\psi_n}{\psi}$.
    By~\eqref{closednessStep}, we have 
    \begin{equation}\label{simultaneousConvergence}
    \mathscr{F}_\lambda[w^\lambda_n-w^\lambda_m]+\Phi^\lambda[\xi_n-\xi_m]\!\longrightarrow \- 0
    \end{equation}
    and, since both $\mathscr{F}_\lambda$ and $\Phi^\lambda$ are closed and positive,~\eqref{simultaneousConvergence} means
    \begin{align*}
        &\mathscr{F}_\lambda[w^\lambda_n-w^\lambda_m]\!\longrightarrow \- 0, &\Phi^\lambda[ \xi_n-\xi_m]\!\longrightarrow \- 0,\qquad \text{as }n,m\to\pInfty.
    \end{align*}
    Hence, $\{w^\lambda_n\}$ and $\{\xi_n\}$ are Cauchy sequences in $H^1(\R^{3(N+1)})\+\cap\+\hilbert*_{N+1}$ and $H^{\frac{1}{2}}(\R^{3N})\+\cap\+\hilbert*_N\+ $, respectively.
    Thus, there exist $w^\lambda\-\in\- H^1(\R^{3(N+1)})\cap\hilbert*_{N+1}$ and $\xi\-\in\- H^{\frac{1}{2}}(\R^{3N})\cap\hilbert*_N$ such that
    \begin{align*}
        &\normConverge[H^1(\R^{3(N+1)})]{w^\lambda_n}{w^\lambda}, & \normConverge[H^{1/2}(\R^{3N})]{\xi_n}{\xi},\qquad\text{as }n\to\pInfty.
    \end{align*}
    Furthermore, since $G^\lambda$ defined in~\eqref{potentialDef} is bounded for all $\lambda>0$, one has that $\psi_n\-=w^\lambda_n+G^\lambda \xi_n$ converges in $\hilbert*_{N+1}$ to the vector $w^\lambda\-+G^\lambda\xi$.
    By uniqueness of the limit, $\psi=w^\lambda+G^\lambda\xi$ and thus, $\psi\-\in\-\dom{Q}$.
    We have shown that $(\dom{Q},\norm{\cdot}[Q])$ is a Banach space, hence $Q$ is closed.
    
\end{proof}

\n
From theorem~\ref{closedBoundedQF}, we know that $Q$ uniquely defines a self-adjoint and bounded from below Hamiltonian $\mathcal{H}$, $\dom{\mathcal{H}}$ in the Hilbert space $\hilbert*_{N+1}=\Lp{2}(\R^3)\otimes \LpS{2}(\R^{3N})$.
We conclude with the proof of theorem~\ref{hamiltonianCharacterization} which characterizes domain and action of $\mathcal{H}$.
\begin{proof}[Proof of theorem \ref{hamiltonianCharacterization}]
    Let us assume that $\psi=w^\lambda\-+G^\lambda\xi\-\in\-\dom{Q}$, with $\lambda\->\-\lambda_0$.
    Then, there exists $f\-\in\-\hilbert*_{N+1}$ such that the sesquilinear form $Q[\+\cdot,\cdot\+]$ associated to $Q[\+\cdot\+]$ via the polarization identity satisfies \begin{equation}\label{sesquilinearQRelation}
        Q[v,\psi]=\scalar{v}{f},\quad \forall v=w^\lambda_v+G^\lambda\xi_v\in\dom{Q}
    \end{equation}
    where $\define*{f;\mathcal{H}\psi}$.
    By definition one has
    \begin{equation}\label{sesquilinearQ}
        Q[v,\psi]=\scalar{(\mathcal{H}_0+\lambda)^{\frac{1}{2}}w^\lambda_v}{(\mathcal{H}_0+\lambda)^{\frac{1}{2}} w^\lambda}-\lambda\scalar{v}{\psi}+\Phi^\lambda[\xi_v,\xi].
    \end{equation}
    Let us consider $v\in H^1(\R^{3(N+1)})\cap\hilbert*_{N+1}\+$, so that $\xi_v\equiv 0$ by injectivity of $G^\lambda$.
    Then
    \begin{equation*}
        \scalar{(\mathcal{H}_0+\lambda)^{\frac{1}{2}} v}{(\mathcal{H}_0+\lambda)^{\frac{1}{2}} w^\lambda}-\lambda\scalar{v}{\psi}=\scalar{v}{f},\quad\forall v\in H^1(\R^{3(N+1)})\cap\hilbert*_{N+1}\+.
    \end{equation*}
    Hence, $w^\lambda\-\in\- H^2(\R^{3(N+1)})\cap\hilbert*_{N+1}$ and
    \begin{equation}\label{HpsiHfreew}
        (\mathcal{H}_0+\lambda)w^\lambda-\lambda\+ \psi=f
    \end{equation}
    which is equivalent to
    \begin{equation}
        \mathcal{H}\psi=\mathcal{H}_0\+ w^\lambda-\lambda G^\lambda\xi.
    \end{equation}
    Now, let $v\-\in\-\dom{Q}$.
    Taking account of~\eqref{HpsiHfreew}, we have
    \begin{equation*}
        \scalar{v}{f+\lambda\+ \psi}=\scalar{w^\lambda_v}{(\mathcal{H}_0+\lambda)w^\lambda}+\scalar{G^\lambda\xi_v}{(\mathcal{H}_0+\lambda)w^\lambda}.
    \end{equation*}
    On the other hand, recalling~\eqref{sesquilinearQRelation} and~\eqref{sesquilinearQ},
    \begin{gather*}
        \scalar{v}{f+\lambda\+ \psi}=Q[v,\psi]+\lambda\scalar{v}{\psi}=\scalar{w^\lambda_v}{(\mathcal{H}_0+\lambda)w^\lambda}+\Phi^\lambda[\xi_v,\xi],
    \intertext{hence,}
        \Phi^\lambda[\xi_v,\xi]=\scalar{G^\lambda\xi_v}{(\mathcal{H}_0+\lambda)w^\lambda}=\tfrac{4\pi}{\mu}\scalar{\xi_v}{\textstyle{\sum_{i=1}^N}w^\lambda|_{\pi_i}}[\hilbert*_N],\quad\forall\xi_v\in H^{\frac{1}{2}}(\R^{3N})\cap\hilbert*_N
    \end{gather*}
    where we have used definition~\eqref{potential} in the last step.
    Therefore, we conclude that $\xi\-\in\-D$ and $\Gamma^\lambda\xi=\frac{4\pi}{\mu}\sum_{i=1}^N w^\lambda\big|_{\pi_i}$.

    \n Concerning the resolvent, thanks to remark~\ref{resolventRemark}, we have shown that for any $z\-\in\-\rho(\mathcal{H}_0)$, the operator $\Gamma(z)$ given by~\eqref{firstResolventIdGamma} fulfils conditions~\eqref{firstResolventGammaTheory}, \eqref{adjointGammaTheory} and~\eqref{invertibilityGammaTheory} of section~\ref{kreinRecall}.
    In particular, we have obtained $\rho(\mathcal{H})\supseteq \C\smallsetminus[-\lambda_0,\pInfty)$ and
    $$\resolvent{\mathcal{H}_0}[z]+G(z)\Gamma(z)^{-1}\adj{G(\conjugate{z})}\-=(\mathcal{H}-z)^{-1}.$$
\end{proof}

\vs\vs

\appendix

\section{}

In this appendix we recall some known abstract results on the theory of s.a. extensions, we characterize the properties of the potential $G^{\lambda}$ and then we give a heuristic construction of  the quadratic form $Q$.

\vs

\subsection{Sketch of Birman-Kre\u{\i}n-Vishik's theory}\label{kreinRecall}
\n

\vs
Here we briefly recall some results of the theory of s.a. extensions which are relevant for our purposes.
Let $\maps{A}{\dom{A}\subseteq\hilbert;\hilbert}$ be a s.a. operator in a complex Hilbert space $\hilbert$. 
Then, consider a symmetric and densely defined operator $\maps{S}{\dom{S}\subset\dom{A};\hilbert}$ that is closed according to the graph norm of $A$ and s.t. $A|_{\dom{S}}=S$.
Following~\cite{P1} and~\cite{P2}, we show how to characterize the whole set of s.a. extensions of $S$.

\n To this end, let $\X*$ be an auxiliary complex Hilbert space and $\maps{\tau}{\dom{A};\X*}$ a linear operator.
Endowing $\dom{A}$ with the graph norm of $A$, we assume that $\tau$ is continuous, its range is dense and it fulfils $\ker{\tau}=\dom{S}$.

\n Then, if $z\-\in\-\rho(A)$, we define $\define*{\adj{(\tau\resolvent{A}[\conjugate{z}])}\!\-;\mathcal{G}(z)}\-\in\-\bounded{\X*,\hilbert}$ and we consider a densely defined, linear operator $\maps{\Gamma(z)}{D \subset \X* ;\X*}$, satisfying
\begin{subequations}\label{conditionsGammaTheory}
\begin{align}
    &\Gamma(z)-\Gamma(w)=(w\--\-z)\adj{\mathcal{G}(\bar{w})}\mathcal{G}(z)\-=\-(w\--\-z)\adj{\mathcal{G}(\bar{z})}\mathcal{G}(w),\quad\forall w,z\in\rho(A),\tag{i}\label{firstResolventGammaTheory}\\
    &\adj{\Gamma(z)}\!=\Gamma(\bar{z}),\tag{ii}\label{adjointGammaTheory}\\
    &\exists z\in\rho(A) \,:\: 0\in\rho(\Gamma(z)).\tag{iii}\label{invertibilityGammaTheory}
\end{align}
\end{subequations}
Notice that, by condition~\eqref{firstResolventGammaTheory}, the domain $D$ does not depend on $z\in\rho(A)$.
Actually, the choice of $\tau$ and $\Gamma$ completely identify the s.a. extension of $S$, denoted by $A^\tau_\Gamma\+$.
Indeed, one of the main results of~\cite{P1} is that, for any $z\-\in\-\rho(A)$ s.t. $\Gamma(z)^{-1}\-\in\bounded{\X*}$, the operator
\begin{equation}
    \define{R^\tau_{\Gamma}(z);\resolvent{A}[z]+\mathcal{G}(z)\Gamma(z)^{-1}\adj{\mathcal{G}}(\conjugate{z})}
\end{equation}
defines the resolvent of a s.a. operator $A^\tau_\Gamma$ that coincides with $A$ on $\ker{\tau}$.
It is characterized by
\begin{equation}\label{saExtensionA}
    \begin{dcases}
    \define{\dom{A^\tau_\Gamma};\left\{\phi\in\hilbert\,\big|\;\phi=\phi_z+\mathcal{G}(z)\Gamma(z)^{-1}\tau\phi_z\+,\:\phi_z\in\dom{A} \right\}}\-,\\
    (A^\tau_\Gamma-z)\phi=(A-z)\phi_z\+.
    \end{dcases}
\end{equation}
Clearly, $R^\tau_\Gamma(z)\in\bounded{\hilbert,\dom{A^\tau_\Gamma}}$ and $\define{A_\Gamma^\tau;R^\tau_\Gamma(z)^{-1}\!+\-z}$ does not depend on $z$.
\begin{note}\label{uniqueDecomposition}
    It is easy to see  that the decomposition of an element of  $\dom{A^\tau_\Gamma}$ is unique, because of the density of $\ker{\tau}$ in $\dom{A}$.
Indeed,  given $\phi\in\X*$ and $\psi\in\ker{\tau}$, one has
    \begin{align*}
        0=\scalar{\phi}{\tau\psi}[\X*]&=\scalar{\phi}{\tau\resolvent{A}[\conjugate{z}](A-\conjugate{z})\psi}[\X*]\\
        &=\scalar{\mathcal{G}(z)\+\phi}{(A-\conjugate{z})\psi}[\hilbert]=\scalar{(A-z)\mathcal{G}(z)\+\phi}{\psi}[\hilbert]
    \end{align*}
    that means $(A-z)\mathcal{G}(z)\+\phi\in\ker{\tau}^\perp$, i.e. $(A-z)\mathcal{G}(z)\+\phi=0$ for all $\phi\in\X*$ and this implies $\ran{\mathcal{G}(z)}\cap\dom{A}=\{0\}$.
\end{note}
\n In the following proposition we state a useful property of $\mathcal{G}(z)$.
\begin{prop}\label{potentialProp}
Let $z,w\in\rho(A)$ and $\psi_1,\psi_2\in\X*$. If $\ran{\tau}$ is dense in $\X*$, then
\begin{equation*}
\mathcal{G}(z)\psi_1-\mathcal{G}(w)\psi_2\in\dom{A} \iff \psi_1=\psi_2\+.
\end{equation*}
\begin{proof}
    $\Leftarrow)\quad$ As a consequence of the first resolvent identity, one has
    \begin{subequations}
    \begin{gather}
        (\conjugate{z}-\conjugate{w})\adj{\mathcal{G}(w)}\resolvent{A}[\conjugate{z}]=\adj{\mathcal{G}(z)}\!-\adj{\mathcal{G}(w)}\!,\\
        \label{potentialResolventIdentity}
        (z-w)\resolvent{A}[z]\mathcal{G}(w)=\mathcal{G}(z)-\mathcal{G}(w)
    \end{gather}
    \end{subequations}
    which proves that $\ran{\mathcal{G}(z)-\mathcal{G}(w)}\in\dom{A}$.
    
    $\Rightarrow)\quad$ Using equation~\eqref{potentialResolventIdentity}, one can rewrite $\mathcal{G}(z)$, obtaining
    \begin{equation*}
        \mathcal{G}(z)\psi_1-\mathcal{G}(w)\psi_2=\mathcal{G}(w)(\psi_1-\psi_2)+(z-w)\resolvent{A}[z]\mathcal{G}(w)\psi_1\+.
    \end{equation*}
    By hypothesis, the left hand side and the last term in the right hand side both belong to $\dom{A}$, while $\ran{\mathcal{G}(w)}\cap\dom{A}=\{0\}$. 
    This means that $\psi_1-\psi_2\in\ker{\mathcal{G}(w)}$, namely $\psi_1=\psi_2\+$, since
    $\ker{\mathcal{G}(z)}=\ran{\tau}^\perp=\{0\}$ as long as $\tau$ has dense range.
    
\end{proof}
\end{prop}
\n We can also give an equivalent representation of the operator $A^\tau_\Gamma$ where a sort of boundary condition appears.
More precisely, if we introduce the ``charge'' $\define{\xi\-;\-\Gamma(z)^{-1}\tau\phi_z}\-\in\-\X*$, it is straightforward to see that 
\begin{equation}
    \begin{dcases}
    \dom{A^\tau_\Gamma}=\left\{\phi\in\hilbert\,\big|\;\phi=\phi_z+\mathcal{G}(z)\+\xi,\,\Gamma(z)\+\xi=\tau\phi_z\+,\:\phi_z\in\dom{A},\,\xi\in D \right\}\-,\\
    A^\tau_\Gamma\+\phi=A\phi_z+z\+\mathcal{G}(z)\+\xi.
    \end{dcases}
\end{equation}
We stress that the definition of $\xi$ does not depend on $z$.

\n Indeed, the uniqueness of the decomposition of an element in $\dom{A^\tau_\Gamma}$, discussed in remark~\ref{uniqueDecomposition},  implies that $$\phi_z-\phi_w=\mathcal{G}(w)\+\xi_w-\mathcal{G}(z)\+\xi_z\+.$$ 
Since the left hand side is in $\dom{A}$, using proposition~\ref{potentialProp}, one has $\xi_z\-=\xi_w\+$, namely $$\Gamma(z)^{-1}\+\tau\phi_z=\Gamma(w)^{-1}\+\tau\phi_w\+.$$

\n Let us establish the connection of the abstract setting with our problem.
The Hilbert space $\hilbert$ is given by $\hilbert*_{N+1}\+$, while the operator $A$ plays the role of the free Hamiltonian $\mathcal{H}_0$ and $S$ corresponds to $\dot{\mathcal{H}}_0\+$.
The auxiliary Hilbert space $\X*$ in our case is $\hilbert*_N$.
Define
\begin{equation}\label{auxiliaryHilbert}
    \define{\X_i;\Lp{2}[\R^3,d\vect{x}_0]\-\otimes\-\LpS{2}[\R^{3(N-1)}\-,\+d\vect{x}_1\cdots d\check{\vect{x}}_i\cdots d\vect{x}_N]}
\end{equation}
and let $\maps{T_i}{\dom{\mathcal{H}_0};\X_i}$ be the trace operator $f\longmapsto \frac{4\pi}{\mu} f|_{\pi_i}$.
Observe that each $\X_i$ is isomorphic to $\hilbert*_N\+$.
The continuous map $\tau$ is represented by the operator $T$, given by
\begin{equation}\label{traceOperator}
\begin{split}
    &\quad\maps{T}{\dom{\mathcal{H}_0}\subset\hilbert*_{N+1}; \hilbert*_N},\\[-5pt]
    &T:\psi\longmapsto \sum_{i=1}^N T_i\+\psi=\tfrac{4\pi}{\mu}\-\sum_{i=1}^N\psi|_{\pi_i}\+.
\end{split}
\end{equation}
The reason behind the constant $\frac{4\pi}{\mu}$ will be clarified in remark~\ref{domainDecomposition}.
Notice that $T$ is bounded since the trace operators $T_i: f\longmapsto f|_{\pi_i}$ are continuous between $H^{\frac{3}{2}+\+s}(\R^{3(N+1)})$ and $H^s(\R^{3N})$ for any $s\!>\!0$.
Moreover, $\ran{T}$ is dense and  $\dom{\dot{\mathcal{H}}_0}\!=\!\ker{T}$.

\n Finally, the operator $\mathcal{G}(z)$ is represented by $\define*{\adj{(T\resolvent{\mathcal{H}_0}[\conjugate{z}])}\!;G(z)}\in\bounded{\hilbert*_N,\hilbert*_{N+1}}$.

\n
Morally, our efforts are devoted to the construction of the invertible operator $\Gamma(z)$ that encodes the  ultraviolet regularization and, together with $T$, fully characterizes the lower semi-bounded Hamiltonian $\mathcal{H}$.

\vs\vs

\subsection{Properties of the potential}\label{computingPotential}
\n

\vs
In the following, we show that the operator $G(-\lambda )$ can be identified with the potential $G^\lambda$, for $\lambda>0$, whose definition has been provided in~\eqref{potentialDef} and then we discuss some of its properties.

\n
Denote by $\xi_i\-\in\-\X_i$ the ``charge'' associated to $\pi_i$ and let $G_i(z)$ be the potential generated by the $i$-th charge, i.e. $G_i(z)\+\xi_i\-=\adj{(T_i\resolvent{\mathcal{H}_0}[\conjugate{z}])}\xi_i\+$.
Since the particles interacting with the impurity are all indistinguishable with each other, all the charges must be equal, namely $\xi_i\-=\mathcal{C}_i\+\xi$, denoting with $\mathcal{C}_i$ the isomorphism sending $\xi\-\in\-\hilbert*_N$ to $\X_i\+$.
In other words, given $\xi\-\in\-\hilbert*_N$ we have that
\begin{equation}\label{potential}
    G(z)\+\xi=\adj{(T\+\resolvent{\mathcal{H}_0}[\conjugate{z}])}\xi=\textstyle{\sum_{i=1}^N} G_i(z)\+\mathcal{C}_i\+\xi.
\end{equation}
Next, we set $\define*{G_i(-z);\maps{G^{z}_i}{\X_i;\hilbert*_{N+1}}}$ and we claim that the continuous operator $G(-\lambda)$
coincides with the definition of the potential $G^\lambda$ given by~\eqref{potentialDef} for any $\lambda>0$.
Let us prove this assertion.
For any $\xi\in\hilbert*_N\+, \psi\in\hilbert*_{N+1}\+$, by definition,
\begin{equation*}
    \scalar{G(-\lambda)\+ \xi}{\psi}[\hilbert*_{N+1}]\-=\scalar{\xi}{T\+\resolvent{\mathcal{H}_0}[-\lambda]\psi}[\hilbert*_N]
    =\scalar{\FT{\xi}}{\widehat{T\+\resolvent{\mathcal{H}_0}[-\lambda]\psi}}[\hilbert*_N].
\end{equation*}
From~\eqref{traceOperator}, for any $\phi\in H^2(\R^{3(N+1)})$, one obtains
\begin{equation}\label{traceFourier}
    \widehat{T_i\+\phi}\+(\vect{P}\-,\vect{k}_1,\ldots\check{\vect{k}}_i\ldots,\vect{k}_N)=\frac{1}{\mu}\+\sqrt{\frac{2}{\pi}\+}\!\-\integrate[\R^3]{\FT{\phi}\+(\vect{Q},\vect{k}_1,\ldots,\vect{k}_{i-1},\vect{P}\!-\vect{Q},\vect{k}_{i+1},\ldots,\vect{k}_N);\-d\vect{Q}}.
\end{equation}
Therefore, in our case,
\begin{align*}
    \scalar{\FT{\xi}}{\widehat{T\+\resolvent{\mathcal{H}_0}[-\lambda]\psi}}[\hilbert*_N]\!=\frac{1}{\mu}\+\sqrt{\frac{2}{\pi}\+}\sum_{i=1}^N\-\integrate[\R^{3N}]{\conjugate*{\FT{\xi}(\vect{P},\vect{k}_1,\ldots\check{\vect{k}}_i\ldots,\vect{k}_N)};\mspace{-11mu}d\vect{P}d\vect{k}_1\cdots d\check{\vect{k}}_i\cdots d\vect{k}_N}\,\times\\
    \times\!\- \integrate[\R^3]{\frac{\FT{\psi}(\vect{Q},\vect{k}_1,\ldots,\vect{k}_{i-1},\vect{P}\!-\vect{Q},\vect{k}_{i+1},\ldots,\vect{k}_N)}{\frac{1}{M}Q^2+\abs{\vect{P}\!-\vect{Q}}^2+\sum_{j\neq i} k_j^2+\lambda}; \-d\vect{Q}}\,.
\end{align*}
Finally, adopting the substitution $\vect{P}\longmapsto\vect{Q}+\vect{k}_i\+$,
\begin{align*}
\scalar{\widehat{T\+\resolvent{\mathcal{H}_0}[-\lambda]\psi}}{\FT{\xi}}[\hilbert*_N]\!=\frac{1}{\mu}\+\sqrt{\frac{2}{\pi}\+}\!\-\integrate[\R^{3(N+1)}]{\conjugate*{\FT{\psi}(\vect{Q},\vect{k}_1,\ldots,\vect{k}_N)};\mspace{-33mu}d\vect{Q}d\vect{k}_1\cdots d\vect{k}_N}\,\times\\
\times\-\sum_{i=1}^N \frac{\FT{\xi}(\vect{Q}+\vect{k}_i,\vect{k}_1,\ldots\check{\vect{k}}_i\ldots,\vect{k}_N)\,}{\frac{1}{M}Q^2+\sum_{j=1}^N k_j^2+\lambda}\+,
\end{align*}
from which the following expression for $G(-\lambda)$ in the space of momenta comes out:
\begin{equation}\label{potentialFourier}
    (\+\widehat{G(-\lambda)\+ \xi}\+)(\vect{p},\vect{k}_1,\ldots,\vect{k}_N)=\frac{1}{\mu}\+\sqrt{\frac{2}{\pi}\+}\-\sum_{i=1}^N\frac{\FT{\xi}(\vect{p}+\vect{k}_i,\vect{k}_1,\ldots\check{\vect{k}}_i\ldots,\vect{k}_N)}{\frac{1}{M}p^2+\sum_{j=1}^N k_j^2+\lambda}\+.
\end{equation}
Notice that definition~\eqref{potentialDef} is recovered.
\begin{note}\label{potentialPropertiesInheritance}
    We stress that, since $G^\lambda$ has been shown to be equal to $G(-\lambda)$, its properties are inherited, for instance, $\ran{G^\lambda}\subset\hilbert*_{N+1}$ and $\ker{G^\lambda}=\{0\}$.
\end{note}
\begin{note}\label{wellPosednessQFDomain}
    Recall, from remark~\ref{uniqueDecomposition}, that $\ran{G^\lambda}$ does not share non-trivial elements with $H^2(\R^{3(N+1)})$.
    Moreover, from~\eqref{potentialFourier}, one can verify that $\ran{G^\lambda}\cap H^1(\R^{3(N+1)})\!=\!\{0\}$ as well.
    This fact is remarkable in defining the quadratic form $Q$ in~\eqref{QF}.
\end{note}

\n
Next, our goal is to extract the asymptotic behaviour of the potential in a neighborhood of the coincidence hyperplanes, in the position representation.
We compute such asymptotic behavior for a regular charge $\xi\in\mathcal{S}(\R^{3N})\cap\X_j$.
From~\eqref{potentialFourier}, we get
\begin{align*}
    \left(G_j^\lambda \xi\right)\!\-(\vect{x}_0,\vect{x}_1,\ldots,\vect{x}_N)\!=\frac{4\pi}{\mu}\!\-\integrate[\R^{3(N+1)}]{\frac{e^{i\vect{q}\,\cdot\,\vect{x}_0\++\+i\!\!\sum\limits_{m=1}^N\!\vect{k}_m\cdot\,\vect{x}_m}\!\!\-}{(2\pi)^{\frac{3}{2}(N+2)}}\;\frac{\FT{\xi}(\vect{q}\-+\-\vect{k}_j,\vect{k}_1,\ldots\check{\vect{k}}_j\ldots,\vect{k}_N)}{\frac{1}{M}q^2+\sum_{m=1}^N k_m^2+\lambda};\mspace{-33mu}d\vect{q}d\vect{k}_1\-\cdots d\vect{k}_N}\\
    =\frac{4\pi}{\mu}\!\-\integrate[\R^{3N}]{;\mspace{-11mu}d\vect{p}\+d\vect{k}_1\cdots d\check{\vect{k}}_j\cdots d\vect{k}_N}\frac{e^{i\vect{p}\+ \cdot\left(\frac{\vect{x}_j+M\vect{x}_0}{M+1}\right)+\+ i\mspace{-6mu}\sum\limits_{m\neq j}\!\vect{k}_m\cdot\,\vect{x}_m}\!\!\!}{(2\pi)^{\frac{3}{2}(N+2)}}\;\FT{\xi}(\vect{p},\vect{k}_1,\ldots\check{\vect{k}}_j\ldots,\vect{k}_N)\+\times\\
    \times\+\frac{M^{\frac{3}{2}}}{(M\!+\-1)^3\!\-}\-\integrate[\R^3]{\frac{e^{i\frac{\sqrt{M}}{M+1}\+\vect{\kappa}\,\cdot\,(\vect{x}_0-\vect{x}_j)}}{\frac{\kappa^2}{M+1}+\frac{p^2}{M+1}+\sum_{m\neq j} k_m^2+\lambda};\-d\vect{\kappa}}\+,
\end{align*}
where a change of variables of Jacobian $\left(\tfrac{\sqrt{M}}{M+1}\-\right)^{\!3}$ has occurred in the last step, where
\begin{equation}
    \begin{cases}
    \vect{p}=\vect{q}+\vect{k}_j\+,\\
    \vect{\kappa}=\tfrac{1}{\sqrt{M}}\,\vect{q}-\sqrt{M}\,\vect{k}_j
    \end{cases}\mspace{-18mu}\iff \begin{cases}
    \vect{q}=\frac{M}{M+1}\!\left(\vect{p}+\tfrac{1}{\sqrt{M}}\,\vect{\kappa}\right)\!,\\
    \vect{k}_j=\frac{1}{M+1}\!\left(\vect{p}-\sqrt{M}\,\vect{\kappa}\right)\!.
    \end{cases}
\end{equation}
The last integral in $d\vect{\kappa}$ is well known, since, given $a>0$, one has
\begin{equation}\label{fourierTransformYukawa}
\integrate[\R^3]{\frac{e^{i\,\vect{k}\,\cdot\, \vect{x}}}{k^2+a^2};\-d\vect{k} }=\frac{2\pi^2}{\abs{\vect{x}}}\,e^{-a\,\abs{\vect{x}}}\:,\quad \forall \vect{x}\neq \vect{0}\+.
\end{equation}
Hence,
\begin{equation}\label{untilNowRigourous}
    \begin{split}
    \left(G_j^\lambda \xi\right)\!(\vect{x}_0,\vect{x}_1,\ldots,\vect{x}_N)\!=\!\-\integrate[\R^{3(N-1)}]{\frac{e^{i\!\sum\limits_{m\neq j}\!\vect{k}_m\cdot\,\vect{x}_m}\!\!}{(2\pi)^{\frac{3}{2}(N-1)}\!\!}; \mspace{-33mu}d\vect{k}_1\-\cdots \check{d\vect{k}_j}\cdots d\vect{k}_N}\integrate[\R^3]{\frac{e^{i\vect{p}\+ \cdot\left(\frac{\vect{x}_j+M\vect{x}_0}{M+1}\right)}}{(2\pi)^{\frac{3}{2}}};\-d\vect{p}}\times\\
    \times\+\frac{\FT{\xi}(\vect{p},\vect{k}_1,\ldots\check{\vect{k}}_j\ldots,\vect{k}_N)}{\abs{\vect{x}_0-\vect{x}_j}}\,e^{-\sqrt{\mu\,}\:\abs{\vect{x}_0-\vect{x}_j}\,\sqrt{\frac{p^2}{M+1}+\sum_{m\neq j} k_m^2+\lambda\,}}.
    \end{split}
\end{equation}
From the last equation, notice that the term $G^\lambda_j\xi$ is regular in $\R^{3(N+1)}\smallsetminus\pi_j\+$.
Furthermore, since we are working with $\FT{\xi}\-\in\-\mathcal{S}(\R^{3N})$, with a Taylor expansion of the exponential, we can easily expand in terms of powers of $\abs{\vect{x}_j-\vect{x}_0}$:
\begin{align*}
    \left(G_j^\lambda \xi\right)\!(\vect{x}_0,\vect{x}_1,\ldots,\vect{x}_N)\!=&\,\frac{\xi\!\left(\frac{\vect{x}_j+M\vect{x}_0}{M+1},\vect{x}_1,\ldots\check{\vect{x}}_j\ldots,\vect{x}_N\right)}{\abs{\vect{x}_0-\vect{x}_j}}\,+\\
    &-\!\sqrt{\mu\,}\!\-\integrate[\R^{3(N-1)}]{\frac{e^{i\!\!\sum\limits_{m\neq j} \!\vect{k}_m\cdot\,\vect{x}_m}}{(2\pi)^{\- \frac{3}{2}(N-1)}}; \mspace{-33mu}d\vect{k}_1\-\cdots \check{d\vect{k}_j}\cdots d\vect{k}_N} \!\-\integrate[\R^3]{ \frac{e^{i\+\vect{p}\+ \cdot\left(\frac{\vect{x}_j+M\vect{x}_0}{M+1}\right)}}{(2\pi)^{\frac{3}{2}}}; \-d\vect{p}}\+\times\\
    &\mspace{68mu}\times\!\textstyle{\sqrt{\frac{p^2}{M+1}+\sum_{m\neq j} k_m^2+\lambda\,}\,}\FT{\xi}(\vect{p},\vect{k}_1,\ldots\check{\vect{k}}_j\ldots,\vect{k}_N)\,+\\
    &+\oBig{\abs{\vect{x}_0-\vect{x}_j}}.
\end{align*}
Therefore, one obtains an explicit behavior of the potential near $\pi_j$ \begin{equation}\label{singularPotentialExtracted1}
    \left(G_j^\lambda \xi\right)\!(\vect{x}_0,\vect{x}_1,\ldots,\vect{x}_N)\!=\frac{\xi\!\left(\tfrac{\vect{x}_j+M\vect{x}_0}{M+1},\vect{x}_1,\ldots\check{\vect{x}}_j\ldots,\vect{x}_N\right)}{\abs{\vect{x}_0-\vect{x}_j}}\,-\Gamma_{\!\mathrm{diag}}^{j,\+\lambda}\+\xi+\oSmall{1},
\end{equation}
where
\begin{equation}\label{diagGamma}
    \begin{split}
    (\Gamma_{\!\mathrm{diag}}^{j,\+\lambda}\+\xi)(\vect{x}_0,\vect{x}_1,\ldots\check{\vect{x}}_j\ldots,\vect{x}_N)\!=\!\sqrt{\mu\,}\-\!\integrate[\R^{3N}]{\frac{e^{i\+\vect{p}\+\cdot\+\vect{x}_0\++\,i\!\!\sum\limits_{m\neq j}\!\vect{k}_m\cdot\,\vect{x}_m}\!\!\!}{(2\pi)^{\frac{3}{2}N}};\mspace{-10mu}d\vect{p}\+d\vect{k}_1\-\cdots d\check{\vect{k}}_j\cdots d\vect{k}_N}\+\times\-\\
    \times\textstyle{\sqrt{\frac{p^2}{M+1}+\sum_{m\neq j} k_m^2+\lambda\,}\,}\FT{\xi}(\vect{p},\vect{k}_1,\ldots\check{\vect{k}}_j\ldots,\vect{k}_N)\+.
\end{split}
\end{equation}
A similar asymptotic expansion holds for $G^\lambda$ in a neighborhood of $\pi_j$
\begin{equation}\label{singularPotentialExtracted2}
    (G^\lambda\xi)(\vect{x}_0,\vect{x}_1,\ldots,\vect{x}_N)\!=\frac{\xi\!\left(\tfrac{\vect{x}_j+M\vect{x}_0}{M+1},\vect{x}_1,\ldots\check{\vect{x}}_j\ldots,\vect{x}_N\right)\!}{\abs{\vect{x}_0-\vect{x}_j}}\+-\+(\Gamma_{\!\mathrm{diag}}^{j,\+\lambda}\!+\Gamma_{\!\mathrm{off}}^{j,\+\lambda})\+\xi+\+\oSmall{1}\-,
\end{equation}
with $\Gamma_{\!\mathrm{off}}^{j,\+\lambda}$ representing the contribution of all other potentials $\{G^\lambda_i\}_{i\neq j}\+$ evaluated on $\pi_j\+$, i.e.
\begin{equation}\label{offGamma}
    \begin{split}
    (\Gamma_{\!\mathrm{off}}^{j,\+\lambda}\+\xi)(\vect{x}_0,\vect{x}_1,\ldots\check{\vect{x}}_j\ldots,\vect{x}_N)\!=\--\tfrac{1}{2\pi^2\+\mu\,}\!\integrate[\R^3]{;\-d\vect{p}}\mspace{-9mu}\integrate[\R^{3N}]{\frac{e^{i(\vect{p}+\vect{k}_j)\+ \cdot\+ \vect{x}_0\+ +\+ i\!\sum\limits_{\ell\neq j}\! \vect{k}_\ell\+ \cdot\+ \vect{x}_\ell}\!\!\-}{(2\pi)^{\frac{3}{2}N}}\,;\mspace{-10mu}d\vect{k}_1\cdots d\vect{k}_N}\times\\
    \times\!\sum_{\substack{m=1\\m\+\neq\+ j}}^N \frac{\FT{\xi}(\vect{p}+\vect{k}_m,\vect{k}_1,\ldots\check{\vect{k}}_m\ldots,\vect{k}_N)}{\frac{1}{M}p^2+\!\sum\limits_{\ell=1}^N\! k_\ell^2+\lambda}\+.
\end{split}
\end{equation}
\begin{note}\label{domainDecomposition}
Notice that, in light of~\eqref{singularPotentialExtracted2}, we have obtained precisely the same singular behavior around $\pi_j$ for both $G^\lambda\xi$ and $\psi\-\in\-\hilbert*_{N+1}\+\cap H^2(\R^{3(N+1)}\setminus\pi)$ satisfying~\eqref{mfBC}.
Therefore, this suggests to write the vector $\psi$  satisfying the boundary condition~\eqref{mfBC} as a sum of two terms, namely
\begin{equation}
    \psi\-=\-G^\lambda\xi+w^\lambda,
\end{equation}
where, 
taking into account definition~\eqref{regGamma}, $w^\lambda\-\in\-\hilbert*_{N+1}\cap H^2(\R^{3(N+1)})$ fulfils \begin{equation}\label{regularComponentTraced}
w^\lambda|_{\pi_j}\!=(\Gamma_{\!\mathrm{diag}}^{j,\+\lambda}\!+\Gamma_{\!\mathrm{off}}^{j,\+\lambda}\-+\Gamma_{\!\mathrm{reg}}^j)\+\xi\+,\qquad\forall j\in\{1,\ldots,N\}.
\end{equation}
\end{note}

\n
Now, let us define the operator $\maps{\Gamma_{\!\lambda}}{\dom{\Gamma_{\!\lambda}} ;\hilbert*_N}\+$, given by
\begin{equation}
    \Gamma_{\!\lambda}:\xi\longrightarrow \sum_{j=1}^N T_j\!\left[\frac{\xi\!\left(\-\frac{\vect{x}_j+M\vect{x}_0}{M+1}\-,\vect{x}_1,\ldots\check{\vect{x}}_j\ldots,\vect{x}_N\!\right)\!}{\abs{\vect{x}_j-\vect{x}_0}}\--G^\lambda\xi(\vect{x}_0,\vect{x}_1,\ldots,\vect{x}_N)\-\right]\!+\tfrac{4\pi}{\mu}\sum_{j=1}^N\+\Gamma_{\!\mathrm{reg}}^j\+\xi
\end{equation}
Therefore, according to~\eqref{singularPotentialExtracted2}, one has
\begin{equation}
    \Gamma_{\!\lambda}\+\xi=\tfrac{4\pi}{\mu}\sum_{j=1}^N (\Gamma_{\!\mathrm{diag}}^{j,\+\lambda}\!+\Gamma_{\!\mathrm{off}}^{j,\+\lambda}+\Gamma_{\!\mathrm{reg}}^j)\+\xi\+.\label{theoreticalGammaConnection}
\end{equation}
By proposition~\ref{potentialProp}, $G^{\lambda_1}\--G^{\lambda_2}\in\dom{\mathcal{H}_0}$ for all $\lambda_1,\lambda_2>0$, hence
\begin{equation}
\Gamma_{\!\lambda_1}\--\Gamma_{\!\lambda_2}=T\+(G^{\lambda_2}-G^{\lambda_1}).
\end{equation}
Recalling the definition $\adj{G(z)}\!=T\+\resolvent{\mathcal{H}_0}[\conjugate{z}]$ valid for all $z\-\in\-\rho(\mathcal{H}_0)$, one can use the first resolvent identity to obtain for all $w\in\rho(\mathcal{H}_0)$
\begin{subequations}
\begin{gather}
    \adj{G(\conjugate{z})}\!-\adj{G(\conjugate{w})}\!=(z-w)\+\adj{G(\conjugate{w})}\resolvent{\mathcal{H}_0}[z]=(z-w)\+\adj{G(\conjugate{z})}\resolvent{\mathcal{H}_0}[w],\\
    G(z)-G(w)=(z-w)\resolvent{\mathcal{H}_0}[z]G(w)=(z-w)\resolvent{\mathcal{H}_0}[w]G(z)
\end{gather}
\end{subequations}
from which
\begin{equation}
    T\+[G(z)-G(w)]=(z-w)\+\adj{G(\conjugate{z})}G(w)=(z-w)\+\adj{G(\conjugate{w})}G(z).
\end{equation}
This means that we can extend the definition of $\define*{\Gamma_{\!\lambda};\Gamma(-\lambda)}$ to all $- \rho(\mathcal{H}_0)$, provided that
\begin{equation}\label{firstResolventIdGamma}
    \define{\Gamma(z)\-;\Gamma_{\!\lambda}-(\lambda+z)\+\adj{G(\conjugate{z})}G^\lambda}=\Gamma_{\!\lambda}-(\lambda+z)\+\adj{G^\lambda}G(z).
\end{equation}
Notice that~\eqref{firstResolventIdGamma} corresponds to condition~\eqref{firstResolventGammaTheory} of section~\ref{kreinRecall}, since for all $z,w\-\in\-\rho(\mathcal{H}_0)$
\begin{align*}
    \Gamma(z)-\Gamma(w)=&\,\Gamma_{\!\lambda}-(\lambda+z)\+\adj{G(\conjugate{z})}G^\lambda-\Gamma_{\!\lambda}+(\lambda+w)\+\adj{G^\lambda}G(w)\\
    =&-(\lambda+z)\+\adj{G(\conjugate{z})}\-\big[G(w)-(\lambda+w)\resolvent{\mathcal{H}_0}[-\lambda]G(w)\big]+\\
    &+(\lambda+w)\-\big[\adj{G(\conjugate{z})}-(\lambda+z)\+\adj{G(\conjugate{z})}\resolvent{\mathcal{H}_0}[-\lambda]\big]G(w)\\
    =&\,(w-z)\+\adj{G(\conjugate{z})}G(w).
\end{align*}
In section~\ref{closure&Boundedness} we have characterized  the operator $\Gamma(z)$.

\vs\vs

\subsection{Heuristic derivation of the quadratic form}\label{actuallyBuildingQF}
\n

\vs
Here we provide a heuristic discussion meant to justify the definition of the quadratic form $Q$ given in~\eqref{QF}.  Our aim is to construct the energy form $\scalar{\psi}{\hat{\mathcal{H}}\psi}$ associated to the formal Hamiltonian $\hat{\mathcal{H}}$ discussed in the introduction for a given vector $\psi\-\in\-\hilbert*_{N+1}\cap H^2(\R^{3(N+1)}\setminus\pi)$ that fulfils the boundary condition~\eqref{mfBC}.
Recalling that the Hamiltonian acts as $\mathcal{H}_0$ outside $\pi$, given $\epsilon\->\- 0$, let $\define{D_\epsilon\!;\-\big\{\!(\vect{x}_0,\vect{x}_1,\ldots,\vect{x}_N)\!\in\R^{3(N+1)}\,\big | \min\limits_{1\leq\+i\,\leq N}\abs{\vect{x}_i-\vect{x}_0}\!>\epsilon\big\}}$.
Then
\begin{equation}\label{freeLimitQF}
    \scalar{\psi}{\hat{\mathcal{H}}\psi}=\lim_{\epsilon\to 0}\- \integrate[D_\epsilon]{\conjugate*{\psi(\vect{x}_0,\vect{x}_1,\ldots, \vect{x}_N)}\+(\mathcal{H}_0\+\psi)(\vect{x}_0,\vect{x}_1,\ldots, \vect{x}_N); \-d\vect{x}_0 d\vect{x}_1\cdots d\vect{x}_N}.
\end{equation}
Introducing the decomposition of remark~\ref{domainDecomposition}, $\psi=w^\lambda\!+G^\lambda\xi\+$, equation~\eqref{freeLimitQF} reads
\begin{equation}
    \scalar{\psi}{\hat{\mathcal{H}}\psi}=\scalar{w^\lambda}{(\mathcal{H}_0+\lambda)w^\lambda}-\lambda\norm{\psi}^2\!+\scalar{G^\lambda\xi}{(\mathcal{H}_0+\lambda)w^\lambda},
\end{equation}
since $\dom{\mathcal{H}_0}\cap\ran{G^\lambda}=\{0\}$.
The last term can be simplified using~\eqref{potential}
\begin{align*}
    \scalar{\psi}{\hat{\mathcal{H}}\psi}=\scalar{w^\lambda}{(\mathcal{H}_0+\lambda)w^\lambda}-\lambda\norm{\psi}^2\!+\scalar{\xi}{T\+w^\lambda}[\hilbert*_N]\\[-2.5pt]
    =\mathscr{F}_\lambda[w^\lambda]-\lambda\norm{\psi}^2\!+\tfrac{4\pi}{\mu}\scalar{\xi}{\textstyle{\sum_{i=1}^N}w^\lambda|_{\pi_i}}[\hilbert*_N]\+.
\end{align*}
In the last step, the first contribution has been rewritten using definition~\eqref{trivialQF}, while replacing $w^\lambda|_{\pi_i}$ via~\eqref{regularComponentTraced} in the last expression, one gets
\begin{equation}\label{inProgressQF1}
    \scalar{\psi}{\hat{\mathcal{H}}\psi}=\mathscr{F}[w^\lambda]-\lambda\norm{\psi}^2\!+\tfrac{4\pi}{\mu}\scalar{\xi}{\textstyle{\sum_{i=1}^N}(\Gamma_{\!\mathrm{diag}}^{i,\+\lambda}\-+\Gamma_{\!\mathrm{off}}^{i,\+\lambda}\-+\Gamma_{\!\mathrm{reg}}^i)\+\xi}[\hilbert*_N]\+.
\end{equation}
Exploiting the symmetry, one obtains
\begin{equation}\label{inProgressQF2}
    \scalar{\psi}{\hat{\mathcal{H}}\psi}=\mathscr{F}[w^\lambda]-\lambda\norm{\psi}^2\!+\tfrac{4\pi N}{\mu}\scalar{\xi}{(\Gamma_{\!\mathrm{diag}}^{N,\+\lambda}\-+\Gamma_{\!\mathrm{off}}^{N,\+\lambda}\-+\Gamma_{\!\mathrm{reg}}^N)\+\xi}[\hilbert*_N]\+.
\end{equation}
Now, we want to show that the last term is equal to $\Phi^\lambda[\xi]$, given by~\eqref{defPhi}.
To this end, we consider separately each component of the inner product in~\eqref{inProgressQF2}.
Firstly, according to~\eqref{diagGamma}, we change the order of integration, obtaining
\begin{align*}
    \tfrac{4\pi N}{\mu}\scalar{\xi}{\Gamma_{\!\mathrm{diag}}^{N,\+\lambda}\xi}[\hilbert*_N]\-&=\tfrac{4\pi N}{\!\-\sqrt{\mu\,}\+}\!\-\integrate[\R^{3N}]{\!\textstyle{\sqrt{\frac{p^2}{M+1}+\-\sum_{m=1}^{N-1} k_m^2\-+\-\lambda\+}\,}\abs{\FT{\xi}(\vect{p},\vect{k}_1,\ldots,\vect{k}_{N-1})}^2;\mspace{-10mu}d\vect{p}\+d\vect{k}_1\-\cdots d\vect{k}_{N-1}}\\
    &=\Phi^\lambda_{\mathrm{diag}}[\xi].
\end{align*}
Indeed, notice that definition~\eqref{diagPhi} is recovered. Similarly, taking into account~\eqref{offGamma},
\begin{equation*}
\begin{split}
    \tfrac{4\pi N}{\mu}\scalar{\xi}{\Gamma_{\!\mathrm{off}}^{N,\+\lambda}\xi}[\hilbert*_N]=-\tfrac{2N}{\pi\,\mu^2\!}\!\-\integrate[\R^{3(N+1)}]{\conjugate*{\FT{\xi}(\vect{p}+\vect{k}_N,\vect{k}_1,\ldots,\vect{k}_{N-1})}\+\times\\[-5pt]
    \times\-\sum_{m=1}^{N-1}\frac{\FT{\xi}(\vect{p}+\vect{k}_m,\vect{k}_1,\ldots\check{\vect{k}}_m\ldots,\vect{k}_N)}{\frac{1}{M}p^2+\sum_{\ell=1}^N k_\ell^2+\lambda};\mspace{-33mu}d\vect{p}\+ d\vect{k}_1\-\cdots d\vect{k}_N}\+.
\end{split}
\end{equation*}
Exchanging the role of $\vect{k}_1$ and $\vect{k}_N$, since $\xi\in\hilbert*_N$, one gets
\begin{equation*}
\begin{split}
    \tfrac{4\pi N}{\mu}\scalar{\xi}{\Gamma_{\!\mathrm{off}}^{N,\+\lambda}\xi}[\hilbert*_N]=-\tfrac{2N}{\pi\,\mu^2\!}\-\sum_{m=2}^N\integrate[\R^{3(N+1)}]{\conjugate*{\FT{\xi}(\vect{p}+\vect{k}_1,\vect{k}_2,\ldots,\vect{k}_N)}\+\times\\[-7.5pt]
    \times\+\frac{\FT{\xi}(\vect{p}+\vect{k}_m,\vect{k}_1,\ldots\check{\vect{k}}_m\ldots,\vect{k}_N)}{\frac{1}{M}p^2+\sum_{\ell=1}^N k_\ell^2+\lambda};\mspace{-33mu}d\vect{p}\+ d\vect{k}_1\-\cdots d\vect{k}_N}\+.
\end{split}
\end{equation*}
In the last integral, we can change the variables setting $\vect{\kappa}_{\sigma_m(i)}\-=\vect{k}_i\+ $, for all $i\in\{1,\ldots,N\}$, where the permutation of $N$ elements $\sigma_m$ is given by
\begin{equation}\label{permutationSigma}
    \define{\sigma_m;\begin{pmatrix}1 & 2 & 3 & \ldots & m-1 & m & m+1 & \ldots & N\\ 1 & 3 & 4 & \ldots & m & 2 & m+1 &\ldots &N\end{pmatrix}}\!,\quad\text{for }\+m>2,
\end{equation}
while, $\sigma_2$ is the identity permutation. Applying this change of variables, one obtains
\begin{equation*}
\begin{split}
    \tfrac{4\pi N}{\mu}\scalar{\xi}{\Gamma_{\!\mathrm{off}}^{N,\+\lambda}\xi}[\hilbert*_N]=-\tfrac{2N}{\pi\,\mu^2\!}\-\sum_{m=2}^N\integrate[\R^{3(N+1)}]{\conjugate*{\FT{\xi}(\vect{q}\-+\vect{\kappa}_1,\vect{\kappa}_{\sigma_m(2)},\ldots,\vect{\kappa}_{\sigma_m(N)})}\+\times\\[-7.5pt]
    \times\+\frac{\FT{\xi}(\vect{q}\-+\vect{\kappa}_2,\vect{\kappa}_1,\vect{\kappa}_3\ldots,\vect{\kappa}_N)}{\frac{1}{M}q^2+\sum_{\ell=1}^N \kappa_\ell^2+\lambda};\mspace{-33mu}d\vect{q} d\vect{\kappa}_1\-\cdots d\vect{\kappa}_N}\+.
\end{split}
\end{equation*}
Notice that the symmetry properties of $\xi\in\hilbert*_N$ make the integrand actually independent of $m$.
Therefore, the expression in definition~\eqref{offPhi} is achieved:
\begin{equation*}
\begin{split}
    \tfrac{4\pi N\-}{\mu}\scalar{\xi}{\!\Gamma_{\!\mathrm{off}}^{N,\+\lambda}\xi}[\-\hilbert*_N]\!\-&=\!-\tfrac{2N(N-1)\!}{\pi\,\mu^2\!}\!\!\-\integrate[\R^{3(N+1)}]{\conjugate*{\-\FT{\xi}(\vect{p}\-+\mspace{-2.25mu}\vect{\kappa}_1,\vect{\kappa}_2,\ldots,\vect{\kappa}_N\-)\-}\,\frac{\FT{\xi}(\vect{p}\-+\mspace{-2.25mu}\vect{\kappa}_2,\vect{\kappa}_1,\vect{\kappa}_3,\ldots,\vect{\kappa}_N\-)\!}{\frac{1}{M}p^2+\sum_{\ell=1}^N \kappa_\ell^2+\lambda};\mspace{-36mu}d\vect{p}\mspace{0.75mu} d\vect{\kappa}_1\-\cdots d\vect{\kappa}_N}\\[-5pt]
    &=\Phi^\lambda_{\mathrm{off}}[\xi].
\end{split}
\end{equation*}
Finally, from~\eqref{regGamma} and~\eqref{runningCoupling},
\begin{align*}
    \tfrac{4\pi N}{\mu}\scalar{\xi}{\!\Gamma_{\!\mathrm{reg}}^N\+\xi}[\-\hilbert*_N]\!\-=\tfrac{4\pi N}{\mu}\!\-\integrate[\R^{3N}]{\!\!\left[\alpha+\gamma\-\sum_{m=1}^{N-1}\frac{\theta(\abs{\vect{x}_m\--\vect{x}_0})}{\abs{\vect{x}_m\--\vect{x}_0}}\right]\-\abs{\xi(\vect{x}_0,\vect{x}_1,\ldots,\vect{x}_{N-1})}^2;\mspace{-10mu}d\vect{x}_0d\vect{x}_1\-\cdots d\vect{x}_{N-1}}\\
    =\tfrac{4\pi N}{\mu}\!\-\integrate[\R^{3N}]{\!\-\left[\alpha+(N\--\-1)\,\gamma\,\frac{\theta(\abs{\vect{x}_1\--\vect{x}_0})-\-1}{\abs{\vect{x}_1\--\vect{x}_0}}\right]\-\abs{\xi(\vect{x}_0,\vect{x}_1,\ldots,\vect{x}_{N-1})}^2;\mspace{-10mu}d\vect{x}_0d\vect{x}_1\-\cdots d\vect{x}_{N-1}}+\\
    +\,\tfrac{4\pi\+N(N-1)\,\gamma}{\mu}\!\-\integrate[\R^{3N}]{\frac{\abs{\xi(\vect{x}_0,\vect{x}_1,\ldots,\vect{x}_{N-1})}^2}{\abs{\vect{x}_1\--\vect{x}_0}};\mspace{-10mu}d\vect{x}_0d\vect{x}_1\-\cdots d\vect{x}_{N-1}}\\
    =\Phi_0[\xi]+\Phi_{\mathrm{reg}}[\xi].
\end{align*}
Summing up, we have shown that, for a sufficiently regular charge $\xi\-\in\-\hilbert*_N$, equation~\eqref{inProgressQF2} reduces to the action of $Q$ defined in~\eqref{QF}, since by definition~\eqref{defPhi} $\Phi^\lambda\!=\-\Phi_0+\Phi^\lambda_{\mathrm{diag}}\-+\Phi^\lambda_{\mathrm{off}}+\Phi_{\mathrm{reg}}\+$.

\vs\vs

\vs\vs\vs\vs

\end{document}